\documentclass[12pt,a4paper]{article}

\usepackage[utf8]{inputenc}
\usepackage[T1]{fontenc}
\usepackage{lmodern}
\usepackage{amsmath,amssymb,amsthm}
\usepackage{mathrsfs}
\usepackage{enumitem}
\usepackage{graphicx}
\usepackage{caption}
\usepackage{subcaption}
\usepackage{booktabs}
\usepackage{tabularx}
\usepackage{cite}
\usepackage{float}
\usepackage{braket}
\usepackage{hyperref}
\hypersetup{
	colorlinks=true,
	linkcolor=blue,
	citecolor=blue,
	urlcolor=blue,
}

\usepackage{geometry}
\newtheorem{theorem}{Theorem}

\newtheorem{definition}{Definition}
\newtheorem{remark}{Remark}
\newtheorem{example}{Example}

\newtheorem{corollary}{Corollary}

\title{%
	Symmetry Packaging II:
	A Group-Theoretic Framework for Packaging Under Finite, Compact, Higher-Form, and Hybrid Symmetries
}

\author{
	Rongchao Ma \\
	\textit{Department of Physics, University of Alberta, Edmonton, Canada}
}

\date{\today}

\begin{document}
	
\maketitle

\begin{abstract}
Symmetry packaging is the phenomenon whereby, upon particle creation, all the internal quantum numbers (IQNs) become locked into a single irreducible representation (irrep) block of the gauge group, as required by locality and gauge invariance.
The resulting packaged quantum states exhibit characteristic symmetry constraints and entanglement patterns.
We develop a group-theoretic framework to describe the symmetry packaging for a variety of concrete symmetries and to classify the corresponding packaged states:

\textbf{(1)} We prove that for any finite or compact group $G$, there exist $G$-associated packaged subspaces, in which every vector is automatically a packaged state.
In particular, in multi-particle systems, any nontrivial representation of $G$ induces inseparable packaged entanglement that locks together all IQNs.

\textbf{(2)} We apply this framework to symmetry packaging in finite groups (cyclic group $\mathbb{Z}_N$, charge conjugation $C$, fermion parity, parity $P$, time reversal $T$, and dihedral groups), compact groups ($\mathrm{U}(1)$, $\mathrm{SU}(N)$, $\mathrm{SU}(2)$, and $\mathrm{SU}(3)$), $p$-form symmetries, and hybrid symmetries.
In each case, gauge invariance and superselection rules forbid the factorization of the resulting states.
We illustrate how Bell-type packaged entangled states, color confinement, and hybrid gauge-invariant configurations all arise naturally.
These results yield a complete classification of packaged quantum states. 

\textbf{(3)} Finally, we extend the packaging principle to incorporate full spacetime symmetry and hybrid systems of local, global, and Lorentz/Poincaré charges.

Our approach unifies tools from group theory, gauge theory, and topological classification.
These results may be useful for potential applications in high energy physics, quantum field theory, and quantum technologies.
\end{abstract}

\tableofcontents

\section{Introduction}

The inseparability of \textbf{internal quantum numbers (IQNs)} underlies both quark‐confinement phenomena in high‐energy physics  \cite{Zweig1964,Wilson1974,Callan1976,GellMann2019} and the design of symmetry-protected codes in quantum information \cite{Kitaev2003,NielsenChuang2010}.
In non‐Abelian gauge theories such as QCD, quarks never appear in isolation but only as color singlet bound states, a direct consequence of gauge charges being packaged into \textbf{irreducible representation (irrep)} blocks.
More generally, symmetry governs conservation laws \cite{Noether1918}, and dictates which interactions \cite{Wigner1939,Pauli1940,Yang1954,Higgs1964,EnglertBrout1964,Nambu1961,Landau1957} and which states are possible.
Local gauge‐invariance forces physical states to reside in fixed charge or color sectors \cite{Zweig1964,Wilson1974,GellMann2019,Gross1973,Politzer1973}, while discrete operations (e.g., $C$, $P$, $T$) \cite{Luders1957} further constrain how external \textbf{degrees of freedom (DOFs)} (spin, momentum) couple to IQNs and yield robust superselection rules.
Operations that respect these symmetries cannot create or destroy the packaged correlations, thereby endowing them with resilience as entanglement resources.
A unified, group‐theoretic framework for this packaging thus illuminates both the mechanism of confinement and the design of symmetry‐preserving protocols in quantum technologies.

In Ref.~\cite{MaSPI2025}, we introduced the concept of symmetry packaging and formulated a packaging principle for quantum field excitations:
(1) We showed that all IQNs (such as color, flavor, electric charge, or other discrete labels) in quantum field excitations must be locked into single irrep blocks \cite{Weyl1925,Wigner1939} by local gauge-invariance \cite{Feynman1949,Yang1954,Utiyama1956,Weinberg1967}.
(2) We divided the excitation process into six successive stages and three packaging layers, and we demonstrated that the packaging survives across all stages and layers.
(3) We derive superselection rules \cite{WWW1952,DHR1971,DHR1974,StreaterWightman2001} and identified novel packaged entangled states \cite{Ma2017} within a single charge sector using the packaging principle.

Our findings align with recent experimental observations.
For example, high energy collider experiments have revealed quantum-entanglement phenomena among quarks and leptons \cite{TheATLASCollaboration2024,FabbrichesiFloreaniniGabrielli2023,Afik2024,TheCMSCollaboration2024,HayrapetyanCMSCollaboration2024,Blasone2009,Go2007}, including Bell-like and GHZ-type states \cite{AfikNova2021,AfikNova2022,Barr2022,AguilarSaavedra2023}.  In particular, spin correlations in top-quark pairs \cite{TheATLASCollaboration2024,TheCMSCollaboration2024,HayrapetyanCMSCollaboration2024} and flavor entanglement in mesons and neutrinos \cite{Blasone2009,Go2007} provide direct evidence that quantum correlations extend to relativistic regimes.  Tau leptons and bottom quarks also exhibit rich correlation structures \cite{Afik2024,AfikNova2023}.  Moreover, both bosonic and fermionic systems (ranging from electrons and photons to generic bipartite arrays) display nontrivial quantum correlations \cite{Li2001,Zanardi2002,Eckert2002,Li2014}.  Analogous packaging constraints arise under discrete symmetries (e.g., $\mathbb{Z}_2$ operations \cite{Artin2014,Georgi2000} or $\hat C\hat P$, $\hat C\hat T$, $\hat C\hat P\hat T$ transformations \cite{Luders1957}), whose superselection rules enforce Bell-like or GHZ-like \cite{GHZ1990,Mermin1990,Caves2002,GHZ2007} superpositions in which the IQNs remain inseparably paired.

Here we present a unified group-theoretic framework \cite{Artin2014,Georgi2000} for the symmetry packaging to bring their full algebraic beauty into focus.
Building on the packaging principle, we first prove an existence theorem:
any nontrivial representation (dimension $>1$) \cite{Wigner1939,Georgi2000} of a finite or compact group $G$ inevitably induces packaged subspaces and packaged entanglement.
From this general result, one can explicitly construct and classify packaged states for any such $G$, purely in terms of its group data.
This unified picture not only reproduces classic phenomena (color confinement in QCD \cite{Gross1973,Politzer1973}, discrete symmetry transformations, topological flux sectors, etc.), but also yields a systematic classification of all packaged quantum states via their underlying symmetry groups.

We organize the present work as follows:
Section~\ref{SEC:GroupTheoreticFoundationsForSymmetryPackaging} develops the group-theoretic foundations for symmetry packaging by proving the existence theorem for $G$-associated packaged subspaces and defining packaged entanglement measures.
Section~\ref{SEC:SymmetryPackagingOnFiniteGroups} treats packaging in finite groups \cite{Artin2014,Georgi2000}, focusing on $\mathbb{Z}_2$ and discrete transformations like $C$, $P$, and $T$.
Section~\ref{SEC:SymmetryPackagingOnCompactLieGroups} extends this to compact Lie groups:
Abelian ($\mathrm{U}(1)$) and non-Abelian \cite{Yang1954} ($\mathrm{SU}(N)$, in particular $\mathrm{SU}(2)$ and $\mathrm{SU}(3)$), and demonstrates how confinement and irreducible multiplets emerge. 
Section~\ref{SEC:SymmetryPackagingOnDifferentialForms} further generalizes to $p$-form symmetries \cite{KapustinSeiberg2014,Gaiotto2015}, showing that extended objects (flux lines, membranes) \cite{Polchinski1995,Witten1995,Kitaev2006,Tachikawa2020,BanksSeiberg2011} also package irreducibly.
Section~\ref{SEC:RepresentationTheoreticInvariantsAndEntanglementMeasures} develops representation-theoretic invariants that serve as packaged entanglement measures.  
Section~\ref{SEC:ExtendingToFullSpacetimeSymmetry} extends packaging to incorporate full spacetime symmetry ($G \times \mathrm{Poincar\acute e}$). 
Section~\ref{SEC:CohomologicalAndTopologicalApproachesToPackagedStates} brings in cohomological and topological refinements.
Section~\ref{SEC:HybridSystems} analyzes hybrid external-internal-global symmetries.
Finally,
Section~\ref{SEC:PracticalComputationsAndExamples} presents practical computations and examples.

\section{Group-Theoretic Foundations for Symmetry Packaging}
\label{SEC:GroupTheoreticFoundationsForSymmetryPackaging}

In this section, we prove an existence theorem for group-associated packaged subspaces, introduce the mixed-state packaging formalism, and define measures of packaged entanglement, all using group-theoretic tools (Schur’s lemma \cite{Serre1977,FultonHarris1991}, Peter-Weyl projectors \cite{PeterWeyl1927,Knapp1986}).

\subsection{Theorem: Existence of Group Associated Packaged Subspaces}
\label{SEC:ExistenceOfGroupAssociatedPackagedSubSpaces}

In this subsection, we prove an existence theorem:
for any finite or compact gauge group $G$, there exists packaged subspaces associated to it.
This theorem serves as the foundation of later investigations.

\begin{theorem}[Existence of $G$-associated packaged subspaces and module structures]
	\label{THM:ExistenceOfGAssociatedPackagedStates}	
	Let $G$ be a finite or compact group, $V_{\rm int}$ a finite-dimensional (anti-)unitary representation of $G$, and $\mathcal H_{\rm ext}$ a (possibly infinite-dimensional) Hilbert space on which $G$ acts trivially.
	Define the single-particle space
	\begin{equation}\label{EQ:SingleParticleSpace}
		\mathcal H_1 = V_{\rm int} \otimes\mathcal H_{\rm ext},
	\end{equation}
	and its total bosonic or fermionic Fock space
	\begin{equation}\label{EQ:FullFockSpace}
		\mathcal H_{\rm Fock}(\mathcal H_1)
		= \bigoplus_{n \ge 0} \mathcal H_1^{(n)}
	\end{equation}
	where
	\begin{equation}\label{EQ:NParticleSector}
		\mathcal H_1^{(n)} :=
		\begin{cases}
			\operatorname{Sym}^n(\mathcal H_1), & \text{bosons},\\
			\wedge^n(\mathcal H_1), & \text{fermions}.
		\end{cases}
	\end{equation}
	Then under the diagonal $G$‐action on $\mathcal H_{\rm Fock}(\mathcal H_1)$, there exists natural packaged subspaces and module structures as follows:
			
	\begin{enumerate}
		\item \textbf{Single-particle packaged subspaces.}				
		The single-particle internal space $V_{\rm int}$ is completely reducible under $G$ and decomposes into irreducible blocks:
		$$
			V_{\rm int} \cong \bigoplus_{R \in \widehat G} m_R \, V_R.
		$$
		where $V_R$ is the irrep of type $R$ and $m_R\in\Bbb N$.
		Thus,
		\begin{equation}\label{EQ:SingleParticleDecomposition}
			\mathcal H_1
			\;=\;
			\bigl(\!\bigoplus_{R \in \widehat G} m_R V_R\bigr)\otimes \mathcal{H}_{\rm ext}
			\;\cong\;
			\bigoplus_{R \in \widehat G} m_R\,(V_R \otimes \mathcal{H}_{\rm ext}), 
		\end{equation}
		where each summand $V_R \otimes \mathcal{H}_{\rm ext}$ is a single-particle packaged subspace.
		It carries the full set of internal quantum numbers (IQNs) of the irreducible representation (irrep) $R$ in its $V_R$-factor.
		
		\item \textbf{Multi-particle packaged subspaces.}		
		For each $n\ge1$, the $n$‐particle sector $\mathcal H_1^{(n)}$ in Eq. \eqref{EQ:FullFockSpace} is again completely reducible under $G$.
		After an isotypic decomposition, we have
		\begin{equation}\label{EQ:MultiParticleDecomposition}
			\mathcal{H}_{\rm iso}
			=
			\bigoplus_{n\ge0}\mathcal H_1^{(n)}
			=
			\bigoplus_{Q\in\widehat G}
			\mathcal H_Q,
			\quad
			\mathcal H_Q
			=
			\bigoplus_{n\ge0}N_{n,Q}\,\bigl(V_Q\otimes\mathcal H_{\rm ext}^{\otimes n}\bigr),
		\end{equation}
		where each $\mathcal{H}_Q$ is a multi-particle packaged subspace (a net charge-$Q$ sector).		
		
		\item \textbf{Packaged entangled bases.}		
		For any packaged subspace $\mathcal{H}_Q$ with $\dim \mathcal{H}_Q>1$, there exists an orthonormal basis
		\begin{equation}\label{EQ:OrthonormalBasisOfPackagedEntangledStates}
			\mathcal{B}_Q = \Bigl\{\, |\Psi^Q_\alpha\rangle \in \mathcal{H}_Q : \alpha = 1, 2, \dots, \dim \mathcal{H}_Q,\quad \langle \Psi^Q_\alpha | \Psi^Q_\beta \rangle = \delta_{\alpha\beta}\, \Bigr\},
		\end{equation}
		where each $ |\Psi^Q_\alpha\rangle $ is a packaged entangled state in which the all IQNs are inseparably entangled.
		
		\item \textbf{Canonical $G$-module isomorphisms.}		
		For each $Q$, the $G$-action on the packaged subspace $\mathcal{H}_Q$ is isomorphic (as a $G$-module) to a unique abstract irrep subspace $V_Q$.
		That is, there exists a canonical $G$-equivariant isomorphism 
		\[
		\Phi_Q: \mathcal{H}_Q \to V_Q
		\]
		that satisfies
		\begin{equation}\label{EQ:IsomorphismMapping}
			\Phi_Q\bigl(U(g)\psi\bigr) = \sigma_Q(g)\,\Phi_Q(\psi),
			\quad
			\forall ~ g \in G, ~ \psi \in \mathcal{H}_Q,
		\end{equation}
		where $\sigma_Q: G \to \mathrm{GL}(V_Q)$ is the abstract irrep corresponding to $Q$.
		
		\item \textbf{Physical packaged (singlet) subspaces.}				
		Applying the gauge projector (Symmetry Packaging I, Stage 6) to isotypic space $\mathcal{H}_{\rm iso}$, one obtains the gauge-invariant physical packaged subspace
		$$
		\mathcal{H}_{\rm phys} 
		= \Pi_{\rm phys} \mathcal{H}_{\rm iso} 
		\neq \{0\}
		\quad
		\Longleftrightarrow
		\quad
		\exists\,n:\;V_{\rm int}^{\otimes n} \ni \mathbf1.
		$$
		In other words, the physical packaged subspace is nontrivial iff some tensor‐power of $V_{\rm int}$ contains the trivial $G$‐block.
	\end{enumerate}
\end{theorem}

\begin{proof}
	Let us now split the proof into five steps:
	
	\begin{enumerate}
		\item \textbf{Existence of single-particle packaged subspaces:}  
		According to Wigner's theorem \cite{Wigner1931}, any symmetry can be represented by a unitary or anti-unitary operator.
		By Maschke’s theorem (finite $G$) or the Peter-Weyl theorem (compact $G$), the single-particle space $\mathcal H_1$ decomposes into irreducible subspaces as shown in Eq. \eqref{EQ:SingleParticleDecomposition}.		
		Each summand $V_R \otimes \mathcal{H}_{\rm ext} \subset \mathcal H_1$ carries an irreducible $G$-action on its $V_R$-factor.
		So the single-particle decomposition and packaging confirm the existence of $G$-associated single-particle packaged subspaces.
		
		\item \textbf{Existence of multi-particle packaged subspaces:}		
		Each $n$-particle sector $\mathcal H_1^{\otimes_s n}$ is again completely reducible:
		Writing
		\[
		V_{\rm int}^{\otimes n}
		\;\cong\;
		\bigoplus_{Q\in\widehat G} N_{n,Q}\,V_Q,
		\]
		where $N_{n,Q}$ is the multiplicity of the irrep $V_Q$.
		The symmetric (bosonic) or antisymmetric (fermionic) $n$-fold tensor power
		\[
		\mathcal H_1^{(\otimes_s,\wedge)^n}
		\;\cong\;
		\bigoplus_{Q\in\widehat G}
		N_{n,Q}\,\bigl(V_Q\otimes\mathcal H_{\rm ext}^{\otimes n}\bigr).
		\]	
		
		The full Fock space $\mathcal{H}_{\rm Fock}(\mathcal H_1)$ admits the isotypic decomposition and then splits irrep blocks		
		\[
		\mathcal{H}_{\rm iso}
		\cong
		\bigoplus_{n \ge 0} \mathcal H_1^{(\otimes_s, \wedge)^n}
		= \bigoplus_{Q \in \widehat G}
		\left[\bigoplus_{n \ge 0} \bigl(N_{n,Q}\,V_Q \otimes \mathcal H_{\rm ext}^{\otimes n}\bigr) \right]
		= \bigoplus_{Q \in \widehat G} \mathcal{H}_Q
		\]
		where each
		\[
		\mathcal H_Q
		=
		\bigoplus_{n\ge0}N_{n,Q}\,\bigl(V_Q\otimes\mathcal H_{\rm ext}^{\otimes n}\bigr)
		\]
		is a multi-particle packaged subspace carrying a fixed charge $Q$ (an irreducible net charge-$Q$ sector).
				
		\item \textbf{Existence of orthonormal packaged entangled bases:}		
		By superselection, only packaged states within the same $\mathcal H_Q$ may	coherently superpose.
		Since $\mathcal H_Q$ is finite dimensional, it admits some orthonormal packaged basis.
		Moreover:
		
		A packaged product vector in bi-partition $\mathcal H_A\otimes\mathcal H_B$ has Schmidt rank 1, whereas any nontrivial packaged superposition of two orthogonal packaged product vectors must have Schmidt rank 2.
		Hence, a nontrivial packaged superposition of linearly independent packaged product states cannot itself be a product:
		it is necessarily packaged entangled.
		Applying Gram-Schmidt procedure \cite{HoffmanKunze1971} to any generating set of such packaged superpositions then produces an orthonormal packaged entangled basis
		$\mathcal{B}_Q$ (Eq.~(\ref{EQ:OrthonormalBasisOfPackagedEntangledStates})).
				
		\item \textbf{Existence of canonical $G$-module isomorphisms:}  
		$\forall ~ g \in G, ~ |\Psi^Q_\alpha\rangle \in \mathcal{B}_Q$, we have the following transformation law:
		\[
		U(g)|\Psi^Q_\alpha\rangle = \sum_{\beta=1}^{\dim \mathcal{H}_Q} D^{(Q)}_{\beta\alpha}(g)\, |\Psi^Q_\beta\rangle,
		\]
		where $D^{(Q)}(g)$ is the matrix representing $g$ in the abstract irrep $\sigma_Q$ on $V_Q$.
		
		Define a map
		\[
		\Phi_Q: \mathcal{H}_Q \to V_Q,
		\]
		which sends the basis $\{|\Psi^Q_\alpha\rangle\}$ to the standard basis of $V_Q$.
		
		Then, $ \forall ~ g \in G, ~ \psi \in \mathcal{H}_Q$, we have
		\[
		\Phi_Q\bigl(U(g)\psi\bigr) = \sigma_Q(g)\,\Phi_Q(\psi).
		\]
		This proves that $\Phi_Q$ is an isomorphism between $\mathcal{H}_Q$ and $V_Q$.
		
		\item \textbf{Existence of physical packaged subspaces:}	
		The gauge projector $\Pi_{\rm phys}$ annihilates every nontrivial $V_Q$.
		Its image $\mathcal H_{\text{phys}}$ is exactly the multiplicity-space of the trivial irrep in $\mathcal{H}_{\rm iso}(V)$.

		$\mathcal H_{\text{phys}}$ contains at least one non-zero packaged state iff some tensor power of $V_{\rm int}$ contains the trivial irreducible block., i.e., $\mathcal{H}_{\rm phys} \ne \{0\}$ holds when some $V_Q \otimes \cdots \otimes V_Q \ni \mathbf 1$.
		This proves the existence of gauge-invariant physical packaged subspaces $\mathcal H_{\rm phys}$.
	\end{enumerate}
	
	Collecting all the above steps, we finally proved the existence theorem:
	for any finite or compact group $G$, there exists corresponding single-particle packaged subspaces, multi-particle packaged subspaces, orthonormal packaged entangled bases, canonical $G$-module isomorphisms, and physical packaged subspaces.
\end{proof}

\begin{remark}[Scope of the existence theorem]
	The construction in Theorem \ref{THM:ExistenceOfGAssociatedPackagedStates} is valid for any finite or compact group $G$ that acts on the quantum physical system.	
	However, if $G$ does not satisfy the ``finite or compact'' condition, then we need to extend and refine the theorem.
	Specifically:
	\begin{enumerate}
		\item Anti-unitary symmetries.
		if $G$ contains anti-unitary elements (e.g., time reversal), one can work with Wigner’s theorem on projective representations and track complex conjugation on matrix elements.
		
		\item Infinite-dimensional representations.
		A compact group $G$ has finite-dimensional irreps, the construction directly applies.
		But if $G$ is non-compact, then it may have infinite-dimensional irreps.
		Therefore, additional technical conditions are required.
		Here, we only focus on finite-dimensional irreps.
		
		\item Superselection (gauge) vs. global symmetry.
		In local gauge theories, the total charge is conserved and the Hamiltonian is block-diagonalized.
		Superselection rules prohibit coherently superposition of packaged states from different charge sectors.
		This is part of the packaging principle and will not change in any circumstance.
		For global symmetries, although cross-sector superpositions of global quantum numbers are allowed, we usually still focuses on a definite $Q$ sector.
	\end{enumerate}
\end{remark}

\begin{example}[Existence of $\mathrm{SU}(3)$-associated packaged subspaces and module structures]
	We now illustrate each of the five points of Theorem \ref{THM:ExistenceOfGAssociatedPackagedStates} using $\mathrm{SU}(3)$ color
	$$
	G=\mathrm{SU}(3)_{\rm color},\quad V_{\rm int}=\mathbf{3},\quad \mathcal H_{\rm ext}=L^2(\mathbb R^3)\otimes\{\uparrow,\downarrow\}\,.
	$$
	both in the first-quantized (coordinate-space wavefunctions) and corresponding second-quantized operators.
	\begin{enumerate}
		\item \textbf{Single-particle packaged subspaces.}
		\begin{itemize}
			\item \textbf{Coordinate representation.}
			A single quark is described by a wave‐function		
			$$
			\psi_i(\mathbf x,s)\,,\qquad i=1,2,3,\;\;s=\uparrow,\downarrow,
			$$		
			which transforms as		
			$$
			\psi_i(\mathbf x,s)\;\longmapsto\;U_i^{\;j}\,\psi_j(\mathbf x,s)\,,\quad U\in\mathrm{SU}(3)\,.
			$$
			Generally, a quark’s wavefunction splits into		
			$$
			\Psi(\mathbf x,s)\;=\;\sum_{i=1}^3 \psi_i(\mathbf x,s) \,|\,i\rangle\,,\quad
			|\,i\rangle\in\mathbf3_{\rm color}\,,
			$$		
			Thus, the single-quark packaged subspace is		
			$$
			\mathcal H_1\;=\;\mathbf{3}\;\otimes\;L^2(\mathbb R^3)\otimes\{\!\uparrow,\downarrow\},
			$$		
			which carries an irreducible color factor $\mathbf3$ together with trivial action on the external variables.

			\item \textbf{Second quantization.}
			We introduce quark creation operators		
			$$
			a_i^\dagger(\mathbf p,s)\,,\quad i=1,2,3,
			$$		
			which obey $\{a_i(\mathbf p,s),a_j^\dagger(\mathbf p',s')\}=\delta_{ij}\,\delta_{ss'}\,\delta^3(\mathbf p-\mathbf p')$.
			Then the single-quark packaged subspaces (the fundamental $\mathbf3$ of $\mathrm{SU}(3)$) are
			$$
			\operatorname{span}
			\{a_i^\dagger(\mathbf p,s)\,\lvert0\rangle\}
			\;\cong\;
			\mathbf3\otimes\mathcal H_{\rm ext},
			\quad i=1,2,3,
			$$
		\end{itemize}

		\item \textbf{Multi-particle packaged subspaces.}
		\begin{itemize}			
			\item \textbf{Coordinate representation.}
			\begin{enumerate}
				\item Two-quark: $\mathbf3\otimes\mathbf3= \mathbf6\oplus\overline{\mathbf3}$
				
				The two-quark product wave function			
				$
				\Psi_{ij}(\mathbf x_1,s_1;\mathbf x_2,s_2)
				=\psi_i(\mathbf x_1,s_1)\,\psi_j(\mathbf x_2,s_2)
				$
				decomposes into
				\begin{enumerate}
					\item symmetric $\mathbf6$:				
					$$
					\Psi^{(6)}_{(ij)}=\frac1{\sqrt2}\bigl[\Psi_{ij}+\Psi_{ji}\bigr],  
					$$
					
					\item antisymmetric $\overline{\mathbf3}$:				
					$$
					\Psi^{(\overline3)}_{[ij]}=\frac1{\sqrt2}\bigl[\Psi_{ij}-\Psi_{ji}\bigr].
					$$
				\end{enumerate}
				
				\item Quark-antiquark (meson): $\mathbf3\otimes\overline{\mathbf3}=\mathbf1\oplus\mathbf8$
				
				The quark-antiquark product wave function
				$
				\Phi_i{}^j(\mathbf x,\mathbf y;s,s')
				=\psi_i(\mathbf x,s)\,\phi^j(\mathbf y,s')
				$	
				decomposes into
				\begin{enumerate}
					\item singlet $\mathbf1$:				
					$$
					\Phi^{(1)}=\frac1{\sqrt3}\,\delta^{\,i}_{j}\,\Phi_i{}^j,
					$$
					
					\item octet $\mathbf8$:				
					$$
					\Phi^{(8)a}
					=\sqrt2\,\Phi_i{}^j\,(T^a)^i_{\;j},\quad a=1,\dots,8.
					$$
				\end{enumerate}
			\end{enumerate}

			\item \textbf{Second quantization.}			
			We introduce antiquark creation operators $b^{\dagger i}$ that transforms in the $\overline{\mathbf3}$.
			After a Peter-Weyl projection, we have
			\begin{enumerate}
				\item Two-quark.
				$$
				\begin{aligned}
					&\text{six:}\quad
					\frac1{\sqrt2}\,a_i^\dagger\,a_j^\dagger\bigl|0\bigr\rangle
					+
					(i\leftrightarrow j)
					\;\in\;\mathbf6,\\
					&\text{anti-triplet:}\quad
					\frac1{\sqrt2}\,a_i^\dagger\,a_j^\dagger\bigl|0\bigr\rangle
					-
					(i\leftrightarrow j)
					\;\in\;\overline{\mathbf3}.
				\end{aligned}
				$$
				
				\item Quark-antiquark (meson).	
				$$
				\begin{aligned}
					&|\text{singlet}\rangle
					=\frac1{\sqrt3}\,a_i^\dagger\,b^{\dagger i}\,\lvert0\rangle
					\;\in\;\mathbf1,\\
					&|\text{octet},a\rangle
					=a_i^\dagger\,(T^a)^i_{\;j}\,b^{\dagger j}\,\lvert0\rangle
					\;\in\;\mathbf8.
				\end{aligned}
				$$
			\end{enumerate}
		\end{itemize}

		\item \textbf{Packaged entangled bases.}
		
		Let us take the two-quark antisymmetric $\overline{\mathbf3}$ sector as an example:
		\begin{itemize}
			\item \textbf{Coordinate representation.}
			In the two-quark antisymmetric $\overline{\mathbf3}$ sector, one can choose an orthonormal set of fully hybrid-packaged entangled wavefunctions, e.g.,
			$$
			\phi^{(\alpha)}_{ij}(\mathbf x,\mathbf y;s,s')
			=\sum_{i<j}c^{(\alpha)}_{ij}\,\psi_{[ij]}(\mathbf x,\mathbf y;s,s'),
			\quad
			\alpha=1,2,3,
			$$
			with $\sum_{i<j}\overline{c}^{(\alpha)}_{ij}c^{(\beta)}_{ij}=\delta^{\alpha\beta}$.  Each $\phi^{(\alpha)}$ cannot factorize into $\chi(\mathbf x,s)\,\chi(\mathbf y,s')$.
			
			\item \textbf{Second quantization.}
			Applying the Gram-Schmidt procedure on the three antisymmetric generators
			$\hat A^\dagger_{ij}\equiv\hat a_i^\dagger\hat a_j^\dagger-\hat a_j^\dagger\hat a_i^\dagger$, we obtain an orthonormal basis
			$\{\Psi^{(\alpha)}\}\subset\mathcal H_{Q=\overline3}$ all of which are fully hybrid-packaged entangled.
		\end{itemize}

		\item \textbf{Canonical $G$-module isomorphisms.}
		Once again, we take the two-quark antisymmetric $\overline{\mathbf3}$ sector as an example:
		\begin{itemize}
			\item \textbf{Coordinate representation.}
			Each packaged subspace $\mathcal H_Q$ carries exactly the same transformation law as the abstract irrep $V_Q$.
			For instance
			$$
			\phi_{[ij]}(x,y)\;\mapsto\;
			U(g)\,\phi_{[ij]}(x,y)\;=\;
			\sum_{k\ell}(g_i{}^k\,g_j{}^\ell-g_j{}^k\,g_i{}^\ell)\,\phi_{[k\ell]}(x,y).
			$$		
			Therefore, if we construct a map $\Phi_Q$ that sends the $\mathcal H_Q$ basis to the standard $\overline{\mathbf3}$ basis, then $\Phi_Q$ is the isomorphism.
			
			\item \textbf{Second quantization.}
			Again on each packaged subspace $\mathcal H_Q$, the results of action
			$U(g)\,\hat a_i^\dagger\cdots\ket0 = (g_i{}^k\,\hat a_k^\dagger)\cdots\ket0$
			is exactly the corresponding irrep $V_Q$.
			Choosing an orthonormal basis in $\mathcal H_Q$ and construct a map $\Phi_Q$ that maps it to the canonical basis of $V_Q$, then $\Phi_Q$ is the isomorphism.
		\end{itemize}

		\item \textbf{Physical packaged (singlet) subspaces.}
		Finally, we use the gauge‐projector		
		$$
		\Pi_{\rm phys}
		=\int_{\mathrm{SU}(3)}\!U(g)\,dg
		$$		
		to single out exactly the color‐singlet sector from each $n$‐particle packaged subspace.
		\begin{itemize}
			\item \textbf{Coordinate representation.}
			The meson packaged subspace (color singlet)
			$\mathcal H_{\rm phys}$ is spanned by
			$$
			\Psi_{\rm meson}(x,y)
			=\frac1{\sqrt3}\sum_{i=1}^3\psi_i(x)\,\bar\psi_i(y),
			$$
			and all higher composites whose coordinate wavefunctions live in the trivial $\mathrm{SU}(3)$ sector.
			
			\item \textbf{Second quantization.}
			For quark-antiquark $\hat a_i^\dagger\hat b_i^\dagger\ket0$, we have the color-singlet meson
			$$
			\Pi_{\rm phys}\bigl(\hat a_i^\dagger\hat b_i^\dagger\ket0\bigr)
			=
			\frac1{\sqrt3}\sum_i\hat a_i^\dagger\hat b_i^\dagger\ket0
			$$
			spans $\mathcal H_{\rm phys}$.			
			For three quarks $\hat a_i^\dagger\hat a_j^\dagger\hat a_k^\dagger\ket0$, we have the color singlet baryon
			$$
			\Pi_{\rm phys}\bigl(\hat a_i^\dagger\hat a_j^\dagger\hat a_k^\dagger\ket0\bigr)
			=
			\frac1{\sqrt6}\,\varepsilon_{ijk}\,\hat a_i^\dagger\hat a_j^\dagger\hat a_k^\dagger\ket0
			$$
			spans $\mathcal H_{\rm phys}$.
			Whenever the trivial $\mathbf1$ occurs in some tensor‐power of $\mathbf3$, we have $\mathcal H_{\rm phys}\neq\{0\}$.			
		\end{itemize}
		
		These five bullet-points use the physically familiar $\mathrm{SU}(3)$-color context to illustrate the items (1)-(5) of Theorem \ref{THM:ExistenceOfGAssociatedPackagedStates}.
	\end{enumerate}
\end{example}

\paragraph{Maximal orthonormal packaged entangled basis.}
In an isotypic space, each charge sector $\mathcal{H}_Q$ is a packaged subspace that includes all packaged states with total charge $Q$.
The states in $\mathcal{H}_Q$ can form packaged entangled states, denote by $E_Q$.  
Since the sum of two packaged entangled states can yield a product (separable) state, $E_Q$ is not closed under addition and hence does not form a linear subspace.
According to item 3 of Theorem \ref{THM:ExistenceOfGAssociatedPackagedStates}, however, we can pick out a complete orthonormal basis from $E_Q$ that spans $\mathcal{H}_Q$.
This is crucial for applications in quantum information because it guarantees the existence of a maximal orthonormal set $\{\lvert \Psi_i\rangle\}$ of packaged entangled states spanning $\mathcal{H}_Q$.  
Consequently, Bell-like measurements on $\mathcal{H}_Q$ become possible.

\begin{example}[Maximal orthonormal packaged entangled basis]
	\label{exm:MaximalOrthonormalBasisOfPES}
	Consider a particle-antiparticle system (e.g., an electron and a positron).  
	A simple basis for the neutral sector $Q=0$ is given by the two orthonormal states
	\begin{equation}\label{eq:PESParticleAntiparticle}
		\bigl\lvert \Psi^\pm \bigr\rangle 
		\;=\;
		\frac{1}{\sqrt2}\,\Bigl[
		\hat{a}^\dagger(\mathbf{p}_1)\,\hat{b}^\dagger(\mathbf{p}_2)
		\;\pm\;
		\hat{b}^\dagger(\mathbf{p}_1)\,\hat{a}^\dagger(\mathbf{p}_2)
		\Bigr]
		\lvert 0\rangle,
	\end{equation}
	where $\hat{a}^\dagger$ and $\hat{b}^\dagger$ create opposite charges $\pm e$.  
	Each $\lvert \Psi^\pm\rangle$ is a packaged entangled state, and $\{\lvert\Psi^+\rangle,\lvert\Psi^-\rangle\}$ forms an orthonormal basis for $\mathcal{H}_{Q=0}$ in this simplified two-particle model.
\end{example}

\subsection{Mixed‑state (Density‑matrix) Packaging}

Let us now develop the adjoint (Heisenberg‑picture) version of the isotypic decomposition on trace‑class operators, block‑diagonal channels, and how packaged subalgebras show up in open system models.

\begin{corollary}[Mixed state (density matrix) packaging]
	\label{COR:DensityMatrixPackaging}
	Let $\mathcal H$ be a unitary $G$‐module with isotypic decomposition
	\[
	\mathcal H
	\;=\;
	\bigoplus_{\lambda \in \widehat G}
	\Bigl(V_\lambda \otimes \mathcal M_\lambda\Bigr),
	\]
	where each $V_\lambda$ carries the irrep of type $\lambda$ (dimension $d_\lambda$) and $\mathcal M_\lambda$ is the corresponding multiplicity space.
	Then we have:
	
	\begin{enumerate}[label=(\alph*)]
		\item \emph{Block form.}
		Every $G$-invariant density operator $\rho \in \mathcal B(\mathcal H)$ decomposes uniquely as
		\[
		\rho =
		\bigoplus_{\lambda\in\widehat G}
		\Bigl(\tfrac1{d_\lambda}\,\operatorname{Id}_{V_\lambda}\Bigr)
		\otimes
		\rho_\lambda,
		\qquad
		\rho_\lambda \in \mathcal B(\mathcal M_\lambda),
		\]
		where each $\rho_\lambda\ge0$ and $\sum_\lambda \operatorname{Tr}(\rho_\lambda)=1$.
		
		\item \emph{Operational superselection.}
		If $O \in \mathcal B(\mathcal H)$ is a $G$-invariant observable, then
		\[
		O =
		\bigoplus_{\lambda \in \widehat G}
		\operatorname{Id}_{V_\lambda} \otimes O_\lambda,
		\quad
		O_\lambda\in \mathcal B(\mathcal M_\lambda),
		\]
		and
		\[
		\operatorname{Tr} \bigl[\rho\,O\bigr]
		\;=\;
		\sum_{\lambda\in\widehat G}
		\operatorname{Tr}\bigl[\rho_\lambda\,O_\lambda\bigr].
		\]
		In particular, any off‐diagonal coherence between different $\lambda$-blocks makes no contribution to any such expectation value.
		
		\item \emph{Decoherence stability.}
		Let $\mathcal E: \mathcal B(\mathcal H) \to \mathcal B(\mathcal H)$ be a quantum channel whose Kraus operators $K_i$ all commute with the $G$-action: $[K_i, U(g)]=0$ for every $g \in G$.
		Then for any $G$-invariant input $\rho$, the output $\mathcal E(\rho)$ is also $G$-invariant and still has the block form of part (a).
		Thus, no $G$-forbidden coherences can be generated by such a symmetric channel.
	\end{enumerate}
\end{corollary}

\begin{proof}
	\leavevmode
	\begin{enumerate}[left=1em]
		\item[\textbf{(a)}]
		Write the isotypic decomposition as
		\[
		\mathcal H =
		\bigoplus_{\lambda} \bigl(V_\lambda\otimes\mathcal M_\lambda\bigr).
		\]
		In this block basis, any operator $\rho$ can be written in matrix‐block form
		\[
		\rho = 
		\sum_{\lambda, \mu}
		\bigl(A_{\lambda\mu}\bigr)
		\quad
		\text{with }
		A_{\lambda\mu} \in
		\mathcal B \bigl(V_\lambda \otimes \mathcal M_\lambda,\;V_\mu\otimes\mathcal M_\mu\bigr).
		\]
		When $V_\lambda\not\cong V_\mu$, by Schur’s Lemma, no nonzero $G$-map exists between in-equivalent irreps.
		Therefore, $G$-invariance,
		$
		U(g)\,\rho\,U(g)^{-1}=\rho
		\;\;\forall g\in G,
		$
		forces every off‐block $A_{\lambda\mu}=0$ for $\lambda\neq\mu$.
		We have
		\[
		\rho
		\;=\;
		\bigoplus_{\lambda}
		A_{\lambda\lambda},
		\quad
		A_{\lambda\lambda}
		\in
		\mathcal B \bigl(V_\lambda \otimes \mathcal M_\lambda\bigr).
		\]
		Now using Schur’s Lemma again, any $G$-invariant operator on $V_\lambda$ is a scalar multiple of the identity.
		Thus
		\[
		A_{\lambda\lambda}
		=
		\bigl(c_\lambda\,\operatorname{Id}_{V_\lambda}\bigr)
		\otimes
		\rho_\lambda,
		\quad
		\rho_\lambda\in\mathcal B(\mathcal M_\lambda),
		\]
		and positivity plus $\operatorname{Tr}\rho=1$ fix
		$
		c_\lambda = 1/d_\lambda
		$
		and $\sum_\lambda \operatorname{Tr}(\rho_\lambda) = 1$.
		
		\item[\textbf{(b)}]
		Any $G$-invariant observable $O$ has the same block structure as $\rho$:
		\[
		O
		=
		\bigoplus_\lambda
		\bigl(\operatorname{Id}_{V_\lambda} \otimes O_\lambda\bigr).
		\]
		Hence
		\[
		\operatorname{Tr}\bigl[\rho\,O\bigr]
		=
		\sum_\lambda
		\operatorname{Tr}\!\bigl[
		(d_\lambda^{-1} \operatorname{Id}_{V_\lambda}\otimes\rho_\lambda)
		(\operatorname{Id}_{V_\lambda} \otimes O_\lambda)
		\bigr]
		=
		\sum_\lambda
		\operatorname{Tr}(\rho_\lambda\,O_\lambda),
		\]
		and any would‐be off‐diagonal term mixing different $\lambda$’s vanishes identically in this trace.
		
		\item[\textbf{(c)}]
		If each Kraus operator $K_i$ commutes with $U(g)$, then for a $G$-invariant input $\rho$,
		\[
		U(g)\,\mathcal E(\rho)\,U(g)^{-1}
		= 
		\sum_i
		U(g)\,K_i\,\rho\,K_i^\dagger\,U(g)^{-1}
		=
		\sum_i
		K_i\,\bigl(U(g)\rho U(g)^{-1}\bigr)\,K_i^\dagger
		= 
		\mathcal E(\rho).
		\]
		Thus, $\mathcal E(\rho)$ is again $G$-invariant and so by part (a) must decompose into the same block form.
		In particular, no cross‐$\lambda$ coherences can ever arise.
	\end{enumerate}
\end{proof}

This corollary shows that the symmetry packaging survives mixing:
purity is not required.
Superselection shows up as block‑diagonal density matrices and decoherence
empties the off‑diagonal blocks faster.

\begin{remark}[Density matrix packaging]
	The same decomposition extends to the space of trace‑class operators via the adjoint action of $G$.
	Packaged blocks correspond to subalgebras of observables and symmetry‑respecting quantum channels act block‑diagonally on each irreducible sector.
\end{remark}

\subsection{Quantifying Packaged Entanglement}
\label{SEC:QuantifyingPackagedEntanglement}

Under some symmetry (gauge or otherwise), the full Hilbert space decomposes into irreducible packets. 
Specifically, if
$$
G\;\curvearrowright\;\mathcal H
\quad\Longrightarrow\quad
\mathcal H
\;\cong\;
\bigoplus_{\lambda\in\widehat G}
\;V_\lambda\;\otimes\;M_\lambda,
$$
then we call each $V_\lambda$ an internal package.
Similarly, if there is a symmetry on some external DOFs (spin, helicity, momentum, $\dots$) with irreps $W_\mu$, then
$$
\mathcal H
\;\cong\;
\bigoplus_{\lambda,\mu}
\;V_\lambda\;\otimes\;W_\mu\;\otimes\;R_{\lambda,\mu}
$$
where each $W_\mu$ is an external package, while $V_\lambda\otimes W_\mu$ is a hybrid package.

Whenever only a subset $\Gamma_{\rm allowed}\subset\widehat G$ of internal charges and/or a subset $\Delta_{\rm allowed}\subset\widehat{G_{\rm ext}}$ of external charges is physically realized, we introduce the hybrid Peter-Weyl projector
$$
P_{\rm hyb}
\;=\;
\sum_{\substack{\lambda\in\Gamma_{\rm allowed}\\\mu\in\Delta_{\rm allowed}}}
\;P^{\rm int}_\lambda\;\otimes\;P^{\rm ext}_\mu,
$$
where
$$
P^{\rm int}_\lambda
\;=\;
d_\lambda
\int_G d\mu(g)\;\chi_\lambda^*(g)\,U(g),
\qquad
P^{\rm ext}_\mu
\;=\;
d_\mu
\int_{G_{\rm ext}} d\nu(h)\;\chi_\mu^*(h)\,V(h).
$$

Physical states obey
$
P_{\rm hyb}\,\ket\Psi
= \ket\Psi
$
and live in
$\mathcal H_{\rm pkg} = \operatorname{Im} P_{\rm hyb} \subset \mathcal H$.
We now give two complementary kinematic measures of the resulting packaged entanglement.

\subsubsection{Packaged Schmidt Rank}

We first do a bi-partition to our system by splitting it into two subsystems $A$, $B$ so that
$$
\mathcal H_{\rm pkg}
\subset
\mathcal H_A \otimes \mathcal H_B.
$$
Then we restrict the Schmidt decomposition:
for any pure $\ket\Psi\in\mathcal H_{\rm pkg}$, we perform the usual Schmidt decomposition within $\mathcal H_{\rm pkg}$:
$$
\ket\Psi
= 
\sum_{k=1}^{R_{\rm pack}}
\lambda_k \ket{\psi_k}_A \otimes \ket{\phi_k}_B,
$$
where by construction each $\ket{\psi_k}_A,\ket{\phi_k}_B$ lies in
$\mathcal H_{A,\rm pkg} = \operatorname{Im}\bigl(P_{\rm hyb}\vert_A\bigr)$ and similarly on $B$.

Finally, we define
$$
\mathrm{Packaged\ Schmidt\ Rank}(\Psi)
\;=\;
R_{\rm pack}
$$
Since $\dim\mathcal H_{\rm pkg}\le\dim(\mathcal H_A\otimes\mathcal H_B)$, one always has $R_{\rm pack}\le R_{\rm full}$, and $R_{\rm pack}=1$ if and only if $\ket\Psi$ is un‐entangled within the packaged subspace.

\subsubsection{Packaged Entanglement Entropy}

We first write out the projected density:
$$
\rho \;=\; \ket\Psi\!\bra\Psi,
\quad
\rho = P_{\rm hyb}\,\rho\,P_{\rm hyb}.
$$
Then take partial trace to obtain the reduced density on $A$ by tracing out $B$:
$$
\rho_A
\;=\;
\operatorname{Tr}_B(\rho),
$$
which satisfies $[\rho_A,P_{\rm hyb}\vert_A]=0$.

Finally, we define the packaged entanglement entropy \cite{Shannon19483,Shannon19484,VedralPlenio1998,Casini2014,DonnellyWall2015} as:
$$
	S_{\rm pack}(A : B)
	\;=\;
	-\,\operatorname{Tr} \bigl(\rho_A\, \log\rho_A\bigr)
$$
This counts only the correlations compatible with both internal and external superselection \cite{Casini2014}.

\begin{example}
	Illustrative Examples:
	
	\begin{itemize}
		\item QCD mesons.
		Internal packaging: $\mathbf3\otimes\bar{\mathbf3}=\mathbf1\oplus\mathbf8$.
		$P_{\rm hyb}=P_{\mathbf1}$ isolates the color singlet component.
		Both $R_{\rm pack}$ and $S_{\rm pack}$ measure genuine color entanglement.
		
		\item $\mathbb Z_N$ Gauge theory.
		Internal packaging: project onto zero‐flux sector.
		The packaged entanglement entropy counts only the flux string entanglement crossing the $A - B$ boundary.
	\end{itemize}
\end{example}

\begin{remark}
	One can refine these further by isolating each irreducible block $(\lambda,\mu)$ individually by defining
	$$
	\rho_{A;\lambda,\mu}=\operatorname{Tr}_B\bigl(P_{\lambda,\mu}\,\rho\,P_{\lambda,\mu}\bigr)
	$$
	and
	$$
	S_{\lambda,\mu}=-\operatorname{Tr}(\rho_{A;\lambda,\mu}\log\rho_{A;\lambda,\mu}).
	$$
	This quantifies hidden correlations inside each internal-external package.
\end{remark}

This unified approach makes it clear that any symmetry yields a projector onto allowed blocks.
Within that block, one can still meaningfully define Schmidt rank and von Neumann entropy measures of packaged entanglement.

\section{Symmetry Packaging under Finite Groups}
\label{SEC:SymmetryPackagingOnFiniteGroups}

A finite group \cite{Artin2014,Serre1977} is one that has a finite number of elements.
All finite groups are special cases of discrete groups (topological groups where the topology is discrete, i.e., every singleton set is open).
Finite groups play a dual role in physics:
they may appear either as local (gauge) symmetries or as global discrete symmetries \cite{Wegner1971,Hooft1978,DijkgraafWitten1990}.
Finite groups also play important roles in many physical systems \cite{Landau1957,Wigner1959} by further constraining or partitioning the IQNs.
In these cases, the packaging principle is refined by imposing that states also lie in definite discrete symmetry sub-sectors.
In either case, the representation theory of a finite group $G$ is completely reducible, which makes the classification of charge sectors and the enforcement of superselection rules transparent.

In this section, we explore how finite groups refine the construction and classification of packaged states and how projection operators \cite{Hellinger1909,ReedSimon1972,Conway2007} and superselection rules ensure that the packaged charges remain intact.

\subsection{General Finite Groups}
\label{SEC:GeneralFiniteGroups}

We now explore symmetry packaging on a general finite group, without restricting ourselves to any particular group structure.
Let us first discuss the finite group $G$ as a local gauge symmetry.

\subsubsection{Finite Group $G$ as a Local Gauge Symmetry}

A finite gauge theory is defined by associating to each link (or edge) of a lattice a group element $U_{xy}\in G$.
It requires that physical states be invariant under local transformations
$ U_{xy} \;\mapsto\; g(x)\,U_{xy}\,g(y)^{-1}, \qquad g(x)\in G \,.$
Matter fields $\psi(x)$ transform at site $x$ by the appropriate finite-group representation:
$$
\psi(x)\;\mapsto\; \rho_R\bigl(g(x)\bigr)\,\psi(x),
$$
where $R$ labels an irrep of dimension $d_R$.
The Hilbert space of a single excitation then decomposes as
$$
\mathcal H_1
= \bigoplus_{R \in \widehat G} \;V_R \otimes \mathcal H_{\rm ext},
$$
where $V_R$ is the carrier space of the irrep $R$ and $\mathcal H_{\rm ext}$ encodes any additional spacetime or flavor DOFs.
The packaging principle asserts that once a local gauge charge $R$ is chosen, the corresponding irrep block $V_R$ cannot be further divided by any gauge-invariant operator.

Confinement or singlet constraints (Gauss law) come from the requirement that global gauge transformations act trivially on physical states:
$$
\bigotimes_{x\in\Lambda} \rho_R\bigl(g(x)\bigr)\,|\mathrm{phys}\rangle
= |\mathrm{phys}\rangle.
$$

Equivalently, the gauge projector onto net charge-zero (or more general net-$R_0$) sectors is built from the group average:
$$
\Pi_{R_0}
= \frac{\dim R_0}{|G|}\sum_{g\in G} \chi_{R_0}^*(g)\,U(g),
$$
which enforces that only the trivial or targeted irrep appears in the global decomposition of the Fock space.

\begin{example}[Packaged states in a finite gauge theory of $S_3$]
	Let $G=S_3$, with three irreps:
	two 1D ($\mathbf1$, $\mathbf{sgn}$) and one 2D ($\mathbf2$).
	A single particle in the $\mathbf2$-sector carries a two-component internal index that cannot be split by gauge-invariant measurements.
	Gauge-invariant multi-particle states live in the trivial sector of the product $R_1\otimes R_2\otimes\cdots$ and may form nontrivial packaged entangled singlets only when $\mathbf2^{\otimes n}$ contains $\mathbf1$ or $\mathbf{sgn}$.
\end{example}

\begin{example}[Symmetry packaging for the quaternion group $Q_8$]
	The quaternion group
	\[
	Q_8 \;=\;\{\pm1,\;\pm i,\;\pm j,\;\pm k\},
	\quad
	i^2=j^2=k^2=ijk=-1,
	\]
	has five irreps over $\mathbb C$:
	\begin{itemize}
		\item Four 1-dimensional irreps (each factors through $Q_8/\{\pm1\}\cong \mathbb Z_2\times\mathbb Z_2$).
		
		\item One 2-dimensional irrep $V_{\mathbf2}$.
	\end{itemize}
	
	By the packaging principle, any state transforming in $V_{\mathbf2}$ furnishes a two-component block that cannot be further split by $Q_8$-invariant operators.
	The projector onto that 2D sector is
	\[
	P_{\mathbf2}
	\;=\;
	\frac{2}{|Q_8|}\,
	\sum_{g\in Q_8}\chi_{\mathbf2}^*(g)\,U(g),
	\]
	where $\chi_{\mathbf2}(g)$ is the character of the 2-dimensional irrep.
	Any packaged state $\ket{\Psi}\in V_{\mathbf2}$ has Schmidt rank at most 2 and remains inseparable under the full $Q_8$ action.
	
	In contrast, states in any of the four 1-dimensional irreps each carry a single discrete charge and behave exactly like $\mathbb Z_2$ quanta.
\end{example}

\begin{example}[Symmetry packaging on local gauge group $\mathbb{Z}_N$]
	Let us now treat the cyclic group $\mathbb{Z}_N$ as a local gauge group.
	Then the $\mathbb{Z}_N$ transformation varies from point to point and a gauge field $U_{xy}$ is introduced to ensure local invariance.
	The physical Hilbert space is restricted to gauge-invariant states.

	Let $g(x) \in \mathbb{Z}_N$ be a local gauge transformation.
	Then for each spacetime point $x$, we have
	\[
	g(x) \;\in\;\{\,\omega^{k(x)}:k(x)\in\{0,\dots,N-1\}\}, 
	\quad \omega = e^{2\pi i/N}.
	\]
	A matter field $\psi(x)$  that carries discrete charge $q\in \{0,1,\dots,N-1\}$ transforms as
	\[
	\psi(x) \;\mapsto\;\omega^{\,q}\,\psi(x).
	\]
	To maintain local gauge-invariance, a gauge field $U_{xy}$ is introduced on links between points $x$ and $y$ so that the covariant derivative transforms homogeneously.
	Its transformation law is:
	\[
	U_{xy} \mapsto \omega^{q(x)-q(y)}U_{xy}.
	\]
	Physical states $|\mathrm{phys}\rangle$ must be invariant under local gauge transformations (must satisfy Gauss-law constraints).
	For instance, if $G(x)$ is the generator of local gauge transformations at site $x$, then
	\[
	G(x)~ |\mathrm{phys}\rangle = |\mathrm{phys}\rangle \quad \forall x.
	\]
	Usually, in confining theories, only the gauge singlet (net charge zero) sector is truly gauge-invariant.

	The irreps of $\mathbb{Z}_N$ (which, in the 1D case, are given by $\rho_d(g) = \exp(2\pi i\,d/N)$) label the local gauge charge at each point.
	However, due to the local nature, the fields must be dressed by gauge fields and the physical Hilbert space is projected onto gauge-invariant states (e.g., net $\mathbb{Z}_N$ charge zero).
	In this way, the IQNs are packaged together, but the constraint is enforced locally via Gauss’s law.
\end{example}

\subsubsection{Finite Group $G$ as a Pure Global Discrete Symmetry}

If $G$ acts globally \cite{Potts1952,Jose1977}, a single operator $U(g)$ implements the same group element everywhere in spacetime:
$$
U(g)\,:\quad \mathcal H \to \mathcal H, \qquad g \in G,
$$
and physical observables or states need only satisfy invariance under the global action.
The full Hilbert space decomposes into isotypic sectors:
$$
\mathcal H
= \bigoplus_{R\in\widehat G}
\Bigl(V_R\otimes\mathcal M_R\Bigr),
$$
where $V_R$ carries the irrep $R$ and $\mathcal M_R$ is its multiplicity space.
A global packaged state is then any vector in a fixed charge sector $V_R\otimes\mathcal M_R$, and no gauge field is needed to enforce superselection.

The corresponding projector onto the $R$-sector is
$$
P_R
= \frac{d_R}{|G|}\sum_{g\in G} \chi_R^*(g)\,U(g),
$$
which picks out from an arbitrary state its component transforming in $R$.
Coherent superpositions between different $R$’s are forbidden by superselection rules, and any global-symmetry-invariant observable is block-diagonal in the $R$-label.

\begin{example}[Global packaging for a non-abelian finite group]
	Consider again $G=S_3$.
	The Hilbert space splits as	
	$$
	\mathcal H = (\mathbf1\otimes\mathcal M_{\mathbf1})\oplus(\mathbf{sgn}\otimes\mathcal M_{\mathbf{sgn}})\oplus(\mathbf2\otimes\mathcal M_{\mathbf2}).
	$$	
	A creation operator $a_{\alpha,i}^\dagger$ that transforms in the 2D irrep satisfies	
	$$U(g),a_{\alpha,i}^\dagger,U(g)^{-1} = D_{\alpha\beta}^{(\mathbf2)}(g),a_{\beta,i}^\dagger
	$$	
	and cannot be decomposed further into smaller symmetry charges.
	Observables commuting with $U(g)$ act as $\operatorname{Id}_{V_R}\otimes O_R$ on each sector.
\end{example}

\begin{example}[Symmetry packaging on a pure global $\mathbb Z_N$]
	Symmetry $\mathbb Z_N$ clock model:
	
	\begin{enumerate}
		\item Local Hilbert space
		At each site (or for each particle) introduce an $N$-dimensional Hilbert space with orthonormal basis		
		$
		\bigl\{\ket{q}\,\big|\,q=0,1,\dots,N-1\bigr\}.
		$
		
		\item Global $\mathbb Z_N$ action
		Let the generator $g$ of $\mathbb Z_N$ act by		
		$
		U(g)\,\ket{q}
		= \ket{q+1 \;(\mathrm{mod}\;N)}.
		$		
		Equivalently, in the character basis (the irreps of $\mathbb Z_N$), define
		$$
		\ket{\chi_k}
		\;=\;\frac{1}{\sqrt{N}}\sum_{q=0}^{N-1}e^{-2\pi i\,kq/N}\,\ket{q},
		\quad k=0,1,\dots,N-1,
		$$		
		on which		
		$
		U(g)\,\ket{\chi_k}
		= e^{2\pi i\,k/N}\,\ket{\chi_k}.
		$	
		Here each $\ket{\chi_k}$ spans a one-dimensional irrep $V_k$ of $\mathbb Z_N$.

		\item Single-particle packaging.
		
		By construction, any local creation operator $a^\dagger_k$ that produces $\ket{\chi_k}$ transforms as		
		$$
		U(g)\,a^\dagger_k\,U(g)^\dagger
		= e^{2\pi i\,k/N}\,a^\dagger_k,
		$$		
		so it carries exactly one irrep $V_k$.
		No operator can split $V_k$ into smaller pieces, since all irreps of $\mathbb Z_N$ are one-dimensional.

		\item Two-particle Hilbert space and net charge sectors.
		
		The two-particle space is		
		$$
		\mathcal H^{(2)}
		= \mathcal H^{(1)}\otimes\mathcal H^{(1)}
		= \bigoplus_{k=0}^{N-1} V_k^{(\rm tot)}
		$$		
		where the total charge $k$ sector is		
		$$
		V_k^{(\rm tot)}
		= \operatorname{span}\!\bigl\{
		\ket{\chi_j}\otimes\ket{\chi_{\,k-j}} \;\big|\;j=0,\dots,N-1
		\bigr\},
		$$		
		since the characters multiply:
		$\chi_j\,\chi_{k-j}=\chi_k$.

		\item Packaged entangled singlet ($k=0$).
		
		In the neutral sector $k=0$, a canonical packaged entangled state is		
		$$
		\ket{\Psi_{\rm singlet}}
		\;=\;\frac{1}{\sqrt{N}}
		\sum_{q=0}^{N-1}
		\ket{q}\otimes\ket{-q\;(\mathrm{mod}\;N)}
		\;=\;
		\frac{1}{\sqrt{N}}
		\sum_{j=0}^{N-1}
		\ket{\chi_j}\otimes\ket{\chi_{-j\;(\mathrm{mod}N)}}.
		$$
		It is invariant under the global $\mathbb Z_N$ since each term carries total charge $0$.
		It is also packaged entangled (Schmidt rank $N$) arising from the irreducible block structure of the neutral sector.
	\end{enumerate}
\end{example}

\subsubsection{Comparison}

In Table \ref{TAB:FiniteGaugeVsGlobal}, we contrast the two realizations of a finite group $G$ as a local gauge symmetry versus a global discrete symmetry.

\begin{table}[h]
	\centering
	\caption{Finite group $G$: gauge vs. global}
	\label{TAB:FiniteGaugeVsGlobal}
	\begin{tabular}{p{3cm}|p{6cm}|p{5.5cm}}
		\toprule
		Feature & Local (Gauge) $G$ & Global Discrete $G$ \\
		\midrule
		Action & $g(x)$ varies on each site, link variables $U_{xy}$ enforce covariance & Single $g$ acts uniformly via $U(g)$ \\
		\hline
		Constraint & Gauss law: physical states lie in gauge-invariant subspace, net charge often trivial & No local constraint beyond global invariance, each charge sector is physical \\
		\hline
		Isotypic decomposition & $\mathcal H = \bigoplus_R (V_R \otimes \mathcal H_{\rm ext})$, then gauge‐project to $R_0$ & $\mathcal H = \bigoplus_R (V_R \otimes \mathcal M_R)$, superselection forbids mixing \\
		\hline
		Projector & $\Pi_{R_0} = \frac{d_{R_0}}{|G|}\sum_{g\in G}\chi_{R_0}^*(g)\;\bigotimes_{x\in\Lambda}U_x(g)$ & $P_R=\frac{d_R}{|G|}\sum_g\chi_R^*(g),U(g)$ \\
		\hline
		Physical sectors & Only gauge‐singlet (or chosen $R_0$) sector survives & All irreps $R$ give distinct, non‐mixing sectors \\
		\bottomrule
	\end{tabular}
\end{table}

This unified view shows that finite groups, whether local or global, impose a rigid block structure on the Hilbert space.
The packaging principle remains the same:
once an internal charge label $R$ is assigned (via an irrep), it cannot be fractionated by any symmetry‐preserving operator.

\subsection{$\mathbb{Z}_2$ Group}
\label{sec:Z2examples}

$\mathbb{Z}_2$ group is the simplest nontrivial discrete group.
It has two elements $\{e,\,\hat{O}\}$, where $e$ is the identity element and $\hat{O}^2 = e$.
Here we will use $\hat{O}$ to represent charge conjugation ($\hat{C}$), parity ($\hat{P}$), time reversal ($\hat{T}$), and their combinations such as $\hat{C}\hat{P}$, $\hat{C}\hat{T}$, $\hat{P}\hat{T}$, and $\hat{C}\hat{P}\hat{T}$.

\subsubsection{Corollary: Packaged Entanglements on $\mathbb{Z}_2$ Symmetry}
\label{SEC:PESBasedOnZ2SymmetryOperators}

Discrete symmetry operators can also induce their own particular form of entanglements.
Here we focus on the simplest discrete group $\mathbb{Z}_2$ \cite{Ising1925,Onsager1944}.
It is simple, but very important in physics.

Referring to Theorem~\ref{THM:ExistenceOfGAssociatedPackagedStates}, we obtain the following corollary:

\begin{corollary}[Bell-like Entanglements under $\mathbb{Z}_{2}$ Symmetries]
	\label{COL:DiscreteSymmetryAndEntanglement}
	Consider a system consisting of two identical subsystems $A$ and $B$.
	Each subsystem is described by a set of quantum numbers $Q = (Q_1, Q_2, \dots, Q_n)$.
	If there is a $\mathbb{Z}_{2}$ symmetry operator $\hat{O}$ acting on the system such that
	$
	\hat{O}\,\lvert Q\rangle = \lvert -Q\rangle, 
	~
	\hat{O}\,\lvert -Q\rangle = \lvert Q\rangle,
	$
	then in the two-level Hilbert space:
	\begin{enumerate}
		\item Mathematically, one can construct four Bell-like entangled states:
		\begin{align}\label{EQ:BellLikeEntangledStates}
			\begin{aligned}
				&\lvert \Psi^{\pm} \rangle_{AB}
				\;=\;
				\frac{1}{\sqrt{2}}\,
				\Bigl(
				\lvert Q \rangle_A \,\lvert -Q \rangle_B
				\;\pm\;
				\lvert -Q \rangle_A \,\lvert Q \rangle_B
				\Bigr), \\[6pt]
				&\lvert \Phi^{\pm} \rangle_{AB}
				\;=\;
				\frac{1}{\sqrt{2}}\,
				\Bigl(
				\lvert Q \rangle_A \,\lvert Q \rangle_B
				\;\pm\;
				\lvert -Q \rangle_A \,\lvert -Q \rangle_B
				\Bigr).
			\end{aligned}
		\end{align}		
		Each of these states is an eigenstate of $\hat{O}\otimes \hat{O}$ with eigenvalue $\pm1$.
		
		\item The physical realization of the states in Eq.~(\ref{EQ:BellLikeEntangledStates}) are subject to superselection rules, i.e., if $\hat{O}$ is the charge conjugation operator $\hat{C}$, then states $\lvert \Phi^{\pm} \rangle_{AB}$ are not allowed.
		
		\item Applying a tensor product to the $\mathbb{Z}_{2}$ groups, we naturally obtain the GHZ states, W states, and Dicke states.
	\end{enumerate}
\end{corollary}

\begin{proof}
	Let us split the proof into three parts.
	We have assumed that $\mathbb{Z}_{2}$ symmetry operator $\hat{O}$ swaps the physical quantity set, $\hat{O}: \lvert Q\rangle\leftrightarrow \lvert -Q\rangle$.
	\begin{enumerate}
		\item On the two-level Hilbert space, we have:
		\[
		(\hat{O}_A \otimes \hat{O}_B)\,\bigl(\lvert Q\rangle_A \,\lvert -Q\rangle_B\bigr)
		= \lvert -Q\rangle_A \,\lvert Q\rangle_B,
		\quad
		(\hat{O}_A \otimes \hat{O}_B)\,\bigl(\lvert -Q\rangle_A \,\lvert Q\rangle_B\bigr)
		= \lvert Q\rangle_A \,\lvert -Q\rangle_B.
		\]
		
		Applying the tensor product $\hat{O}_A \otimes \hat{O}_B$ to Eq.~\eqref{EQ:BellLikeEntangledStates}, we have
		\[
		(\hat{O}_A \otimes \hat{O}_B)\,\lvert \Psi^\pm\rangle_{AB}
		= \pm\,\lvert \Psi^\pm\rangle_{AB},
		\quad
		(\hat{O}_A \otimes \hat{O}_B)\,\lvert \Phi^\pm\rangle_{AB}
		= \pm\,\lvert \Phi^\pm\rangle_{AB}.
		\]
		
		This shows that each $\lvert \Psi^\pm\rangle$ and $\lvert \Phi^\pm\rangle$ is an eigenstate of $\hat{O}_A \otimes \hat{O}_B)$.
		
		\item In case of charge conjugation ($\hat{O} = \hat{C}$, superselection rules forbid mixing $\lvert +Q\rangle$ and $\lvert -Q\rangle$ unless the total charge of the bipartite system is zero.
		Physically, one must ensure the allowed total charge sector remains consistent with gauge-invariance.
		Therefore, states $\lvert \Phi^{\pm} \rangle_{AB}$ are not allowed because they try to superpose states with different charge.
		
		\item By applying the tensor product of the $\mathbb{Z}_{2}$ groups, we naturally obtain the GHZ states, W states, and Dicke states.
	\end{enumerate}
\end{proof}

According to Corollary~\ref{COL:DiscreteSymmetryAndEntanglement}, a $\mathbb{Z}_2$ symmetry gives rise to $\mathbb{Z}_2$-type entanglements.
Because the $\mathbb{Z}_2$ group here is not treated as a gauge group, it can result in both internal entanglement (packaged entanglement) and external entanglement (spin, momentum, etc.).
This depends on the specific operation represented by the $\mathbb{Z}_2$ group.

\subsubsection{Internal: $\mathbb{Z}_2$ as a Gauge Group}

In the special case of the finite group $\mathbb{Z}_2$ (with elements $\{e, \hat{O}\}$ and $\hat{O}^2 = I$),\cite{Wegner1971} one typically represents the nontrivial element by, say, $U_o = \sigma_x$ in a two-dimensional Hilbert space with basis $\{|+\rangle, |-\rangle\}$. The corresponding projectors
\[
P_{+} = \frac{1}{2}\Bigl(I + U_o\Bigr), \quad P_{-} = \frac{1}{2}\Bigl(I - U_o\Bigr)
\]
then partition the Hilbert space into two superselection sectors. Consequently, a packaged state carrying a local gauge charge $Q$ is assigned the full label $(Q,+)$ or $(Q,-)$, ensuring that its discrete IQNs remain completely inseparable.

Therefore, when $G$ is finite, its complete reducibility and often 1D irreps guarantee that discrete labels cannot be partially changed within a single local gauge charge sector.
The packaging principle thus remains in force:
each single-particle or multiparticle wavefunction belongs to a definite irreps block $\mathcal{H}_{Q,d}$ with no partial factorization allowed.
Hence, packaged quantum states in finite gauge theories are classified by $(Q,d)$, combining the usual gauge charge $Q$ with the discrete label $d$.

\begin{example}[$\mathbf Z_2$ lattice gauge theory]
	Consider a gauge group $G=\mathbf Z_2=\{1,-1\}$ with $|G|=2$.
	
	\begin{itemize}
		\item Packaged single-charge operator.
		
		At a site $x$ introduce $c^\dagger_x$ that transforms as
		$$
		U_{-1}(x)\,c^\dagger_x\,U_{-1}^{-1}(x)=-c^\dagger_x.
		$$
		This is a packaged block in the non-trivial irrep (odd charge).
		
		\item Gauss projector at $x$.
		
		$P_x=\tfrac12\bigl(U_{1}(x)+U_{-1}(x)\bigr)$.
		
		Action on a single charge.		
		$$
		P_x\,c^\dagger_x\ket0
		\;=\;
		\tfrac12\bigl(c^\dagger_x-c^\dagger_x\bigr)\ket0
		\;=\;0 .
		$$
		
		Hence an isolated charge is projected out - it violates Gauss law.
		
		\item Physical excitation.
		
		Take two charges at $x,y$:		
		$$
		\ket{\psi_{xy}}
		=c^\dagger_x\,c^\dagger_y\ket0,
		\qquad
		P_xP_y\,\ket{\psi_{xy}}=\ket{\psi_{xy}} .
		$$
		The pair lives in the $+1$ irrep ($\mathbf 1$) of $\mathbf Z_2$ and survives the projector.
	\end{itemize}
	
	This exactly parallels the color-singlet construction in QCD.
\end{example}

\subsubsection{Internal: Charge Conjugation $\hat{C}$}

In QFT, charge conjugation $\hat{C}$ is an operator invented to perform a reversal operation that transforms each particle $\lvert P\rangle$ into its antiparticle $\lvert \bar{P} \rangle$, and vice versa, $\lvert P\rangle \leftrightarrow \lvert \bar{P} \rangle$ \cite{WWW1952,Luders1954}.
The charge conjugation $\hat{C}$ here is not a gauge symmetry.
Although $\hat{C}$ can be defined as either unitary or anti‑unitary in different conventions, here we treat it as a linear (unitary) operator.

The action of charge conjugation $\hat{C}$ forms a discrete symmetry that is isomorphic to $\mathbb Z_{2}$ group.
It interchanges $\,+q \leftrightarrow -q$ globally, but has no Gauss law, just a global symmetry (may be broken).
The projector is $\frac12(\mathbf 1+\hat C)$, which keeps $\hat C$-even part.
It forbids nothing charged, just relates states.
This shows that the charge conjugation $\hat{C}$ is a useful auxiliary symmetry, but it is not what enforces gauge-invariance.

\paragraph{(1) Single-Particle.}
At single-particle level, the charge conjugation $\hat{C}$ is a unitary operator.
It transforms each particle $\lvert P\rangle$ into its antiparticle $\lvert \bar{P} \rangle$ and vice versa:
$$
\hat{C}\,\lvert P\rangle \;=\; \lvert \bar{P}\rangle,
\quad
\hat{C}\,\lvert \bar{P}\rangle \;=\; \lvert P\rangle,
$$
where $\lvert P\rangle$ has net charge $+Q$ and $\lvert \bar{P}\rangle$ has net charge $-Q$. This indicates that $\hat{C}$ swaps charges, i.e., $\hat{C}: \lvert +Q\rangle\leftrightarrow \lvert -Q\rangle$.
	
In fact, $\hat{C}$ reverses all IQNs (electric charge, flavor, color, baryon number, etc.), 
$$
\hat{C}: \text{IQNs} \leftrightarrow - \text{IQNs}.
$$
	
But $\hat{C}$ does not affect external degree of freedoms (spin, momentum, position, etc.) unless combined with other discrete symmetries like $\hat{P}$ or $\hat{T}$.

Under charge conjugation $\hat{C}$, we can write the basis of a single-particle space as 
$$
\{\lvert P\rangle,\,\lvert \bar{P}\rangle\}.
$$
where $\hat{C}$ swaps these states.

\paragraph{(2) Particle-Antiparticle Pair and Packaged Bell States.}
For a particle-antiparticle system, the product state basis is:
$$
\{
\lvert P\rangle_A\,\lvert P\rangle_B,
\lvert \bar{P}\rangle_A\,\lvert \bar{P}\rangle_B,
\lvert P\rangle_A\,\lvert \bar{P}\rangle_B,
\lvert \bar{P}\rangle_A\,\lvert P\rangle_B
\}.
$$
The corresponding net charges of the product states are:
$$
\{+2Q, -2Q, 0, 0\}
$$

According to Corollary \ref{COL:DiscreteSymmetryAndEntanglement}, superselection forbids mixing different total charges $+2Q$ with $-2Q$.
Only the net-zero charge sector is allowed in a particle-antiparticle pair.
More specifically:
\begin{enumerate}
	\item Forbidden superposition ($+2Q$ vs. $-2Q$):
	Superposing $\lvert P\rangle_A \lvert P\rangle_B$ with $\lvert \bar{P}\rangle_A \lvert \bar{P}\rangle_B$ would mix total charge $+2Q$ and $-2Q$, which is blocked by the charge superselection rule. Thus, it is not physically realizable.
	
	\item Physically allowed superposition (0 vs. 0):
	Superposing $\lvert P\rangle_A \,\lvert \bar{P}\rangle_B$ with $\lvert \bar{P}\rangle_A \,\lvert P\rangle_B$ would mix total charge $0$ and $0$, which obeys the charge superselection rule.
	Thus, it is physically realizable.
	We have packaged entangled states:	
	\begin{equation}\label{EQ:PsiPES}
		\bigl\vert \Psi_P^\pm\bigr\rangle_{AB}
		\;=\;
		\frac{1}{\sqrt{2}}\,
		\Bigl(
		\lvert P\rangle_A \,\lvert \bar{P}\rangle_B 
		\;\pm\;
		\lvert \bar{P}\rangle_A \,\lvert P\rangle_B
		\Bigr).
	\end{equation}
	
	Under tensor product $\hat{C}\otimes \hat{C}$, we have:
	\[
	\hat{C} \otimes \hat{C}
	\;:\;
	\lvert P\rangle_A\,\lvert \bar{P}\rangle_B
	\;\leftrightarrow\;
	\lvert \bar{P}\rangle_A\,\lvert P\rangle_B,
	\]
	and
	\[
	\hat{C} \otimes \hat{C} \,\bigl\vert \Psi_P^\pm\bigr\rangle_{AB} \;=\; \pm\,\bigl\vert \Psi_P^\pm\bigr\rangle_{AB}.
	\]
	We see that packaged entangled states $\bigl\vert \Psi_P^\pm\bigr\rangle_{AB}$ become the eigenstates of $\hat{C}\otimes \hat{C}$ with eigenvalues $\pm1$.
\end{enumerate}

\begin{example}[$\mathbb Z_{2}$ charge conjugation]
	\label{EX:Z2CPackaging}
	\leavevmode
	\begin{enumerate}
		\item Symmetry group.
		$G = \mathbb Z_{2}^{\;C}=\{1,\hat C\}$, with a unitary charge-conjugation
		operator $\hat C$ that flips every additive IQN
		(electric charge, baryon number, flavor, color index, $\cdots$).
		
		\item One-particle carrier space.
		For a single complex field with charge magnitude $q_{0}$ one has the
		occupation‑number basis
		\[
		|{\rm p}\rangle \equiv a^\dagger|0\rangle,
		\qquad
		|\bar{\rm p}\rangle \equiv b^\dagger|0\rangle,
		\quad
		\hat C\,|{\rm p}\rangle = |\bar{\rm p}\rangle,\;
		\hat C\,|\bar{\rm p}\rangle = |{\rm p}\rangle .
		\]
		Introduce the C‑even / C‑odd combinations
		\[
		|{\psi_+}\rangle
		:=\tfrac{1}{\sqrt2}(|{\rm p}\rangle+|\bar{\rm p}\rangle),\qquad
		|{\psi_-}\rangle
		:=\tfrac{1}{\sqrt2}(|{\rm p}\rangle-|\bar{\rm p}\rangle),
		\]
		which satisfy $\hat C|\psi_\pm\rangle=\pm|\psi_\pm\rangle$.
		Thus
		\[
		V_+ := {\rm span}\{|{\psi_+}\rangle\},
		\qquad
		V_- := {\rm span}\{|{\psi_-}\rangle\}
		\]
		are the two one‑dimensional irreps of $\mathbb Z_{2}^{\;C}$.
		
		\item Isotypic decomposition.
		Including momentum and spin, the full one‑particle Hilbert space factorizes as
		\[
		\mathcal H_{\text{1‑p}}
		\;=\;
		\bigl(V_+\otimes\mathcal M_+\bigr)
		\;\oplus\;
		\bigl(V_-\otimes\mathcal M_-\bigr),
		\]
		where the multiplicity spaces
		$\mathcal M_\pm$ carry the usual external labels $\mathbf p,\sigma$.
		
		\item Block diagonality.
		For every $\hat C$‑invariant observable
		$O\in\mathcal O_{\rm phys}^{\;(\hat C)}$ one has, by the
		Symmetry Packaging Theorem,
		\[
		O
		\;=\;
		\mathbf 1_{V_{+}}\otimes O_{+}
		\;\oplus\;
		\mathbf 1_{V_{-}}\otimes O_{-},
		\qquad
		O_{\pm}\in\mathcal B(\mathcal M_{\pm}).
		\]
		Hence an operator that respects charge conjugation can neither
		flip the sign of the charge (that would map $V_\pm\to V_\mp$) nor
		distinguish particle from antiparticle inside a fixed irrep.
		
		\item Physical consequences.
		\begin{itemize}
			\item Confinement of charge sign.
			A $\hat C$‑respecting interaction may move the particle in
			momentum space or rotate its spin, but it always acts with the same scalar factor on $|\rm p\rangle$ and $|\bar{\rm p}\rangle$.
			The sign of the charge is therefore packaged.
			
			\item $\mathbb Z_{2}$ superselection.
			Because $V_{+}$ and $V_{-}$ belong to inequivalent irreps,
			coherent superpositions of even and odd states are forbidden by
			$\hat C$ superselection mirroring the usual charge parity rule in particle physics.
			
			\item Entangled particle-antiparticle pairs.
			For a two‑mode system ($A,B$) the symmetry
			$\hat C_A\otimes\hat C_B$ packages the Bell states
			$
			|\Psi_P^\pm\rangle,\;
			|\Phi_P^\pm\rangle
			$
			of Eq. (2.13) into $\pm1$ eigenvectors.
			Any $\hat C$‑respecting	detector collapses the whole pair.
		\end{itemize}
	\end{enumerate}

	Thus $\mathbb Z_{2}$ charge conjugation furnishes a minimal, discrete example	of internal symmetry packaging: the sign of every additive charge is confined inside a one‑particle or multi‑particle C‑eigenblock, fully analogous to the confinement of color in $\mathrm{SU}(3)$ but with Abelian group structure.
\end{example}

Thus, charge conjugation yields these internal, packaged entangled states within the net-zero charge sector.
They show how gauge or superselection constraints guide us to physically meaningful superpositions.

It should be emphasized that the packaged entangled states Eq.~(\ref{EQ:PsiPES}) induced by the charge conjugation $\hat{C}$ possess the strongest packaged entanglement because they entangled with every IQNs.

\subsubsection{Internal: Fermion Parity}
\label{SEC:Z2FermionParity}

Fermion parity is a special instance of a $\mathbb Z_2$ symmetry that often lives only projectively on the true symmetry group.
Equivalently, the physical symmetry is an extension
\[
1 \;\longrightarrow\;\mathbb Z_2=\{1,P_F\}
\;\longrightarrow\;
\widetilde G
\;\longrightarrow\;
G
\;\longrightarrow\;
1,
\]
where $P_F^2=1$ is the operator that counts fermion number mod 2.
In the language of $G$, states transform projectively:
\[
U(g)\,U(h)
\;=\;\omega(g,h)\;U(gh), 
\qquad
\omega(g,h)\in\{\pm1\},
\]
with $\omega$ a $\mathbb Z_2$-valued 2-cocycle (the Schur multiplier).

\paragraph{(1) Fermion Parity Operator \& Superselection.}

On any fermionic Fock space we have
\[
(-1)^{\hat N_f}
\;=\;
\prod_j
\bigl(1 - 2\,c_j^\dagger c_j\bigr),
\]
whose two one‐dimensional irreps (even/odd) split
\[
\mathcal H
\;=\;
\mathcal H_{\rm even}
\;\oplus\;
\mathcal H_{\rm odd}.
\]
The projectors
\[
P_{\rm even}
=\tfrac12\bigl(1 + (-1)^{\hat N_f}\bigr),
\quad
P_{\rm odd}
=\tfrac12\bigl(1 - (-1)^{\hat N_f}\bigr)
\]
obey $P^2=P$, and any physical (parity‐even) operator commutes with $(-1)^{\hat N_f}$, acting block‐diagonally on $\mathcal H_{\rm even}\oplus\mathcal H_{\rm odd}$.

\paragraph{(2) Twisted Projectors (Peter-Weyl).}

Viewed as a finite‐group symmetry $G=\{e,\;P_F\}\cong\mathbb Z_2$, the irreducible characters are
$\chi_{+}(e)=\chi_{-}(e)=1$, $\chi_{\pm}(P_F)=\pm1$.
Hence
\[
P^{(\epsilon)}
\;=\;
\frac1{2}\sum_{g\in G}\chi_{\epsilon}^*(g)\,U(g)
\;=\;
\tfrac12\bigl(1 + \epsilon\,(-1)^{\hat N_f}\bigr),
\qquad
\epsilon=\pm1,
\]
which recovers exactly $P_{\rm even/odd}$.

If the underlying symmetry were only $G$ but realized projectively, one would instead write
\[
\Pi_{\rm proj}
\;=\;
\frac1{|G|}\sum_{g\in G}\omega(g)\,U(g),
\]
with $\omega(g)\equiv\omega(g,g^{-1})$ the appropriate 2-cocycle insertion.

\paragraph{(3) Packaged Entanglement under Fermion Parity.}

Split two fermion modes $c_1,c_2$ into the basis
$\{|00\rangle,|01\rangle,|10\rangle,|11\rangle\}$.  
Total parity divides it into
\[
\mathcal H_{\rm even}
=\operatorname{span}\{|00\rangle,|11\rangle\},
\quad
\mathcal H_{\rm odd}
=\operatorname{span}\{|01\rangle,|10\rangle\}.
\]
A canonical even packaged entangled state is
\[
\ket{\Psi_{\rm even}}
\;=\;
\tfrac1{\sqrt2}\bigl(\ket{00}+\ket{11}\bigr),
\]
which lies entirely in $\mathcal H_{\rm even}$, has Schmidt rank 2, and cannot be split by any parity‐even operator.

\begin{example}[Majorana qubits]
	In topological superconductors, Majorana zero modes $\gamma_{1,2}$ realize a nonlocal fermionic mode $d=(\gamma_1+i\gamma_2)/2$.
	The two ground states $|0\rangle,|1\rangle=d^\dagger|0\rangle$ differ by parity.
	Fusion or tunneling operations preserve $(-1)^{\hat N_f}$, so quantum information encoded in the parity qubit is packaged and robust against parity-even perturbations.
\end{example}

\subsubsection{External: Parity $\hat{P}$}

In contrast to charge conjugation $\hat{C}$, both parity $\hat{P}$ and time-reversal $\hat{T}$ flip external two-level, but leave the IQNs unchanged.

\paragraph{(1) Parity $\hat{P}$.}
In QFT, parity $\hat{P}$ is an operator invented to perform a spatial mirror operation, i.e., inverts space: $\mathbf{x} \leftrightarrow -\mathbf{x}$.\cite{LeeYang1956,Wu1957}

At single-particle level, the parity $\hat{P}$ is a pure unitary operator.
In addition to the spatial coordinates $\mathbf{x}$, $\hat{P}$ also flips momentum $\mathbf{p} \leftrightarrow -\mathbf{p}$ because momentum is related to spatial coordinates.
Consider a single-particle with two-level momenta.
Let us label right-moving vs. left-moving single-particle states as $\lvert +\rangle$, $\lvert -\rangle$, we have
$$
\hat{P}\,\lvert +\rangle \;=\; \lvert -\rangle,
\quad
\hat{P}\,\lvert -\rangle \;=\; \lvert +\rangle.
$$
Because spin is similar to orbital angular momentum, $\hat{P}$ also flips (or preserves) spin states (depending on spin's vector or pseudovector nature).
In many formulations (especially in the nonrelativistic limit), parity leaves the spin unchanged (up to a phase).
But here we are considering a context where parity acts non-trivially on spin (a flip operation on an external two‑level system).
Let us consider a single-particle with spin-$\frac{1}{2}$ and
label the up vs. down states as $\lvert \uparrow\rangle$, $\lvert \downarrow\rangle$, we have
$$
\hat{P}\,\lvert \uparrow\rangle \;=\; \lvert \downarrow\rangle,
\quad
\hat{P}\,\lvert \downarrow\rangle \;=\; \lvert \uparrow\rangle.
$$

But $\hat{P}$ leaves IQNs (like electric charge or baryon number) unchanged, as they do not transform under spatial inversion.

\subsubsection{External: Time-Reversal $\hat{T}$}

\paragraph{(2) Time-Reversal $\hat{T}$.}
In QFT, time-reversal $\hat{T}$ is an operator invented to perform a time-reversal operation, i.e., inverts time: $t \leftrightarrow -t$.\cite{Wigner1932}

In addition to the direction (flow) of time $t$, $\hat{T}$ flips momentum $\mathbf{p} \leftrightarrow -\mathbf{p}$ because momentum is related to time $t$.

At single-particle level, the time-reversal $\hat{T}$ is an anti‑unitary operator.
This means that $\hat{T}$ includes both a unitary spin/coordinate transformation and a complex-conjugation operation.
There is subtlety for single spin-$\tfrac12$.  
On a single spin-$\tfrac12$ space, $\hat{T}$ satisfies $\hat{T}^2 = -\,\mathbf{I}$, which implies that we cannot have a single spin-$\tfrac12$ state that is strictly an eigenstate of $\hat{T}$ with a real eigenvalue.
This is the essence of Kramers degeneracy: every level is at least twofold degenerate.
However, we can define states that are $\pm1$ eigenstates under $\hat{T}_A \otimes \hat{T}_B$ (see two-particle systems).
This yields the same Bell structure as above, which can exhibit $\hat{T}^2 = +\,\mathbf{I}$ on certain subspaces.
In this case, it is consistent to say certain bipartite states have $\pm1$ eigenvalues under $\hat{T}$.
In this sense, we say that $\hat{T}$ also flips spin/angular momentum $\lvert \uparrow\rangle \leftrightarrow \lvert \downarrow\rangle$.

But $\hat{T}$ leaves IQNs (electric charge $Q$, baryon number $B$, etc.) invariant, since these do not directly depend on the sign of time.

\paragraph{(3) Two-Particle Momentum/Spin Bell States.} Let us now move to the two-particle systems, where we need to consider the tensor product $\hat{P}_A \otimes \hat{P}_B$ and $\hat{T}_A \otimes \hat{T}_B$.

Tensor product $\hat{P}_A \otimes \hat{P}_B$ acts like a reflection on each local subsystem (in spin or momentum space).
This produces Bell-type entangled states with well-defined symmetry under spatial inversion, which are $\pm1$ eigenstates under $\hat{P}_A \otimes \hat{P}_B$.
Similarly, we can also define Bell-type entangled states that are $\pm1$ eigenstates under tensor product $\hat{T}_A \otimes \hat{T}_B$.

\begin{itemize}
	\item \textbf{Momentum:}	
	For a two-particle system of two-level momentum basis, the product state basis is:
	$$
	\{
	\lvert + \rangle_A\,\lvert + \rangle_B, ~
	\lvert - \rangle_A\,\lvert - \rangle_B, ~
	\lvert + \rangle_A\,\lvert - \rangle_B, ~
	\lvert - \rangle_A\,\lvert + \rangle_B
	\}.
	$$
	Both parity $\hat{P}$ and time reversal $\hat{T}$ flip momentum.
	Applying Corollary \ref{COL:DiscreteSymmetryAndEntanglement}, we again obtain four Bell states of momentum in a two-level momentum basis:
	\begin{align}\label{MomentumBellStates}
		\begin{aligned}
			&\lvert \Psi_m^{\pm} \rangle_{AB}
			\;=\;
			\tfrac{1}{\sqrt{2}}
			\Bigl(
			\lvert +\rangle_A \,\lvert -\rangle_B
			\;\pm\;
			\lvert -\rangle_A \,\lvert +\rangle_B
			\Bigr), \\
			&\lvert \Phi_m^{\pm} \rangle_{AB}
			\;=\;
			\tfrac{1}{\sqrt{2}}
			\Bigl(
			\lvert +\rangle_A \,\lvert +\rangle_B
			\;\pm\;
			\lvert -\rangle_A \,\lvert -\rangle_B
			\Bigr),
		\end{aligned}
	\end{align}
	which satisfy
	\[
	\hat{P}_A \otimes \hat{P}_B \,\lvert \Psi_m^{\pm} \rangle_{AB} \;=\; \pm\,\lvert \Psi_m^{\pm} \rangle_{AB},
	\quad
	\hat{P}_A \otimes \hat{P}_B \,\lvert \Phi_m^{\pm} \rangle_{AB} \;=\; \pm\,\lvert \Phi_m^{\pm} \rangle_{AB}.
	\]
	and
	\[
	\hat{T}_A \otimes \hat{T}_B\,\lvert \Psi_m^{\pm} \rangle_{AB} \;=\; \pm\,\lvert \Psi_m^{\pm} \rangle_{AB},
	\quad
	\hat{T}_A \otimes \hat{T}_B\,\lvert \Phi_m^{\pm} \rangle_{AB} \;=\; \pm\, \lvert \Phi_m^{\pm} \rangle_{AB}.
	\]
	
	\item \textbf{Spin:}	
	For a two-particle system of spin-$\frac{1}{2}$, the product state basis for spin is:
	$$
	\{
	\lvert \uparrow   \rangle_A\,\lvert \uparrow   \rangle_B, ~
	\lvert \downarrow \rangle_A\,\lvert \downarrow \rangle_B, ~
	\lvert \uparrow   \rangle_A\,\lvert \downarrow \rangle_B, ~
	\lvert \downarrow \rangle_A\,\lvert \uparrow   \rangle_B
	\}.
	$$
	Both parity $\hat{P}$ and time reversal $\hat{T}$ flip spin.
	Applying Corollary \ref{COL:DiscreteSymmetryAndEntanglement}, we again obtain four Bell states of spin:
	\begin{align}\label{SpinBellStates}
		\begin{aligned}
			&\lvert \Psi_s^{\pm} \rangle_{AB}
			\;=\;
			\tfrac{1}{\sqrt{2}}
			\Bigl(
			\lvert \uparrow\rangle_A \,\lvert \downarrow\rangle_B
			\;\pm\;
			\lvert \downarrow\rangle_A \,\lvert \uparrow\rangle_B
			\Bigr), \\
			&\lvert \Phi_s^{\pm} \rangle_{AB}
			\;=\;
			\tfrac{1}{\sqrt{2}}
			\Bigl(
			\lvert \uparrow \rangle_A \,\lvert \uparrow \rangle_B
			\;\pm\;
			\lvert \downarrow \rangle_A \,\lvert \downarrow \rangle_B
			\Bigr),
		\end{aligned}
	\end{align}
	which satisfy	
	\[
	\hat{P}_A \otimes \hat{P}_B \,\lvert \Psi_s^{\pm} \rangle_{AB} \;=\; \pm\,\lvert \Psi_s^{\pm} \rangle_{AB},
	\quad
	\hat{P}_A \otimes \hat{P}_B \,\lvert \Phi_s^{\pm} \rangle_{AB} \;=\; \pm\,\lvert \Phi_s^{\pm} \rangle_{AB}.
	\]
	and
	\[
	\hat{T}_A \otimes \hat{T}_B\,\lvert \Psi_s^{\pm} \rangle_{AB} \;=\; \pm\,\lvert \Psi_s^{\pm} \rangle_{AB},
	\quad
	\hat{T}_A \otimes \hat{T}_B\,\lvert \Phi_s^{\pm} \rangle_{AB} \;=\; \pm\,\lvert \Phi_s^{\pm} \rangle_{AB},
	\]
\end{itemize}

From above discussion, we see that entangled states
\[
\lvert \Psi_m^{\pm} \rangle_{AB},~
\lvert \Phi_m^{\pm} \rangle_{AB},~
\lvert \Psi_s^{\pm} \rangle_{AB},~
\lvert \Phi_s^{\pm} \rangle_{AB}
\]
are the $\pm1$ eigenstates of both parity $\hat{P}$ and time reversal $\hat{T}$.

Especially, in either momentum or spin space, time reversal $\hat{T}$ classifies certain bipartite states into $\pm1$ symmetry sectors and therefore show a neat link between $\hat{T}$ invariance and Bell-type entanglement.
Even though $\hat{T}$ is anti‑unitary, for these two-subsystem states (with integer total spin or balanced momentum pairs), it is meaningful to speak of $\hat{T}$-eigenstates.

Finally, we conclude that $\mathbb{Z}_2$ symmetry operators exemplify the simplest instance of Theorem~\ref{THM:ExistenceOfGAssociatedPackagedStates}: whenever a flip transformation has a nontrivial irrep (dimension $2$), it forces a 1D submodule to vanish and yields entangled states in that subspace.
If the transformation flips IQNs (charge), then the resulting states are necessarily Bell-like packaged entangled states.

\subsubsection{Combined $\hat{P}\hat{T}$, $\hat{C}\hat{P}$, $\hat{C}\hat{T}$, and $\hat{C}\hat{P}\hat{T}$}

We have seen how the discrete symmetries $\hat{C}$ (charge conjugation), $\hat{P}$ (parity), and $\hat{T}$ (time reversal) each generates entangled states (Bell-like states) in the bipartite systems.
Now we would like to go a step further to combine these operators into products such as $\hat{P}\hat{T}$, $\hat{C}\hat{P}$, $\hat{C}\hat{T}$, and $\hat{C}\hat{P}\hat{T}$.\cite{Luders1957,Hooft1976,Banks1991,KimCarosi2010}
The combined operators $\hat{C}\hat{P}$, $\hat{C}\hat{T}$, and $\hat{C}\hat{P}\hat{T}$ act on both external (spin, momentum) and internal (charge) quantum numbers simultaneously.
We then obtains hybrid-packaged entangled states that cannot be factorized into a pure spin part times a pure charge part.

The hybrid-packaged entangled states are especially relevant for analyzing CP or CPT violation in meson systems (e.g., kaons $\lvert K^0\rangle,\lvert\bar{K}^0\rangle$), where one forms CP eigenstates with definite entanglement structure.
For example, we can use the hybrid-packaged entangled states to
classify states by $\pm1$ (or sometimes complex phases) under each discrete transformation, 
analyzing discrete-symmetry tests or violations (e.g., CP violation in neutral mesons),  
and enforcing superselection rules (e.g., total charge) while still forming valid entangled superpositions.

Below we will discuss each of these combined symmetries and show how they produce hybridized entangled states.

\paragraph{(1) $\hat{P}\hat{T}$ Symmetry.}
From earlier discussions, we know that $\hat{P}$ inverts space ($\mathbf{x} \leftrightarrow -\mathbf{x}$, $\mathbf{p} \leftrightarrow -\mathbf{p}$), and $\hat{T}$ reverses time ($t \leftrightarrow -t$, again $\mathbf{p} \leftrightarrow -\mathbf{p}$, flips spin, etc.).
Acting twice on external two-level sometimes leads to a trivial net action on certain subspaces (e.g., $\hat{P}\hat{T} \approx \hat{I}$ on spin-0 or momentum-symmetric states).
Symbolically:
\[
\hat{P}\,\hat{T}\,\lvert \Psi_s^{\pm} \rangle_{AB} 
\;=\; \pm \,\hat{P}\,\lvert \Psi_s^{\pm} \rangle_{AB}
\;=\; \lvert \Psi_s^{\pm} \rangle_{AB},
\quad
\hat{P}\,\hat{T}\,\lvert \Phi_s^{\pm} \rangle_{AB}
\;=\;\pm \,\hat{P}\,\lvert \Phi_s^{\pm} \rangle_{AB}
\;=\;\lvert \Phi_s^{\pm} \rangle_{AB},
\]
and similarly for momentum Bell states $\lvert \Psi_m^{\pm} \rangle_{AB},\lvert \Phi_m^\pm\rangle_{AB}$.
Thus, certain two-particle states become invariant (up to signs) under $\hat{P}\hat{T}$.
In some frameworks (particularly integer-spin subspaces), $\hat{P}\hat{T}$ can act as the identity.

\paragraph{(2) $\hat{C}\hat{P}$ and $\hat{C}\hat{T}$ Symmetry.}
$\hat{C}\hat{P}$ symmetry flips both IQNs and spatial coordinates (momenta/spins).
Similarly, $\hat{C}\hat{T}$ symmetry also flip both IQNs and spatial coordinates (momenta/spins).
As a result, these two combined symmetric operations mix internal (packaged) and external (un-packaged) two-level systems, therefore, allow one to define states with definite eigenvalues under $\hat{C}\hat{P}$ or $\hat{C}\hat{T}$. 
Here we will focus on $\hat{C}\hat{P}$ for illustration purpose.
The same reasoning applies to $\hat{C}\hat{T}$ by replacing the parity $\hat{P}$ with time reversal $\hat{T}$.

For a two-particle system, if we consider the IQNs of $\lvert P\rangle, \lvert\bar{P}\rangle$, external spin-$\tfrac12$ states $\lvert \uparrow\rangle, \lvert\downarrow\rangle$, and external momentum states $\lvert+\rangle, \lvert-\rangle$, then we obtain hybrid-packaged entangled states:
\begin{align}\label{EQ:HybridPES}
	\begin{aligned}
		&\lvert \Psi_h^{\pm} \rangle_{AB}
		\;=\;
		\frac{1}{\sqrt{2}}
		\Bigl(
		\lvert P, \uparrow, +\rangle_A \,\lvert \bar{P}, \downarrow, -\rangle_B
		\;\pm\;
		\lvert \bar{P}, \downarrow, -\rangle_A \,\lvert P, \uparrow, +\rangle_B
		\Bigr), \\
		&\lvert \Phi_h^{\pm} \rangle_{AB}
		\;=\;
		\frac{1}{\sqrt{2}}
		\Bigl(
		\lvert P, \uparrow, +\rangle_A \,\lvert \bar{P}, \uparrow, +\rangle_B
		\;\pm\;
		\lvert \bar{P}, \downarrow, -\rangle_A \,\lvert P, \downarrow, -\rangle_B
		\Bigr).
	\end{aligned}
\end{align}
These are the eigenstates of $(\hat{C}\hat{P})_A \otimes (\hat{C}\hat{P})_B$ and $(\hat{C}\hat{T})_A \otimes (\hat{C}\hat{T})_B$ with eigenvalue $\pm 1$, i.e.,
\[
(\hat{C}\hat{P})_A \otimes (\hat{C}\hat{P})_B \,\lvert \Psi_h^{\pm} \rangle_{AB}
\;=\;
\pm\,\lvert \Psi_h^{\pm} \rangle_{AB},
\quad
(\hat{C}\hat{P})_A \otimes (\hat{C}\hat{P})_B \,\lvert \Phi_h^{\pm} \rangle_{AB}
\;=\;
\pm\,\lvert \Phi_h^{\pm} \rangle_{AB}
\]
and
\[
(\hat{C}\hat{T})_A \otimes (\hat{C}\hat{T})_B \,\lvert \Psi_h^{\pm} \rangle_{AB}
\;=\;
\pm\,\lvert \Psi_h^{\pm} \rangle_{AB},
\quad
(\hat{C}\hat{T})_A \otimes (\hat{C}\hat{T})_B \,\lvert \Phi_h^{\pm} \rangle_{AB}
\;=\;
\pm\,\lvert \Phi_h^{\pm} \rangle_{AB}.
\]
These show that $\hat{C}\hat{P}$ and $\hat{C}\hat{T}$ mixed IQNs, external spin, and external momentum in hybridized packaged entangled states.

\begin{example}[Neutral-Meson Systems ($K^0$, $B^0$, etc.)]
	An important application arises in neutral mesons, such as kaons $\lvert K^0, \uparrow\rangle$ vs.\ $\lvert \bar{K}^0,\downarrow\rangle$. In analyzing CP violation or correlated decays, we often looks at eigenstates of $\hat{C}\hat{P}$.
	\begin{enumerate}
		\item For a single meson:
		\[
		\hat{C}\hat{P} \,\lvert K^0,\uparrow\rangle = e^{i\alpha}\,\lvert \bar{K}^0,\downarrow\rangle,
		\quad
		\hat{C}\hat{P} \,\lvert \bar{K}^0,\downarrow\rangle = e^{-i\alpha}\,\lvert K^0,\uparrow\rangle,
		\]
		This leads to the well-known CP-eigenstates
		\[
		\lvert K_{1,2}\rangle = \tfrac{1}{\sqrt{2}}
		\bigl(\lvert K^0,\uparrow\rangle \pm \lvert \bar{K}^0,\downarrow\rangle\bigr).
		\]
		
		\item For two mesons (e.g., in $\phi \to K^0\bar{K}^0$ decays), we can build correlated states that are eigenstates of $(\hat{C}\hat{P})_A \otimes (\hat{C}\hat{P})_B$.
		In that sense, they become packaged or hybrid entangled states in the internal (strangeness) degree of freedom and can exhibit correlated decays that test CP symmetry.
	\end{enumerate}
	
\end{example}

\paragraph{(3) $\hat{C}\hat{P}\hat{T}$ Symmetry.}
In quantum field theory, CPT is a fundamental combined symmetry guaranteed by the CPT theorem: any Lorentz-invariant local QFT must be invariant under $\hat{C}\hat{P}\hat{T}$.
In this case, one indeed consider the particles and antiparticles as an integrated system, which is invariant under the flip of charge $\hat{C}: Q \mapsto -Q$, spatial coordinates $\hat{P}: \mathbf{x} \mapsto -\mathbf{x}$, and time $\hat{T}: t \mapsto -t$.

However, here we are dealing with packaged entangled states, $\hat{P}\hat{T}$ may act trivially (or as an overall reflection) on certain subspaces.
Consequently, $\hat{C}\hat{P}\hat{T}$ effectively reduces to $\hat{C}$. Regardless, for analyzing discrete-symmetry properties in meson-antimeson pairs, or more generally for exploring whether a process respects or breaks CPT, one can construct states with well-defined $\hat{C}\hat{P}\hat{T}$ transformation properties. If time reversal $\hat{T}$ is anti‑unitary, extra care in defining eigenstates is needed, but in multiparticle systems it can still be consistent to speak of $\pm 1$ CPT eigenvalues in certain subspaces.

Thus, the combined symmetries ($\hat{P}\hat{T}$, $\hat{C}\hat{P}$, $\hat{C}\hat{T}$, $\hat{C}\hat{P}\hat{T}$) yield a rich structure of hybrid entangled states, which can be systematically constructed by how each operator acts on the internal and external DOFs.
This construction underpins many discrete-symmetry tests in particle physics, such as detecting CP violation in neutral-meson decays, or examining whether CPT invariance might be violated in exotic scenarios.

\subsection{Dihedral Group}

In this subsection, we show that the symmetry packaging we have built also works well for dihedral group $D_n$.
Specifically:

\paragraph{(1) The dihedral group and its irreps.}

The dihedral group $D_n$ of order $2n$ is generated by a rotation $r$ of order $n$ and a reflection $s$ with
$$
r^n = e,\quad s^2=e,\quad s\,r\,s = r^{-1}.
$$
Its irreps over $\mathbb C$ are:
\begin{itemize}
	\item Four 1‑dimensional irreps (for all $n\ge2$) when $n$ is even, or two when $n$ is odd, characterized by how $r$ and $s$ act by $\pm1$.
	
	\item $\lfloor (n-1)/2\rfloor$ two‑dimensional irreps, carried by the sin/cos pairs under rotation, with $s$ swapping them.
\end{itemize}
Denote the set of irreps by $\widehat{D_n}=\{\chi_1,\dots,\chi_{k},\rho_1,\dots,\rho_{m}\}$ with $\chi_i$ one‑dimensional and $\rho_j$ two‑dimensional.

\paragraph{(2) Packaging spaces for $D_n$.}

By Theorem I (finite/compact packaging) the full Hilbert space carrying a unitary $D_n$-action decomposes as
$$
\mathcal H
\;\cong\;
\Bigl(\bigoplus_{i=1}^{k} V_{\chi_i}\otimes M_{\chi_i}\Bigr)
\;\oplus\;
\Bigl(\bigoplus_{j=1}^{m} V_{\rho_j}\otimes M_{\rho_j}\Bigr),
$$
where each $V_{\chi_i}\cong\mathbb C$ and $V_{\rho_j}\cong\mathbb C^2$.
These are the internal packages.
If you also have an external dihedral action (say on orbital labels), you’d similarly decompose by its irreps into external packages and then form hybrid packages $V_\lambda\otimes W_\mu$.

\paragraph{(3) Projectors and Packaged Subspaces.}

The projector onto the $\lambda$-irrep is
$$
P_\lambda
\;=\;
\frac{d_\lambda}{|D_n|}\sum_{g\in D_n}\chi_\lambda^*(g)\,U(g),
$$
where $d_\lambda$ is the irrep’s dimension.
If only a subset $\Gamma_{\rm allowed}\subset\widehat{D_n}$ (e.g., the trivial 1‑dimensional irrep, or a particular 2‑dimensional doublet) is allowed by physics, then
$$
P_{\rm hyb}
=\sum_{\lambda\in\Gamma_{\rm allowed}}P_\lambda,
\qquad
\mathcal H_{\rm hyb}=\operatorname{Im} P_{\rm hyb}.
$$

\paragraph{(4) Packaged Entanglement in a $D_n$‑Symmetric System.}

Everything from our previous section applies:
\begin{itemize}
	\item Packaged Schmidt Rank:
	decompose $\mathcal H_{\rm pkg}\subset\mathcal H_A\otimes\mathcal H_B$ and count nonzero Schmidt coefficients.
	
	\item Packaged Entropy:
	form $\rho=P_{\rm hyb}|\Psi\rangle\langle\Psi|P_{\rm hyb}$, trace out $B$, and compute $- \operatorname{Tr}(\rho_A\log\rho_A)$.
\end{itemize}

\begin{example}[Internal $D_4$ packaging: Square‑Planar $d$‑Orbitals]
	A classic realization of a $D_4$-action on IQNs is the splitting of the five $d$‑orbitals of a transition‐metal ion in a square‑planar ligand field (point group $D_{4h}$, dihedral part $D_4$).
	\begin{itemize}
		\item Symmetry group:
		the square‐planar ligand field has (at least) the point‐group
		$$
		D_{4h}\;\cong\;C_{4v}\times C_i
		$$
		acting on the five $d$‐orbitals
		\begin{table}[h]
			\centering
			\caption{$D_4$ Symmetry packaging on square‑planar $d$‑orbitals}
			\begin{tabular}[hbt!]{|p{3cm}|p{3cm}|p{3cm}|}
				\hline\hline
				Orbital           & Irrep of $D_4$  & Dimension \\
				\hline
				$d_{z^2}$         & $A_1$ (trivial) & 1         \\
				$d_{x^2-y^2}$     & $B_1$           & 1         \\
				$d_{xy}$          & $B_2$           & 1         \\
				$(d_{xz},d_{yz})$ & $E$ (doublet)   & 2         \\
				\hline
			\end{tabular}
		\end{table}
		
		\item Irreps:
		Four one‑dimensional irreps $A_1,B_1,B_2,B_2'$ and the two‑dimensional $E$.
		
		\item Packaging:
		Packaging then says:
		\begin{itemize}
			\item The internal subspace $V_{A_1}$ (the $d_{z^2}$ orbital) is a one‑dimensional package.
			
			\item $V_E$ (the $\{d_{xz},d_{yz}\}$ pair) is a two‑dimensional package.
		\end{itemize}
		Any ligand or electron‐transfer operator (crystal‐field Hamiltonian or perturbation) commuting with the full $D_4$ symmetry must act block‑diagonally on these orbital packages (preserve each of these four subspaces):	
		$$
		\mathbf1_{V_{A_1}}\otimes M_{A_1}
		\;\oplus\;
		\mathbf1_{V_{B_1}}\otimes M_{B_1}
		\;\oplus\;
		\mathbf1_{V_{E}}\otimes M_{E}
		\;\oplus\;\cdots
		$$
		No symmetry‐respecting perturbation can split, say, $d_{xz}$ from $d_{yz}$ without also mixing with the entire 2‑dimensional $E$-block.
		In other words, the 2‐dimensional $E$ block cannot be split (you cannot, by a $D_{4h}$‐symmetric perturbation, isolate $d_{xz}$ without dragging $d_{yz}$ along).
	\end{itemize}	
\end{example}

\section{Symmetry Packaging under Compact Groups}
\label{SEC:SymmetryPackagingOnCompactLieGroups}

Compact groups such as $\mathrm{U}(1)$ and $\mathrm{SU}(N)$ possess a rich representation theory \cite{Yang1954,Gross1973,Weyl1929,GellMann1961}.
In gauge theories, the IQNs of single-particle excitations are packaged as inseparable units by virtue of their irreducible transformation properties.
As a result, multiparticle states inherit a structure that reflects the underlying representation theory of the gauge group $G$ \cite{Utiyama1956,Cartan1913}.

In this section, we briefly review the main results of the representation-theoretic decomposition and explain how projection operators play a leading role in identifying the physically relevant (e.g., color-singlet) subspaces that underlie phenomena such as confinement and superselection.
We treat Abelian and non-Abelian cases separately.

\subsection{Abelian Group: $\mathrm{U}(1)$}

For the Abelian gauge group $\mathrm{U}(1)$, every irrep is one‐dimensional and can be labeled by a continuous or discrete charge $q$ (integer in many physical settings) \cite{Weyl1929}.
Specifically, one has
\[
\rho_q\bigl(e^{i\theta}\bigr) \;=\; e^{\,i\,q\,\theta}.
\]
In standard QED, for instance, $q$ takes integer values (in units of the elementary charge).
By the Peter‐Weyl theorem, the full Hilbert space decomposes as
\begin{equation}\label{EQ:U(1)Decomposition}
	\mathcal{H}
	\;\cong\;
	\bigoplus_{q \,\in\, \mathbb{Z}} \mathcal{H}_q,
\end{equation}
where $\mathcal{H}_q$ is the subspace of states carrying net $\mathrm{U}(1)$ charge $q$.

\paragraph{(1) Single‐Particle Representations.}
Suppose a single‐particle field operator $\hat{\psi}(x)$ has charge $q$. Under a local gauge transformation $g(x) = e^{i\,\alpha(x)}$, it transforms by
\[
\hat{\psi}(x) \;\mapsto\;
e^{\,i\,q\,\alpha(x)}\,\hat{\psi}(x).
\]
Because each irrep of $\mathrm{U}(1)$ is one‐dimensional, the phase $e^{\,i\,q\,\alpha(x)}$ is an inseparable block.
One cannot factor out or split the charge $q$ among multiple sub-pieces within the same operator, e.g., you cannot split an electron’s charge $-1$ into fractional parts.
Formally, for the single‐particle creation operator $\hat{a}^\dagger_{q}$, one writes:
\[
U\bigl(e^{\,i\theta}\bigr)\,\hat{a}^\dagger_{q}\,U\bigl(e^{\,i\theta}\bigr)^{-1}
\;=\;
e^{\,i\,q\,\theta}\,\hat{a}^\dagger_{q},
\]
illustrating that $\hat{a}^\dagger_{q}$ itself furnishes a 1D representation with charge $q$.

\paragraph{(2) Multi‐Particle States and Additive Charges.}
For multi‐particle excitations, let us first consider the product of field operators:
\[
\hat{\psi}^\dagger_{q_1}(x_1)\;\hat{\psi}^\dagger_{q_2}(x_2)
\;\cdots\;\hat{\psi}^\dagger_{q_n}(x_n),
\]
where each operator $\hat{\psi}^\dagger_{q_i}(x_i)$ transforms as
\[
\hat{\psi}^\dagger_{q_i}(x_i) \mapsto e^{i q_i \alpha(x_i)}\, \hat{\psi}^\dagger_{q_i}(x_i)
\]
and picks up a phase $e^{\,i\,q_i\,\alpha(x_i)}$ under $g(x)=e^{\,i\,\alpha(x)}$.

Since $\mathrm{U}(1)$ is Abelian, the total gauge transformation only leads to a product of phases
\[
\prod_{i=1}^n
e^{\,i\,q_i\,\alpha(x_i)}
\;=\;
e^{\,i\,\bigl(q_1+q_2+\cdots+q_n\bigr)\,\alpha}.
\]
So the net charge of the multi‐particle state is
\[
Q_\text{tot} = q_1 + q_2 + \cdots + q_n.
\]

We see that, in a non‐confining theory like QED, states with nonzero net charge can in principle appear as free asymptotic states.
However, the superselection rules prohibit the mixing of states from different charge sectors $\mathcal{H}_q$ as shown in Eq.(\ref{EQ:U(1)Decomposition}).

\paragraph{(3) Projecting onto a Definite Charge Sector.}

$\mathrm{U}(1)$ is a continuous Lie group that acts via a local phase factor $e^{i\alpha(\mathbf x)}$ at each point.
Physical states satisfy the Gauss law $\hat G(\mathbf x)\,\ket{\mathrm{phys}} = 0$, where $\hat G(\mathbf x)=\nabla\!\cdot\!\mathbf E - \hat\rho(\mathbf x)$ is the Gauss operator.
The gauge projector can be written as
$$
\Pi_{\rm phys}
= \int\!\mathcal D\alpha \,
\exp\!\biggl[i\!\int\!{\rm d}^3x\,\alpha(\mathbf x)\,\hat G(\mathbf x)\biggr],
$$
which projects onto locally neutral (gauge-singlet) states and filters out any configuration with unscreened charge.

Because $\mathrm{U}(1)$ is a continuous group, one can define the packaged projection operator $P_q$ using the normalized Haar measure on $\mathrm{U}(1)$, i.e.,
\[
P_q
\;=\;
\int_{0}^{2\pi}
\frac{d\theta}{2\pi}
\;e^{-\,i\,q\,\theta}
\;U\!\bigl(e^{\,i\,\theta}\bigr),
\]
where $U\bigl(e^{\,i\,\theta}\bigr)$ is the unitary operator that implements the global gauge transformation $e^{\,i\,\theta}\in \mathrm{U}(1)$.
From the physical meaning of $P_q$, we should have
\[
P_q^2=P_q 
\text{  and  }
P_q\,\mathcal{H}
\;=\;\mathcal{H}_q.
\]
Hence, $\forall ~ \vert \psi\rangle \in \mathcal{H}$, the component $\vert \psi_q\rangle = P_q\,\vert \psi\rangle$ completely lies in the charge‐$q$ sector.
This superselection structure ensures that a physical state must carry a definite net charge. Cross‐sector interference is disallowed.

Thus, for an Abelian gauge group $\mathrm{U}(1)$, each single-particle creation operator $\hat a^\dagger\_q$ furnishes a one-dimensional irrep labeled by charge $q$, which cannot be further decomposed.
Multi‐particle states acquire net charge $Q_\text{tot}= q_1+\cdots +q_n$ in an additive fashion.
The superselection rule, enforced via the projection operator $P_q$, forbids coherent superpositions of different net charges $q\neq q'$. 
All IQNs relevant to $\mathrm{U}(1)$ remain packaged and cannot be split among multiple excitations.
This is precisely the packaging principle as applied to an Abelian local gauge group.

\begin{example}[Two-particle example in $\mathrm{U}(1)$]	
	Let $P^\dagger(\mathbf x)$ create a charge $+q$ excitation and
	$\bar P^\dagger(\mathbf x)$ create $-q$.
	A bare two-particle Fock state is, for instance,
	$$
	|\mathbf x,\mathbf y;\,+q,-q\rangle\;=\;
	P^\dagger(\mathbf x)\,\bar P^\dagger(\mathbf y)\,|0\rangle .
	$$
	The gauge projector is
	$$
	\Pi_{\rm phys}\;=\;
	\prod_{\mathbf z}\int_{0}^{2\pi}\!\!\frac{{\rm d}\alpha(\mathbf z)}{2\pi}
	\;e^{\,i\!\int{\rm d}^3z\,\alpha(\mathbf z)\hat G(\mathbf z)} ,
	\qquad
	\hat G(\mathbf z)=\nabla\!\cdot\!\mathbf E(\mathbf z)-\hat\rho(\mathbf z).
	$$
	Because $\hat\rho(\mathbf z)=q\bigl[\delta^{3}(\mathbf z-\mathbf x)-\delta^{3}(\mathbf z-\mathbf y)\bigr]$ on the state above, the phase picked up under a local gauge transformation
	$e^{i\alpha(\mathbf z)}$ is	
	$$
	\exp \Bigl\{ i q \bigl[\alpha(\mathbf x)-\alpha(\mathbf y)\bigr] \Bigr\}.
	$$
	Integrating over all $\alpha(\mathbf z)$ kills the state unless we dress it with a Wilson line that carries the opposite phase factor and makes the composite locally neutral:
	$$
	\Bigl[\bar P^\dagger(\mathbf y)\;
	\exp\!\Bigl(+iq\!\int_{\mathbf y}^{\mathbf x}\!A_i\,{\rm d}x^{i}\Bigr)
	P^\dagger(\mathbf x)\Bigr]\,|0\rangle
	\quad\longmapsto\quad
	\text{survives under } \Pi_{\rm phys}.
	$$
	Thus,
	\begin{itemize}
		\item $|P P\rangle$ ($+2q$) and $|\bar P\bar P\rangle$ ($-2q$):
		annihilated by $\Pi_{\rm phys}$ (globally and locally charged).
		
		\item $|P\bar P\rangle$ without the Wilson line:
		also annihilated (two opposite charges sitting at different points violate local Gauss law).
		
		\item Gauge-dressed meson $P^\dagger W\bar P^\dagger|0\rangle$:
		physical color singlet of $\mathrm{U}(1)$.
	\end{itemize}
\end{example}

\begin{example}[Internal: $\mathrm{U}(1)$ electric charge]
	\label{EX:U1Packaging}
	
	\leavevmode
	\begin{itemize}
		\item Symmetry group.
		$G=\mathrm{U}(1)$ with unitary action $U(\theta)=e^{\,i\theta Q}$, where $Q\in\mathbb Z$ is the integer‑valued electric charge operator.
		
		\item Hilbert space.
		The Fock space decomposes into one‑dimensional $\mathrm{U}(1)$ irreps
		$
		V_Q \cong \mathbb C
		$:
		\[
		\mathcal H
		\;=\;
		\bigoplus_{Q\in\mathbb Z}
		\bigl(V_Q \otimes\mathcal M_Q\bigr)
		=
		\bigoplus_{Q\in\mathbb Z}
		\mathcal H_Q,
		\]
		where $\mathcal M_Q$ carries momentum, spin, flavour, etc.
		
		\item Packaging statement.
		For every gauge‑invariant operator
		$O\in\mathcal O_{\mathrm{phys}}\subseteq U(\mathrm{U}(1))'$,
		\[
		O
		\;=\;
		\bigoplus_{Q\in\mathbb Z}
		\bigl(\mathbf 1_{V_Q}\otimes O_Q\bigr),
		\qquad
		O_Q\in\mathcal B(\mathcal M_Q).
		\]
		Hence the irrep factor $V_Q$ is a $\mathrm{U}(1)$‑packaged degree of freedom.
		No physical process can act on half the charge or transform a
		$Q=+1$ state into $Q=-1$ while respecting gauge-invariance.
		
		\item Physical consequences.
		\begin{itemize}
			\item Superselection.
			States with different $Q$ cannot be	coherently superposed,
			The global phase of each charge sector is unobservable.
			
			\item Charge conservation.
			Local creation/annihilation	operators shift the sector according to
			$Q\!\to\!Q\pm1$ but never break the package.
		\end{itemize}
	\end{itemize}		
	This is the simplest instance of degree‑of‑freedom confinement:
	the label $Q$ is locked and invisible to any $\mathrm{U}(1)$‑respecting detector.
\end{example}

\paragraph{(4) Comparison: Abelian $\mathrm{U}(1)$ vs. Discrete $\mathbb{Z}_N$.}
In both $\mathrm{U}(1)$ gauge group and $\mathbb{Z}_N$ gauge group, irreps are 1‐dimensional, so the group action is a pure phase.
We now compare them in Table \ref{TAB:U1VsZN}:

\begin{table}[!h]
	\centering
	\caption{Abelian $\mathrm{U}(1)$ vs. Discrete $\mathbb{Z}_N$}
	\label{TAB:U1VsZN}
	\begin{tabular}{p{4cm}|p{5cm}|p{6cm}}
		\toprule
		Feature  & Abelian $\mathrm{U}(1)$ & Discrete $\mathbb{Z}_N$ \\
		\midrule
		Irreducible representations & 1-dimensional, labeled by integer charge $q\in\mathbb{Z}$ & 1-dimensional, labeled by residue $a\in\{0,1,\dots,N-1\}$ \\
		\hline
		Group action on single-particle states & Multiplication by phase $e^{i q\theta}$ & Multiplication by root of unity $\omega^a$, $\omega=e^{2\pi i/N}$ \\
		\hline
		Packaging mechanism & Charge $q$ cannot be split among separate operators & Residue $a$ cannot be fractionated, each excitation is a discrete block \\
		\hline
		Free vs. confined charges & Non-confining: individual charges $q\neq0$ are freely observable (electrons, protons) & Confining: only net flux $a_{\rm tot}\equiv0\pmod N$ propagates freely, nonzero $a$ must bind into composites summing to zero mod $N$ \\
		\hline
		Superselection & States with different $q$ lie in distinct superselection sectors & States with different $a$ lie in distinct superselection sectors \\
		\hline
		Multi-particle total charge & Additive: $q_{\rm tot}=\sum_i q_i\in\mathbb Z$ & Additive mod $N$: $a_{\rm tot}=\sum_i a_i\pmod N$  \\
		\bottomrule
	\end{tabular}
\end{table}

This highlights how both gauge groups enforce atomic (indivisible) charge labels, yet differ in spectrum (infinite $\mathbb Z$ vs. finite $\mathbb Z_N$) and in which total‐charge sectors can appear as isolated, deconfined excitations.

\subsection{Non-Abelian Group $\mathrm{SU}(N)$}

In non-Abelian gauge theories \cite{Utiyama1956,Cartan1913,Hooft1978}, the packaging principle implies that each single-particle operator carries an indivisible set of gauge quantum numbers.
For $ \mathrm{SU}(N) $, single-particle operators transform in irreps labeled by Young diagrams (with up to $N-1$ rows).
Multi-particle states are constructed as tensor products of these irreps and decompose into a direct sum of irreps via Peter-Weyl projection.
In a local gauge theory (e.g., QCD), only the gauge-singlet subspace (the $\mathbf{1}$) is physically observed, whereas in global symmetries (such as isospin or flavor), non-singlet multiplets may appear externally.
We now describe the general structure before specializing to $\mathrm{SU}(2)$ and $\mathrm{SU}(3)$.

\begin{remark}[Global vs. Covering Groups]
	If one replaces $\mathrm{SU}(N)$ by $\mathrm{PSU}(N)=\mathrm{SU}(N)/\mathbb Z_N$ or by $\mathrm{U}(N)$, then the same decomposition still holds, but the allowed irreps and the package labels shift.
\end{remark}

\paragraph{(1) Single-particle irreps \cite{PeskinSchroeder,WeinbergBook}.}  
Let $\hat{a}^\dagger_{i}$ denote a creation operator that transforms in an irrep $ \mathbf{R} $ of $\mathrm{SU}(N)$.
Its transformation under $g \in \mathrm{SU}(N)$ is given by
\[
U(g)\,\hat{a}^\dagger_{i}\,U(g)^{-1} = \sum_{j=1}^{d_\mathbf{R}} D^{(\mathbf{R})}_{ji}(g)\,\hat{a}^\dagger_{j},
\]
where $D^{(\mathbf{R})}(g)$ is a $d_\mathbf{R}\times d_\mathbf{R}$ matrix.
No proper invariant subspace exists within an irrep.
Thus, the IQNs (e.g., color or flavor) are packaged as a single block.

\paragraph{(2) Multi-particle tensor products and decomposition.}  
For a system of $n$ excitations with each particle transforming in an irrep $ \mathbf{R}_{\alpha_i} $, the multiparticle state belongs to
\[
\mathbf{R}_{\alpha_1} \otimes \mathbf{R}_{\alpha_2} \otimes \cdots \otimes \mathbf{R}_{\alpha_n}.
\]
Using Peter-Weyl projection (or Young diagram) methods, this product decomposes into a direct sum of irreps:
\[
\mathbf{R}_{\alpha_1} \otimes \mathbf{R}_{\alpha_2} \otimes \cdots \otimes \mathbf{R}_{\alpha_n} \cong \bigoplus_{\beta} N_{\beta}\,\mathbf{R}_{\beta}.
\]
For local gauge groups (e.g., color $\mathrm{SU}(3)$), physical states are restricted to the singlet (or net-neutral) subspace.
For global symmetries (e.g., isospin), non-singlet multiplets may be observed as free states.

\paragraph{(3) Module isomorphism.}  
Once the Hilbert space is decomposed as
\[
\mathcal{H} \cong \bigoplus_{\beta} \mathcal{H}_{\beta},
\]
the theorem guarantees an isomorphism of $\mathrm{SU}(N)$-modules:
\[
\Phi_{\beta}: \mathcal{H}_{\beta} \to V_{\beta},
\]
where $V_{\beta}$ is the abstract representation space for the irrep $\mathbf{R}_\beta$.
This means that for all $g \in \mathrm{SU}(N)$ and $\psi \in \mathcal{H}_{\beta}$,
\[
\Phi_{\beta}\bigl(U(g)\psi\bigr) = \rho_{\beta}(g)\,\Phi_{\beta}(\psi),
\]
with $\rho_{\beta}(g)$ the matrix representation of $g$ in $V_{\beta}$.

\begin{example}[Global $\mathrm{SU}(2)$ Example: Isospin in Quantum Mechanics]
	Even without any gauge constraint, a purely global $\mathrm{SU}(2)$ doublet (e.g., proton-neutron isospin) packs into irreducible multiplets by exactly the same Peter-Weyl projector machinery:
	\[
	V_{\tfrac12}\otimes V_{\tfrac12}
	\;\cong\;
	V_1\oplus V_0,
	\]
	and the two-nucleon projector
	\[
	P_{\,J}
	= (2J+1)\int_{\mathrm{SU}(2)}dg\;\chi^{(J)}(g)\,U(g)\otimes U(g)
	\]
	selects the triplet ($J=1$) or singlet ($J=0$) subspace.
	No Gauss law is needed.
	This is pure global packaging.
\end{example}

\paragraph{(4) $\mathrm{SU}(N)$ gauge group.}
	
For the general $\mathrm{SU}(N)$ gauge group,	
Take a quark $q^{a}$ in the fundamental $\mathbf N$ ($a=1,\dots,N$) and an antiquark $\bar q_{b}$ in the $\overline{\mathbf N}$.

\begin{enumerate}
	\item The product decomposes as	
	$$
	\mathbf N\;\otimes\;\overline{\mathbf N}
	\;=\;
	\mathbf 1\;\oplus\;\mathbf{Adj},\qquad
	\dim\mathbf{Adj}=N^{2}-1.
	$$	

	\item A color octet (adjoint) state can be written as
	$$
	\Psi^A_{\rm adj} = \sqrt2\,T^A{}_{a}{}^{b}\,q^a\bar q_b\ket0, ~ A=1,\dots,N^2-1,
	$$
	where ${T^A}$ are the generators normalized by $\operatorname{Tr} [T^A T^B] = \tfrac12 \delta^{AB}$.
	
	\item Packaging.
	The quark and antiquark are individually packaged ($\mathbf N$ and $\overline{\mathbf N}$).
	The superposition lives in a fixed color sector (Adj), so it is a perfectly valid packaged-entangled state.
	
	\item Projector.
	$\Pi_{\text{phys}}$ averages over all local $\mathrm{SU}(N)$ transformations and keeps only the singlet component.
	Because $\ket{\Psi_{\text{adj}}}$ transforms in the adjoint irrep, the projector kills it:	
	$$
	\Pi_{\text{phys}}\ket{\Psi_{\text{adj}}}=0.
	$$	
	Thus the state is non-physical.
\end{enumerate}

\subsection{$\mathrm{SU}(2)$ Group}
\label{SEC:SU2Packaging}

In this section, we work out the full machinery of symmetry packaging when the gauge group is $\mathrm{SU}(2)$ \cite{Yang1954}.
Recall that the finite-dimensional irreps of $\mathrm{SU}(2)$ are labeled by spin $j\in\{0,\tfrac12,1,\tfrac32,\dots\}$, each of dimension $2j+1$.

\subsubsection*{1. Single-particle packaging}

Let $V_{\frac12}\cong\mathbb C^2$ be the fundamental (doublet) of $\mathrm{SU}(2)$, with carrier index $i=1,2$.  Any one-particle creation operator $\hat q^\dagger_i$ transforms under the unitary representation
\[
U(g)\,\hat q^\dagger_i\,U(g)^{-1}
\;=\;
D^{(\frac12)}_{ji}(g)\,\hat q^\dagger_j,
\quad
g\in\mathrm{SU}(2),
\]
where $D^{(\frac12)}(g)\in\mathrm{U}(2)$.  Since $V_{\frac12}$ is irreducible, no nontrivial $\mathrm{SU}(2)$-invariant operator can split the doublet into smaller blocks: the entire two-component spinor is an inseparable package.

More generally, if the internal multiplet space is 
\[
V_{\rm int}
\;\cong\;
\bigoplus_{j\in\frac12\mathbb N} m_j\,V_j,
\]
then, by Maschke/Peter-Weyl theorem, each summand $V_j$ is a single-particle packaged subspace of dimension $2j+1$.

\subsubsection*{2. Two-particle Decomposition and CG Packaging}

On the two-particle sector one considers the tensor product
\[
V_{\tfrac12}\otimes V_{\tfrac12}
\;\cong\;
V_1\;\oplus\;V_0,
\]
i.e.,
\[
\mathbf2\otimes\mathbf2
=
\underbrace{\mathbf3}_{j=1}
\;\oplus\;
\underbrace{\mathbf1}_{j=0}.
\]
Using Clebsch-Gordan coefficients \cite{Racah1942,Wigner1959}, one writes the coupled basis
\[
\bigl|\,j,m\,\bigr\rangle
\;=\;
\sum_{m_1,m_2}
\bigl\langle\tfrac12\,m_1;\tfrac12\,m_2\bigm|j\,m\bigr\rangle
\;\bigl|\,\tfrac12,m_1\,\bigr\rangle
\otimes
\bigl|\,\tfrac12,m_2\,\bigr\rangle,
\]
with
\[
(j,m)\;=\;(1,1),\,(1,0),\,(1,-1)\quad\text{(triplet)},
\quad
(0,0)\quad\text{(singlet)}.
\]
The singlet state is
$$
\ket{\Psi_{0,0}}
= \frac1{\sqrt2}\Bigl(\ket{\tfrac12,+\tfrac12;\tfrac12,-\tfrac12}
- \ket{\tfrac12,-\tfrac12;\tfrac12,+\tfrac12}\Bigr).
$$
equivalently $ \epsilon_{ij}\,\hat q^\dagger_i\hat q^\dagger_j\ket0$.
Since it carries $j=0$, it is invariant under every $U(g)$ and hence is a fully packaged singlet.

\subsubsection*{3. Projectors onto Irreducible Sectors}

In general, the projector onto spin-$j$ in the $n$-fold tensor product is
$$
P_j = (2j+1) \int_{\mathrm{SU}(2)} dg ~ \chi^{(j)}(g)^* U(g)^{\otimes n},
\qquad \int_{\mathrm{SU}(2)} dg = 1.
$$
where $\chi^{(j)}(g)=\mathrm{Tr}\,D^{(j)}(g)$.
Acting on any $n$-particle state, $P_j$ isolates the packaged subspace $V_j\otimes \mathcal H_{\rm ext}^{\otimes n}$.

\subsubsection*{4. Multi-particle Packaging: General Clebsch-Gordan Series}

Under repeated Clebsch-Gordan coupling
\[
V_{j_1}\otimes V_{j_2}
\;\cong\;
\bigoplus_{J=|j_1-j_2|}^{\,j_1+j_2}\,V_J,
\]
and similarly for more factors, one obtains the full isotypic decomposition
\[
\mathcal H_{\rm iso}
\;=\;
\bigoplus_{n\ge0}
\Bigl(V_{\rm int}^{\otimes n}\otimes\mathcal H_{\rm ext}^{\otimes n}\Bigr)
\;=\;
\bigoplus_{J\in\frac12\mathbb N}
\Bigl(V_J\;\otimes\;\mathcal M_J\Bigr),
\]
where $\mathcal M_J$ is the multiplicity (external) space.
Each $V_J$ is a multi-particle packaged subspace of total spin $J$.

\subsubsection*{5. Packaged Entangled Bases and Schmidt Rank}

Within a fixed $J$-sector ($\dim>1$), pick the orthonormal coupled basis $\{|J,m\rangle\}_{m=-J}^J$.  
Any pure state in that sector $\ket\Psi=\sum_m c_m\ket{J,m}$ has Schmidt rank $\le2J+1$ and is entangled whenever more than one $c_m\neq0$.

\subsubsection*{6. Physical (singlet) packaged subspaces}

Finally, the gauge-invariant physical subspace is the total-singlet sector
\[
\mathcal H_{\rm phys}
\;=\;
P_{0}\,\mathcal H_{\rm iso},
\]
which is nontrivial exactly if some tensor power of $V_{\rm int}$ contains $V_0$.
In $\mathrm{SU}(2)$ that is always true for even numbers of fundamentals.

\bigskip
\noindent
This completes a model-independent derivation of how $\mathrm{SU}(2)$ packages its irreducible spin blocks into indecomposable gauge-charge units, how one projects onto them, and how one builds packaged entangled bases via Clebsch-Gordan coupling.

\begin{remark}[Isomorphism and Classification]
	By applying the projector to the two-color system, one projects the full Hilbert space $\mathcal{H}$ onto the singlet subspace $\mathcal{H}_{\mathbf{1}}$.
	In this subspace, for all $g\in\mathrm{SU}(2)$, the gauge transformation $U(g)$ acts on the projected packaged state $|\Psi_{\mathbf{1}}\rangle \in \mathcal{H}_1$ trivially, i.e.,
	\[
	U(g)\,|\Psi_{\mathbf{1}}\rangle = |\Psi_{\mathbf{1}}\rangle.
	\]
	Thus, we have an isomorphism
	\[
	\Phi_{\mathbf{1}}: \mathcal{H}_{\mathbf{1}} \to V_{\mathbf{1}},
	\]
	where $V_{\mathbf{1}}$ is the one-dimensional abstract representation space corresponding to the singlet.
	This example shows that each single-particle operator carries an indivisible gauge charge and multiparticle states can be projected onto gauge-invariant and packaged subspaces.
\end{remark}

\begin{example}[weak isospin]	
	Let $\ell^{\alpha}$ and $\ell^{\beta}$ be two fermions in the doublet $\mathbf 2$.
	The decomposition	
	$$
	\mathbf 2\;\otimes\;\mathbf 2
	\;=\;
	\mathbf 1_A\;\oplus\;\mathbf 3_S
	$$	
	produces
	\begin{itemize}
		\item Singlet (physical):
		$$
		\ket{\Psi_{1}} =\tfrac1{\sqrt2}\bigl(
		\ket{\uparrow\downarrow}-\ket{\downarrow\uparrow}
		\bigr).
		$$
		
		\item Triplet (non-physical):
		
		$$
		\ket{\Psi_{3}} =\ket{\uparrow\uparrow}
		\quad\text{or}\quad
		\tfrac1{\sqrt2}\bigl(
		\ket{\uparrow\downarrow}+\ket{\downarrow\uparrow}
		\bigr),\;
		\ket{\downarrow\downarrow}.
		$$
		
		The triplet states are again packaged-entangled but get removed by $\Pi_{\text{phys}}$ because they transform in the adjoint $\mathbf 3$.
	\end{itemize}
\end{example}

\subsection{$\mathrm{SU}(3)$ Group}
\label{SEC:SU3Packaging}

The compact Lie group $\mathrm{SU}(3)$ has irreps labeled by a pair of non-negative integers $(p,q)$, with dimension
$$
\dim V_{(p,q)} = \tfrac12\,(p+1)(q+1)(p+q+2).
$$
In physics one often writes
\[
\mathbf{3}\;=\;V_{(1,0)},\quad
\bar{\mathbf{3}}\;=\;V_{(0,1)},\quad
\mathbf{6}=V_{(2,0)},\quad
\mathbf{8}=V_{(1,1)},\;\dots
\]
but the same packaging logic holds for any $(p,q)$.

\subsubsection*{1. Single-particle Packaged Subspaces}

Let
\[
V_{\rm int}
\;\cong\;
\bigoplus_{(p,q)\in\widehat{\mathrm{SU}(3)}} m_{(p,q)}\,V_{(p,q)}
\]
be the internal representation carried by one particle, and let $\mathcal H_{\rm ext}$ be any external Hilbert space (spin, momentum, flavor, $\cdots$) on which $\mathrm{SU}(3)$ acts trivially.
Then
\[
\mathcal H_{\mathbf 1}
=
V_{\rm int} \otimes \mathcal H_{\rm ext}
\cong
\bigoplus_{(p,q)} m_{(p,q)} \bigl(V_{(p,q)}\otimes\mathcal H_{\rm ext}\bigr).
\]
Each summand $V_{(p,q)}\otimes\mathcal H_{\rm ext}$ is a single-particle packaged subspace:
no gauge-invariant operator can split the full $(p,q)$ block into smaller pieces.

\subsubsection*{2. Two-particle Clebsch-Gordan decomposition \cite{Racah1942,Wigner1959}}

On the two-particle sector we have the tensor product of two fundamentals, or more generally any two irreps:
\[
V_{(p_1,q_1)}\otimes V_{(p_2,q_2)}
\;\cong\;
\bigoplus_{(r,s)} c^{(r,s)}_{(p_1,q_1),(p_2,q_2)}\;V_{(r,s)},
\]
where the Clebsch-Gordan multiplicities $c^{(r,s)}_{\cdots}$ are determined by the usual $\mathrm{SU}(3)$ weight-diagram or Young-tableau rules \cite{Georgi2000,Gross1973,GellMann1961}.
Specifically
\[
\mathbf{3}\otimes\mathbf{3}
=\;V_{(1,0)}\!\otimes V_{(1,0)}
\;=\;
V_{(2,0)}\;\oplus\;V_{(0,1)}
\;=\;
\mathbf{6}\;\oplus\;\bar{\mathbf3},
\]
\[
\mathbf{3}\otimes\bar{\mathbf{3}}
=\;
V_{(1,0)}\!\otimes V_{(0,1)}
\;=\;
V_{(0,0)}\;\oplus\;V_{(1,1)}
\;=\;
\mathbf{1}\;\oplus\;\mathbf{8},
\]
and so on.

Choosing an orthonormal basis $\{\ket{(r,s);\,\alpha}\}$ of each $V_{(r,s)}$, one writes
\[
\ket{(r,s);\alpha}
\;=\;
\sum_{i,j}
\braket{(p_1,q_1),i;\,(p_2,q_2),j}{(r,s);\alpha}
\;\ket{(p_1,q_1),i}\otimes\ket{(p_2,q_2),j},
\]
where the coefficients are the $\mathrm{SU}(3)$ Clebsch-Gordan coefficients.
Each $V_{(r,s)}$ is a two-particle packaged subspace of dimension $\dim V_{(r,s)}$.

\subsubsection*{3. General $n$-particle Packaging and Isotypic Decomposition}

Extending to the Fock space
\[
\mathcal H_{\rm Fock}(\mathcal H_{\mathbf{1}})
\;=\;
\bigoplus_{n=0}^\infty
\begin{cases}
	\operatorname{Sym}^n(\mathcal H_{\mathbf{1}}) & (\text{bosons}),\\
	\wedge^n(\mathcal H_{\mathbf{1}}) & (\text{fermions}),
\end{cases}
\]
Maschke/Peter-Weyl gives an isotypic decomposition
\[
\mathcal H_{\rm iso}
\;=\;
\bigoplus_{n\ge0}\bigl(\mathcal H_{\mathbf{1}}^{\,(\otimes_s\,\text{or}\,\wedge)^n}\bigr)
\;=\;
\bigoplus_{(p,q)}
\Bigl[V_{(p,q)} \otimes \mathcal M_{(p,q)}\Bigr],
\]
where $\mathcal M_{(p,q)}$ is the multiplicity (external) space.
Each $\mathcal H_{(p,q)} = V_{(p,q)} \otimes \mathcal M_{(p,q)}$ is a multi-particle packaged subspace of total color type $(p,q)$.

\subsubsection*{4. Projectors via Peter-Weyl Theorem}

The projector onto the irrep $(p,q)$ in the $n$-fold sector is
\[
P_{(p,q)}
\;=\;
\dim V_{(p,q)}
\int_{\mathrm{SU}(3)}dg\;\chi_{(p,q)}(g)^*\,
\bigl[\,U(g)\bigr]^{\otimes n},
\]
with $\chi_{(p,q)}(g)=\operatorname{Tr} D^{(p,q)}(g)$ and Haar measure normalized to 1.
On any state
$\ket\Psi \in \mathcal H_{\rm iso}$,
$P_{(p,q)}\ket\Psi\in \mathcal H_{(p,q)}$, and 
$\sum_{(p,q)} P_{(p,q)} = \mathbf{1}$.

\subsubsection*{5. Packaged Entangled (Clebsch-Gordan) Bases}

Within each $\mathcal H_{(p,q)}$ of dimension $>1$ one may choose the coupled basis
$\{\ket{(p,q);\alpha}\}$.
Any vector mixing two or more $\alpha$-labels is entangled in color, i.e., it cannot be written as a product of lower-spin building blocks.
The Schmidt rank in the factorization $\mathcal H_{(p,q)} \cong V_{(p,q)} \otimes \mathcal M_{(p,q)}$ is the packaged Schmidt rank.

\subsubsection*{6. Physical (Color Singlet) Packaged Subspace}

Finally, the local Gauss law (or full gauge‐averaging) projects onto the trivial $(0,0)$ block:
\[
\mathcal H_{\rm phys}
\;=\;
P_{(0,0)}\,\mathcal H_{\rm iso},
\]
and is nonzero exactly when some combination of fundamentals and anti-fundamentals contain the singlet.
Specifically, in QCD $\mathbf3\otimes\bar{\mathbf3}\supset\mathbf1$ and $\mathbf3^{\otimes3}\supset\mathbf1$, etc., so mesons and baryons survive while any isolated colored block is annihilated.

\bigskip
\noindent
In this way the entire tower of $\mathrm{SU}(3)$ irreps, their Clebsch-Gordan decomposition, and the Peter-Weyl projector machinery combine to give a fully general packaging principle:
\[
\text{any }n\text{-particle state}
\;\longrightarrow\;
\bigoplus_{(p,q)}
\underbrace{V_{(p,q)}}_{\substack{\text{inseparable}\\\text{`color' package}}}
\;\otimes\;
\underbrace{\mathcal M_{(p,q)}}_{\substack{\text{external}\\\text{multiplicities}}}.
\]
Only the trivial $(0,0)$ block is physical in a confining gauge theory, and every nontrivial $(p,q)\neq(0,0)$ is an indivisible $\mathrm{SU}(3)$ package.

\begin{remark}[Local Versus Global Symmetries]
	For local gauge symmetries (e.g., QCD color $\mathrm{SU}(3)$), only the gauge-singlet ($\mathbf{1}$) states are free due to confinement.
	Non-singlet states remain internal.
	For global symmetries such as isospin or flavor, non-singlet multiplets (e.g., the baryon octet or decuplet in $\mathrm{SU}(3)_{\text{flavor}}$) can appear as observable free states.
	In both cases, however, the packaging principle holds:
	each single-particle operator carries an irreducible set of quantum numbers that cannot be partially factored.
\end{remark}

\begin{example}[QCD colors]	
	For concreteness write the quark colors as $\{r,g,b\}$.
	
	\begin{itemize}
		\item Singlet (physical)
		
		In $\mathbf 3\otimes\overline{\mathbf 3}$ the color singlet is	
		$$
		\ket{\Psi_{1}}
		\;=\;
		\tfrac1{\sqrt3}\bigl(
		\ket{r\bar r}\!+\!\ket{g\bar g}\!+\!\ket{b\bar b}
		\bigr)
		$$	
		
		\item $\mathbf 8$ example (unphysical)
		
		The eight color octet states span $\mathbf 8$.
		Pick one of them, e.g.,
		$$
		\ket{\Psi_{8}}
		\;=\;
		\tfrac1{\sqrt2}\bigl(
		\ket{r\bar r}-\ket{g\bar g}
		\bigr).
		$$		
		Here $\ket{\Psi_{8}}$ is packaged entangled between quark and anti-quark and obeys the fixed-charge (color) condition.		
		Under any color rotation $U \in \mathrm{SU}(3)$, it transforms as an octet, not as a singlet.
		Consequently		
		$$
		\Pi_{\text{phys}}\ket{\Psi_{8}}=0 .
		$$		
		The state is nonphysical.
	\end{itemize}	
\end{example}

Thus, the packaged states in non-Abelian gauge theories are classified by the irreducible subrepresentations obtained from the tensor product of single-particle blocks.
This construction explains physical phenomena such as color confinement \cite{Gross1973,Politzer1973} and superselection, highlighting that the internal DOFs are intrinsically entangled and cannot be partially separated.
In QCD, only packaged entangled states like color singlet are observable.
While in global symmetries non-singlet multiplets may appear.

\subsection{Combined Finite $\times$ Compact Symmetry: $G \times D$}

\paragraph{(1) Combined symmetry group.}
In many theories, the symmetry group comprises both a continuous local gauge symmetry $G$ and an additional discrete symmetry $D$ \cite{KraussWilczek1989,IbanezRoss1992,BanksDine1992}.
For example, in quantum chromodynamics (QCD) the gauge group $G$ might be $\mathrm{SU}(3)$ (local color symmetry), while an additional discrete symmetry $D$ (such as charge conjugation $C$ or baryon parity) is present.
In such cases, the full classification of packaged states is refined by a combined label.
The total symmetry group is given by the direct product
\[
\mathcal{G} = G \times D.
\]
The Hilbert space $\mathcal{H}$ of the theory then naturally decomposes into sectors characterized by both the net gauge (local) charge and the discrete charge,
i.e., if $\hat{Q}$ is the self-adjoint operator for the gauge charge and $d$ labels the eigenvalues of the discrete symmetry, then one may write
\[
\mathcal{H} = \bigoplus_{Q\in \sigma(\hat{Q})} \mathcal{H}_Q \quad \text{with} \quad \mathcal{H}_Q = \bigoplus_{d\in D}\, \mathcal{H}_{Q,d},
\]
where $Q$ runs over the gauge-charge spectrum $\sigma(\hat Q)$ and $d$ over the discrete labels $D$.

Here, $\mathcal{H}_{Q,d}$ is the subspace of states with net gauge charge $Q$ and discrete label $d$.
For example, if baryon number is broken down to $\mathbb{Z}_N$, then each state carries a combined label $(Q,d)$.
In the case of a local gauge symmetry, superselection rules ensure that states belonging to different $\mathcal{H}_Q$ cannot interfere.
In contrast, since $D$ is usually a global symmetry, it often allows additional structure within a given gauge sector.
In other words, the discrete symmetry further partitions $\mathcal{H}_Q$ into finer sub-sectors and packages the wavefunction by imposing additional superselection rules.

\paragraph{(2) Irreps of the combined group $\boldsymbol{G \times D}$.}
When a field transforms under both $G$ and $D$, a single-particle creation operator $\hat{a}^\dagger$ carries a composite index:
\[
\hat{a}^\dagger_{\alpha,d}(\mathbf{p}),
\]
where $\alpha$ labels the irrep of the continuous gauge group $G$ (for example, $\mathbf{3}$ for quarks in $\mathrm{SU}(3)$) and $d$ labels the discrete charge under $D$ (for example, the eigenvalue $\pm 1$ under a $\mathbb{Z}_2$ symmetry).
Under a combined transformation $(g,\delta)\in G\times D$, the operator transforms as
\[
U_{(g,\delta)}\,\hat{a}^\dagger_{\alpha,d}\,U_{(g,\delta)}^{-1}
=\sum_{\alpha',d'} \Bigl[D^{(\alpha,d)}(g,\delta)\Bigr]_{\alpha',\,d';\;\alpha,\,d}\,\hat{a}^\dagger_{\alpha',d'}\,.
\]
If the discrete group $D$ is Abelian (e.g., $\mathbb{Z}_N$), the discrete part of the representation is typically one-dimensional (a phase factor, e.g., $\omega^k$ with $\omega=e^{2\pi i/N}$).
Thus, the total transformation is a tensor product of the continuous and discrete parts, and the state's full charge is given by the combined label $(Q,d)$.

\paragraph{(3) Physical interpretation.}
While the local gauge symmetry $G$ imposes a strict superselection rule (so that no state may be a superposition of different net gauge charges), the discrete symmetry $D$ - being a global symmetry - allows an additional refinement within each gauge-charge sector.
For instance, even if a state has net gauge charge $Q$, it can still carry a discrete label $d$ from $D$.
The decomposition
\[
\mathcal{H}_Q = \bigoplus_{d\in D}\, \mathcal{H}_{Q,d}
\]
means that states with different discrete labels cannot interfere with one another.
In this sense, the discrete symmetry further packages the internal DOFs, refining the overall structure to a combined charge $(Q,d)$.

\subsubsection{Representation‐theoretic Decomposition for $G\times D$}
\label{sec:GD-decomp}

Since $G$ and $D$ commute, every irrep of $G\times D$ is a tensor product of an irrep of $G$ and an irrep of $D$.
Writing
\[
\widehat G=\{\,R\},\quad \widehat D=\{\delta\},
\]
we have the isomorphism of $G\times D$‐modules
\[
\mathcal H
\;\cong\;
\Bigl(\bigoplus_{R\in\widehat G}V_R\otimes\mathcal M_R\Bigr)
\;\otimes\;
\Bigl(\bigoplus_{\delta\in\widehat D}W_\delta\Bigr)
\;=\;
\bigoplus_{R,\delta}
\Bigl(V_R\otimes W_\delta\otimes\mathcal M_R\Bigr),
\]
where
$V_R$ is the carrier space of the continuous irrep $R$,  
$W_\delta$ is the carrier of the discrete irrep $\delta$,  
and $\mathcal M_R$ is the multiplicity (external) space for $G$.

Equivalently, one may collect the two sums into a single double‐sum over combined labels $(R,\delta)$:
\[
\mathcal H
=\bigoplus_{(R,\delta)}
\mathcal H_{R,\delta},
\qquad
\mathcal H_{R,\delta}
:=V_R\otimes W_\delta\otimes \mathcal M_R.
\]
Each summand is an inseparable package carrying both a continuous charge $R$ and a discrete charge $\delta$.

\subsubsection{Joint Projectors onto $(R,\delta)$‐sectors}
\label{sec:GD-projector}

By Peter-Weyl on $G$ and ordinary orthogonality on $D$, the projector onto the $(R,\delta)$ block is simply the product of the two
\[
P_{R,\delta}
=\underbrace{\frac{d_R}{|G|}\sum_{g\in G}\chi_R^*(g)\,U_G(g)}_{P_R}
\;\times\;
\underbrace{\frac{1}{|D|}\sum_{h\in D}\chi_\delta^*(h)\,U_D(h)}_{P_\delta}\,.
\]
Since $[U_G(g),U_D(h)]=0$, $P_{R,\delta}$ satisfies
\[
P_{R,\delta}^2=P_{R,\delta},
\quad
P_{R,\delta}P_{R',\delta'}=0\;\text{for }(R,\delta)\neq(R',\delta'),
\quad
\sum_{R,\delta}P_{R,\delta}=\mathbf1.
\]
Acting on any state $\ket\Psi$, $P_{R,\delta}\ket\Psi\in\mathcal H_{R,\delta}$ exactly.

\subsubsection{Packaged entanglement in $G\times D$}
\label{sec:GD-entanglement}

Within a fixed $(R,\delta)$‐sector, the Hilbert space factorizes as
\[
\mathcal H_{R,\delta}
\;\cong\;
V_R\;\otimes\;W_\delta\;\otimes\;\mathcal M_R.
\]
We call a pure state in $\mathcal H_{R,\delta}$ hybrid‐package‐entangled if its Schmidt rank across the split 
$
V_R\otimes W_\delta\;\big|\;\mathcal M_R
$
or across 
$
V_R\;\big|\;W_\delta\otimes\mathcal M_R
$
is strictly greater than 1.
In particular, any non‐trivial superposition
$
\sum_i c_i\,\phi_i\otimes\psi_i
\in V_R\otimes W_\delta
$
is a packaged‐entangled state in the combined discrete×continuous label.

\begin{example}[$\mathrm{SU}(2)\times\mathbb Z_2$ fermion‐parity]
	Consider the symmetry groups $G=\mathrm{SU}(2)$ and $D = \mathbb Z_2^{\rm f} = \{+1,-1\}$ be fermion‐parity, with irreps
	$
	\widehat G=\{j=0,\tfrac12,1,\dots\}
	$
	and
	$
	\widehat D=\{\delta=+1\,,\,-1\}.
	$
	\begin{enumerate}
		\item \textbf{Single‐particle.}
		
		\begin{itemize}
			\item Single‐particle operator.
			A fermion creation operator
			$
			\hat\psi^\dagger_{i,p}
			$
			carries an $\mathrm{SU}(2)$‐index $i=1,2$ (fundamental $j=\tfrac12$) and a parity label $p=\pm1$ under $\mathbb Z_2^{\rm f}$.
			
			\item Combined transformation.
			Under
			$
			(g,p')\in \mathrm{SU}(2)\times\mathbb Z_2,
			$
			it transforms as
			\[
			U(g,p')\,\hat\psi^\dagger_{i,p}\,U(g,p')^{-1}
			=\sum_{j=1}^2D^{(1/2)}_{ji}(g)\;\chi_p(p')\;\hat\psi^\dagger_{j,p},
			\]
			where the $\mathbb Z_2$‐character is 
			$\chi_{+1}(p')=+1$, $\chi_{-1}(p')=p'$.
			
			\item Single-particle packages.
			There are two irreducible blocks
			\[
			V_{1/2}\otimes W_{+1}
			\quad\text{and}\quad
			V_{1/2}\otimes W_{-1}
			\]
			corresponding to fermion-parity even/odd doublets.
						
			\item Projector onto $(j,\delta)$‐sector.
			For example,
			\[
			P_{\tfrac12,-1}
			\;=\;
			\underbrace{2\!\int_{\mathrm{SU}(2)}dg\;\chi^{(1/2)}(g)^*\,U_{\mathrm{SU}(2)}(g)}_{P_{1/2}}
			\;\times\;
			\underbrace{\tfrac12\sum_{p'=\pm1}p'\,U_{\mathbb Z_2}(p')}_{P_{-1}}.
			\]
			This isolates exactly the $(j=\tfrac12,\delta=-1)$ subspace.
		\end{itemize}
		
		\item \textbf{Two-particle packaged entangled state.}
		
		An explicit example in the $((\tfrac12,-1)^{\otimes2})$ sector is
		\[
		\ket{\Psi}
		= \frac{1}{\sqrt2}\Bigl(
		\hat\psi^\dagger_{1,-}(A)\,\hat\psi^\dagger_{2,-}(B)
		+ \hat\psi^\dagger_{2,-}(A)\,\hat\psi^\dagger_{1,-}(B)
		\Bigr)\ket0,
		\]
		which has nontrivial Schmidt rank across the $\mathrm{SU}(2)\otimes\mathbb Z_2$ factor and therefore is a genuine hybrid‐packaged entangled state.
	\end{enumerate}
\end{example}

\begin{example}[$\mathrm{SU}(2) \times C$ (Charge Conjugation)]
	For $\mathrm{SU}(2)$ gauge symmetry, the fundamental representation is 2-dimensional.
	Let $C$ denote charge conjugation, which, in this context, interchanges particles and antiparticles (in $\mathrm{SU}(2)$, a quark and an antiquark both transform as doublets, but with conjugate labels).
	\begin{itemize}
		\item A single-particle operator might be written as $\hat{a}^\dagger_{\alpha,c}(\mathbf{p})$, where $\alpha\in\{1,2\}$ labels the $\mathrm{SU}(2)$ index and $c\in\{+,-\}$ denotes the eigenvalue of the charge conjugation operator $C$ (i.e., whether it is a particle or an antiparticle).  
		
		\item Under a combined transformation $(g,C)\in \mathrm{SU}(2)\times \mathbb{Z}_2$, the operator transforms as 
		\[
		U(g,C)\,\hat{a}^\dagger_{\alpha,c}\,U(g,C)^{-1} \;=\; \sum_{\beta} D^{(\mathbf{2})}_{\beta\alpha}(g)\, \chi_c(C)\,\hat{a}^\dagger_{\beta,c'}\,,
		\]
		where $\chi_c(C)$ is the one-dimensional character of $\mathbb{Z}_2$ (for instance, $\chi_+(C)=+1$ and $\chi_-(C)=-1$).  
		
		\item The full Hilbert space then decomposes into sectors labeled by the net $\mathrm{SU}(2)$ charge (which may be chosen by the rules of confinement) and a discrete label from $C$.
	\end{itemize}
\end{example}

\begin{example}[$\mathrm{SU}(3) \times \mathbb{Z}_N$ (Baryon Parity)]
	In QCD, the local gauge group is $\mathrm{SU}(3)$ (color).
	Suppose, in addition, there is an exact discrete symmetry $\mathbb{Z}_N$ associated with baryon parity (a remnant of a broken $\mathrm{U}(1)$).
	\begin{itemize}
		\item Each quark creation operator is in the $\mathbf{3}$ of $\mathrm{SU}(3)$ and carries a definite baryon number (say, $-1/3$) along with a discrete $\mathbb{Z}_N$ label.
		For simplicity, denote the discrete charge by $d\in\{0,1,\dots,N-1\}$.  
		
		\item A meson state, for example, is built from one quark and one antiquark:
		\[
		\vert M\rangle \;=\; \frac{1}{\sqrt{3}} \sum_{c=1}^3 \vert q_c,d\rangle \otimes \vert \bar{q}_c,\bar{d}\rangle~,
		\]
		where the color indices are summed so that the state is a color singlet, and the discrete charges satisfy $d+\bar{d} \equiv 0 \pmod{N}$.  
		
		\item The Hilbert space is then decomposed as 
		\[
		\mathcal{H}_0 = \bigoplus_{d\in \mathbb{Z}_N} \mathcal{H}_{0,d}\,,
		\]
		where $\mathcal{H}_{0,d}$ is the subspace with net color singlet and discrete label $d$.  
		
		\item The discrete symmetry further refines the state classification:
		states with different $\mathbb{Z}_N$ labels belong to distinct superselection sectors, so that no physical observable can create interference between, say, $\mathcal{H}_{0,0}$ and $\mathcal{H}_{0,1}$.
	\end{itemize}
\end{example}

\begin{remark}[Spontaneous symmetry breaking]
	When a compact gauge group $G$ is Higgsed down to a subgroup $H \subset G$, the packaging principle still applies:
	irreps of $G$ decompose under $H$ into blocks $\bigoplus_{\mu}V_\mu$.
	One then projects onto those sub-blocks via the same Peter-Weyl projectors, now built from the residual $H$-action.
\end{remark}

\section{Symmetry Packaging under Higher-Form Symmetries}
\label{SEC:SymmetryPackagingOnDifferentialForms}

In conventional gauge theories, we usually work with pointlike particles or local fields.
We call the symmetries related to these objects 0‑form symmetries.
In modern physics, however, $p$-form symmetries are also important \cite{Gaiotto2015,Polyakov1975,Luscher1980,Bali1995,Bali2001,Takahashi2002}.
In a $p$-form symmetry, the symmetry acts on objects supported on $p$-dimensional submanifolds, which can also be in packaged states.

In this section, 
we introduce $p$-form symmetries, focusing on 1-form and 2-form cases in confining gauge theories.
We will generalize the packaging principle (no partial factorization) to extended flux lines and surfaces (see Theorem VI: Fusion-Category Packaging in \cite{MaSPI2025}).
In other words, the flux carried by a Wilson loop or the flux on a membrane must appear as a single, inseparable unit.
We will also discuss mixed packaged objects and symmetries that combine 0-form charges and 1-form or 2-form fluxes.
Finally, we will relate winding (topological) superselection on lattices to 1-form (or higher-form) charges.
These topological sectors remain superselected because local operators cannot alter global flux quantum numbers.

\subsection{Extended Packaged Objects and $p$-Form Symmetries}

We now describe extended packaged objects in detail and show how it naturally leads to the notion of $p$-form symmetries acting on these extended packaged operators.

\paragraph{(1) Extended Packaged Objects.}
Consider an extended operator $O^{(p)}$ defined on a $p$-dimensional submanifold $\mathcal{M}_p$ of spacetime.\cite{Polchinski1995,Witten1995,Kitaev2006}
Such an operator might represent a flux tube (if $p=1$) or a surface defect (if $p=2$).
Suppose $O^{(p)}$ is labeled by a charge (or flux) $\gamma$ taking values in a group $G$ (or in one of its representations).
This charge $\gamma$ is analogous to the electric charge or color charge in pointlike particles.
It is fully packaged and cannot be subdivided.
In analogy with the pointlike case, the packaging principle asserts that the operator $O^{(p)}$ transforms irreducibly under $G$.
In other words, the full flux label $\gamma$ appears as an inseparable unit.

Assume that under a group transformation $g\in G$, the extended operator transforms as
\[
U(g)\, O^{(p)}\, U(g)^{-1} = \rho(g)\, O^{(p)},
\]
where $\rho$ is a representation of $G$ on the space of such operators.
Since $\rho(g)$ is irreducible, there is no invariant subspace of $O^{(p)}$.
Hence, the entire flux label $\gamma$ remains intact.
If one attempted to split the operator into two parts,
\[
O^{(p)} \stackrel{?}{=} O_1^{(p)} \otimes O_2^{(p)},
\]
with individual labels $\gamma_1$ and $\gamma_2$ such that $\gamma_1 \cdot \gamma_2 = \gamma$.
Then by the linearity of the representation, one would have
\[
U(g)\,(O_1^{(p)} \otimes O_2^{(p)})\,U(g)^{-1} = \rho_1(g)\,O_1^{(p)} \otimes \rho_2(g)\,O_2^{(p)}.
\]
If the original $O^{(p)}$ transforms irreducibly, no nontrivial decomposition into invariant subspaces exists.
Therefore, the entire extended object must be treated as a single packaged block, and its flux cannot be fractionated.
This is the extended analogue of the statement that, for pointlike excitations, the full internal charge is inseparable.

\begin{example}[Flux lines in confining theories]
	In a confining $\mathrm{SU}(N)$ gauge theory, a flux line (or Wilson loop) carries a center charge (e.g., in $\mathbb{Z}_N$).
	Under a center transformation,
	\[
	U(g)\, W(C)\, U(g)^{-1} = e^{2\pi i k/N}\, W(C),
	\]
	where $k$ is the charge of the loop.
	Since $W(C)$ transforms irreducibly under the center, it cannot be split into segments carrying fractional center charge.
	The entire flux line is a packaged object:
	its center charge is inseparable, similar to how a quark's color charge cannot be split.
\end{example}

In lattice formulations or theories with periodic boundary conditions, extended objects (flux lines or surfaces) may wrap around non-contractible cycles.
This leads to topologically protected superselection sectors where the total flux is conserved.

\paragraph{(2) $p$-Form Symmetries.}
When dealing with extended operators, it is natural to introduce the concept of a $p$-form symmetry \cite{KapustinSeiberg2014,Gaiotto2015,Tachikawa2020,BanksSeiberg2011}. In $d$-dimensional spacetime, a $p$-form global symmetry is characterized by conserved charges measured on $(d-p)$-dimensional manifolds. The charged operators under a $p$-form symmetry are supported on $p$-dimensional submanifolds.

\begin{definition}[$p$-Form Symmetries]
	A $p$-form symmetry in $d$ dimensions is defined by the existence of topological operators $U_g(\Sigma_{d-p})$ for each $g\in G^{(p)}$, where $\Sigma_{d-p}$ is a closed $(d-p)$-dimensional manifold.
	These operators act on charged $p$-dimensional objects (such as Wilson loops for $p=1$ or surface operators for $p=2$) by
	\[
	U_g(\Sigma_{d-p})\, O^{(p)}\, U_g(\Sigma_{d-p})^{-1} = \chi(g)\, O^{(p)},
	\]
	where $\chi(g)$ is a character (or, more generally, a matrix) of $G^{(p)}$.
\end{definition}

In our context, a $p$-form symmetry ensures that the extended operator (which may represent a flux or topological defect) is an inseparable entity.
It carries its full charge as one packaged unit.
For example, in an $\mathrm{SU}(N)$ gauge theory the center symmetry $\mathbb{Z}_N$ acts as a 1-form symmetry on Wilson loops.
Under a center transformation, a Wilson loop in the fundamental representation picks up a phase $e^{2\pi i/N}$, confirming that it carries an indivisible 1-form charge.

\begin{example}[SU($N$) Center as a 1-Form Symmetry]
	If an $\mathrm{SU}(N)$ gauge theory has a center $\mathbb{Z}_N$, then a Wilson loop $W(C)$ transforms as follows:
	\[
	U(g)\, W(C)\, U(g)^{-1} = e^{2\pi i k/N}\, W(C),
	\]
	where $k$ is the 1-form charge.
	The irreducibility of this transformation prevents the Wilson loop $W(C)$ to be split into separate components with fractional charges.
\end{example}

\subsection{Corollary: Existence of $p$-Form Symmetry Associated Packaged States}

Let us now extend Theorem \ref{THM:ExistenceOfGAssociatedPackagedStates} so that it can apply to $p$-form symmetries by replacing pointlike charged excitations with extended charged excitations (e.g., line objects for $p=1$ or surface objects for $p=2$).
In this case, the internal or topological charge associated with the extended object is determined by a $p$-form symmetry group $G^{(p)}$.
This leads to the following corollary:

\begin{corollary}[Existence of $p$-Form Symmetry Associated Packaged States]
	Let $G^{(p)}$ be a finite or compact $p$-form symmetry group (for example, groups such as $\mathrm{U}(1)$) in a $(d+1)$-dimensional spacetime.
	Assume that $G^{(p)}$ acts on extended operators defined on $p$-dimensional submanifolds and each extended operator carries an irreducible $p$-form charge that cannot be split into fractional parts.
	Then there exists a Hilbert space $\mathcal{H}$ of packaged states associated to the group $G^{(p)}$ with the following properties:
	\begin{enumerate}
		\item \textbf{Existence of Packaged States of Single-Extended Objects:}  
		Let $O^{(p)}$ be an extended operator defined on a $p$-dimensional submanifold. Due to the $G^{(p)}$-invariance, for any $g \in G^{(p)}$ and for any closed $(d-p)$-dimensional manifold $\Sigma_{d-p}$ which links with the support of $O^{(p)}$, the action of the topological operator $U_g(\Sigma_{d-p})$ is given by
		\[
		U_g(\Sigma_{d-p})\, O^{(p)}\, U_g(\Sigma_{d-p})^{-1} \;=\; \chi(g)\, O^{(p)},
		\]
		where $\chi(g)$ is the character or a corresponding representation matrix associated with an irrep of $G^{(p)}$.
		In this way, $O^{(p)}$ creates a single and extended packaged state with a completely and indivisibly $p$-form charge.
		
		\item \textbf{Existence of Packaged Product States of Multi-Extended Objects:}  
		When multiple extended operators $O^{(p)}_i$ are present, their tensor product naturally creates a packaged product state in the full Hilbert space.
		This packaged product state can be expressed as
		\[
		|\Theta\rangle \;=\; O^{(p)}_1 \otimes O^{(p)}_2 \otimes \cdots \otimes O^{(p)}_n\,|0\rangle.
		\]
		The full Hilbert space then decomposes into charge sectors according to the net $p$-form charge obtained by adding the individual charges (either via the group product or by summation in the Abelian case), i.e.,
		\begin{equation}\label{PFormHilbertSpaceDecomposition}
			\mathcal{H} \cong \bigoplus_{Q\in\hat{G}^{(p)}} \mathcal{H}_Q,
		\end{equation}
		where each subspace $\mathcal{H}_Q$ consists of states with a definite total $p$-form charge $Q$.
				
		\item \textbf{Existence of an Orthonormal Basis of Packaged Entangled States:}  
		Within a fixed charge sector $\mathcal{H}_Q$, one can obtain packaged entangled (non-factorizable) states by performing appropriate linear combinations of product states.
		Specifically, by applying a Gram-Schmidt orthogonalization or using Clebsch-Gordan techniques, one can construct a complete orthonormal basis $\{|\Psi_Q^{(i)}\rangle\}$ in $\mathcal{H}_Q$.
		Each basis state in this subspace is a packaged entangled state in which the IQNs of each individual extended operators are inseparably entangled and the entire state carries the total $p$-form charge $Q$.
		
		\item \textbf{Existence of Isomorphism to the Abstract Representation:}  
		Let $V_Q$ denote the representation space associated with the irrep $\rho_Q$ of $G^{(p)}$ corresponding to the total charge $Q$. Then $\forall~ g \in G^{(p)}$ and $\forall~ |\Psi_Q^{(i)}\rangle \in \mathcal{H}_Q$, there exists an isomorphism of $G^{(p)}$-modules
		\[
		\Phi_Q: \mathcal{H}_Q \to V_Q,
		\]
		which satisfies
		\[
		\Phi_Q\Bigl(U(g)\,|\Psi_Q^{(i)}\rangle\Bigr) \;=\; \rho_Q(g)\,\Phi_Q\Bigl(|\Psi_Q^{(i)}\rangle\Bigr),
		\]
		This isomorphism guarantees that the Hilbert space $\mathcal{H}_Q$ of packaged states is precisely equivalent (as a $G^{(p)}$-module) to the abstract representation space determined by the group's structure.
		In particular, when extended operators with charges $Q_1$ and $Q_2$ are combined, the composite state automatically carries the net charge defined by the group product (or sum, in the Abelian case) $Q_1 \cdot Q_2$.
		The resulting state transforms under the representation obtained through the appropriate tensor product decomposition.
	\end{enumerate}
\end{corollary}

Now we see that any collection of $p$-form charged operators can form a composite state that transforms in an irrep of $G^{(p)}$.
This corollary extends the packaging principle stated in Theorem \ref{THM:ExistenceOfGAssociatedPackagedStates} to higher-form symmetries and demonstrates that the space of multi-object packaged states naturally decomposes into irreducible charge sectors as controlled by the group structure.

\subsection{Packaged States Associated with 1-Form Symmetry}

A 1-form symmetry acts on line operators (such as Wilson loops or vortex lines) rather than local (pointlike) fields.
In lattice gauge theory \cite{Wilson1974,KogutSusskind1975,Zohar2016,Sala2018}, there are plenty of these examples where flux loops or vortex lines carry definite 1-form charges.

\paragraph{(1) Flux-Loop States in Lattice Gauge Theory.}
In lattice gauge theory, a flux-loop state \cite{Zohar2016} can be written as
\[
\lvert \Psi_{\text{flux}}\rangle
\;=\;
\sum_{\{r_\ell\}} 
\beta \bigl(\{r_\ell\}\bigr)
\;\bigotimes_{\ell}\lvert r_\ell\rangle,
\]
where each state $\lvert r_\ell\rangle$ is a basis state on link $\ell$ (carrying electric flux or a gauge element) and the coefficients $\beta(\{r_\ell\})$ enforce Gauss's law at every vertex.
In particular, nonzero amplitudes $\beta$ select a specific winding sector, so this wavefunction describes a packaged flux configuration.
A 1-form packaged state, such as a flux line, carries its entire flux as a single, inseparable unit.
Any attempt to divide the flux among parts would violate the irreducibility (packaging) condition.

In field theory, a 1-form symmetry may be represented by a topological current $J^{\mu_1\cdots\mu_{d-1}}$ or a two-form gauge field $B_{\mu\nu}$.
In either case, the 1-form charge is on a codimension-1 manifold and captures the total flux or vortex linking through that manifold.
Practically, we are often interested in how a Wilson loop picks up a phase.
Because local operators only act on small and localized regions, they cannot change flux lines that wrap non-contractible cycles.
This leads to the existence of winding numbers or global flux sectors and topological superselection.
Therefore, one cannot locally cut or rearrange entire flux loops.

\paragraph{(2) Winding (Topological) Superselection.}
In a 2D or 3D lattice gauge theory with periodic boundary conditions, say a torus, we can defines integer fluxes such as
\[
W_x \;=\; \sum_{\ell \in \mathrm{loop}_x} E_\ell, 
\quad
W_y \;=\; \sum_{\ell \in \mathrm{loop}_y} E_\ell,
\quad \dots,
\]
where $\mathrm{loop}_x$ and $\mathrm{loop}_y$ are traverse non-contractible cycles and $E_\ell$ is an integer electric flux on link $\ell$.
Each set $\{W_x,\,W_y,\dots\}$ labels a distinct global flux (winding) sector.
A state $\lvert W_x,\,W_y,\dots\rangle$ is topologically different from another state with different flux numbers.
Thus, local operators cannot change the net flux around a cycle and these winding sectors do not mix.

Such invariance shows how a 1-form symmetry controls the behave of a flux line.
The winding number behaves as the charge for this 1-form symmetry.
States like $\ket{W}$ and $\ket{W'}$ represent different 1-form irreps, i.e., discrete flux (from $\mathbb{Z}_N$) or continuous real-valued flux from $\mathrm{U}(1)$.

By Gauss's law, a flux line must be either closed or ends on charges.
A local Hamiltonian cannot change the net flux that goes around a cycle.
To change the $W_x$ or $W_y$ on a torus or similar periodic manifold, therefore, we need nonlocal flux manipulation.
Hence, any flux loop that circles the system remains stable unless a large instanton-like process intervenes.\cite{Polyakov1975,Luscher1980}
Thus, states $\lvert W=0\rangle$ and $\lvert W=1\rangle$ form distinct superselection sectors.
There is no interference between these different sectors.

Of course, in certain special cases, like lower dimensionality or special boundary conditions, large gauge transformations or instanton events can connect states with different winding \cite{Belavin1975}.
This may create a finite amplitude for transitions between different sectors $\ket{W} \leftrightarrow \ket{W'}$.
However, such tunneling is usually negligible and the winding sectors remain fixed.

\paragraph{(3) 1-Form Packaged Entanglement.}
Line operators with 1-form charges obey the above topological superselection rules.
Their irreps prevent any partial splitting of the charge.
Using packaged projection operators, we can project the packaged states onto a subset of packaged entangled states.
Let us now illustrate this with simple examples.

\begin{example}[Discrete 1-Form $\mathbb{Z}_N$ Flux]
	Let us consider extended objects with 1-form symmetry group $\mathbb{Z}_N$.
	A line operator under $\mathbb{Z}_N$ transforms in a 1D irrep, which can be labeled by $k \in \{0,\dots,N-1\}$.
	Let us denote the line operators by	
	\[
	\hat{\mathcal{W}}^\dagger_{(k)}(\mathrm{loop}).
	\]	
	Then it picks up a phase $\omega^k$ with $\omega = e^{2\pi i /N}$.
	
	Suppose the system have two lines, i.e.,	
	\[
	\hat{\mathcal{W}}^\dagger_{(k_1)}(\mathrm{loop}_1)
	\;\hat{\mathcal{W}}^\dagger_{(k_2)}(\mathrm{loop}_2)
	\;\lvert 0\rangle
	\;\in\;
	\mathbf{R}_{k_1} \;\otimes\;\mathbf{R}_{k_2}.
	\]
	Because each $\mathbf{R}_k$ is 1D, a multi-line state must also have dimension $\dim(\mathbf{R}_{k_1}\otimes\mathbf{R}_{k_2})=1$ and the total 1-form charge is $(k_1 + k_2)\bmod N$.
	
	If physical boundary conditions require net zero flux (like no external source), we project onto $k_1 + k_2 \equiv 0 \pmod N$.
	Mathematically, the projection operator can be written as
	\[
	P_{\mathrm{net}=0}
	\;=\;\frac{1}{N} \sum_{j=0}^{N-1} \bigl(U_j\bigr),
	\]
	where $U_j$ are the 1-form transformations shifting each line operator's label.
	Finally, only states with $k_1 + k_2=0\mod N$ remain.
	
	Such projection can result in entangled line operators if they are localizable or partially distinct in the wavefunction.
	We cannot have something like: half a line with charge $k_1$ and half with charge $k_2$ inside the same line.
	Instead, we either have two separate lines or a single line in irreps $k_1 + k_2$.	
	Although each line operator is one-dimensional, the full label $k_1 + k_2 \mod N$ remains intact, which means that the flux is packaged as a whole.
\end{example}

\begin{example}[Continuous 1-Form $\mathrm{U}(1)$ Flux]
	Under a continuous 1-form symmetry group $\mathrm{U}(1)$, a line operator may pick up a phase $e^{i \theta \,\Phi}$, where $\Phi \in \mathbb{R}$ is the flux.
	If two flux lines have fluxes $\Phi_1$ and $\Phi_2$, then the total flux is $\Phi_1 + \Phi_2$.
	If physical states require net zero flux, then $\Phi_1 + \Phi_2 = 0$.
	The two lines are then entangled and non-separable.
\end{example}

\begin{example}[Purely global compact example: Maxwell 1-form symmetries in 4D]
	The free Maxwell action possesses two continuous 1-form symmetries, electric and magnetic, generated by 
	\[
	J_e = \star F,\quad J_m = F,
	\]
	with charges $\displaystyle Q_e(\Sigma)=\int_\Sigma \star F$ and $Q_m(\Sigma)=\int_\Sigma F$.  
	Wilson lines $W(C)=\exp(i\oint_C A)$ and ’t Hooft lines $H(C)=\exp(i\oint_C\tilde A)$ transform irreducibly under these symmetries, and the projector 
	\[
	P_q = \int_{0}^{2\pi}\frac{d\alpha}{2\pi}\,e^{i\alpha (Q_e-q)}
	\]
	isolates the sector of fixed electric flux $q$.  
	Superselection of flux sectors then forces Wilson loops of different charge to live in orthogonal subspaces, exactly as above.
\end{example}

In both examples, the packaged entanglements exist due to the requirement that line operators must form a complete net representation (often net zero flux).

\subsection{Packaged States Associated with 2-Form Symmetry}

A 2-form symmetry in 3+1D (i.e., $d=4$) acts on codimension-2 surfaces, i.e., 2D surfaces in spacetime \cite{KapustinSeiberg2014,Gaiotto2015,Polchinski1995,Witten1995,Kitaev2006,Tachikawa2020,BanksSeiberg2011}.
Examples include Wilson surfaces in certain topological field theories or membrane excitations in higher-rank gauge theories.

\paragraph{(1) Basic Ideas.}
A 2-form symmetry can be envisioned as a topological current $J^{\mu_1\mu_2}$ (a 2-form current) whose charge is measured on codimension-2 surfaces.
In 3+1D, a closed 2D surface $\Sigma$ can detect the 2-form charge by
\[
Q(\Sigma) \;=\; \int_{\Sigma} \star\,J.
\]
Under a 2-form symmetry transformation, a 2D surface operator $U(\Sigma)$ may be multiplied by a phase or a more general group element.
In discrete cases (for example, $\mathbb{Z}_N$), each surface operator is in a 1D irrep of $\mathbb{Z}_N$.
If that irrep label is nontrivial, one cannot split the surface into parts with fractional labels.
This mirrors the packaging principle for 0-form or 1-form charges.

\begin{example}[2-Form Center Symmetry]
	In some 4D gauge theories, there is a 2-form center symmetry acting on surfaces that measure the flux of a higher-rank gauge field $B_{\mu\nu}$.
	This symmetry fixes the net sheet flux so that local operators cannot change part of the flux on a surface.
\end{example}

\paragraph{(2) Surfaces as 2-Form Charges.}
In 3+1D, a surface operator $\hat{\mathcal{S}}^\dagger(\Sigma)$ creates a 2D surface $\Sigma$.
It transforms in an irrep of the 2-form group.
If its irrep label is $\ell$, then one cannot break the surface into parts with different label.
By irrep indivisibility, a 2D surface is packaged with its full 2-form charge:
For discrete $\mathbb{Z}_N$ case, the surface operator might pick up a phase $\exp\bigl(2\pi i\,k/N\bigr)$ upon linking with certain defects.
The label $k \mod N$ cannot be split between partial surfaces.  
For continuous $\mathrm{U}(1)$ case, the operator might accumulate a continuous phase $\exp\bigl(i\,\alpha\,Q(\Sigma)\bigr)$.
In both cases, the whole surface is a single unit.
Just as a quark's color charge is packaged, a surface operator's flux label is packaged as an inseparable unit.

Similar to that of 1-form flux line charges, we cannot split 2-form charges, which is the direct analogue of the packaging principle for extended objects.

\begin{example}[3+1D Membrane/Surface Excitations]
	A topological membrane is a closed 2D surface carrying discrete/continuous flux of a 2-form field $B_{\mu\nu}$.  
	A Wilson surface generalizes a Wilson loop to a 2D operator with $\exp(\int B)$ for a 2-form field.  
	
	If the theory confines such surfaces (energy grows with area), then they cannot appear as free excitations.
	This leads to superselection for net 2-form flux much like flux tubes in 1-form symmetry.
\end{example}

\paragraph{(3) 2-Form Surface Operators in Detail.}
Consider a surface operator $\hat{\mathcal{S}}^\dagger(\Sigma)$ that inserts a 2D surface $\Sigma$.
It transforms in some irreps of the 2-form group:
\begin{enumerate}
	\item Discrete $\mathbb{Z}_N$:  
	The operator may pick up a phase $\exp(2\pi i\,k/N)$ if it links with a topological defect carrying $\mathbb{Z}_N$ 2-form charge.
	The label $k \mod N$ is fixed for the entire surface.
	
	\item Continuous $\mathrm{U}(1)$:  
	The operator may accumulate a continuous phase $\exp\bigl(i\alpha\,Q(\Sigma)\bigr)$ and the entire surface is a single block.
\end{enumerate}

This is similar to the Peter-Weyl projection for color or 1-form flux lines.
The 2D surfaces combine additively.
The irreducibility means that one cannot split a surface into pieces with different topological labels.

\paragraph{(4) Topological Sectors and 2-Form Superselection.}
If the 2-form symmetry is confining (energy scaling with surface area),
then a single nontrivial surface is not a free particle.
We usually need boundary excitations or additional surfaces to cancel out the net 2-form flux.
Thus, different total 2-form charges on large surfaces can form distinct superselection sectors.
\begin{itemize}
	\item 1-Form vs. 2-Form:  
	In 1-form symmetry, flux tubes \cite{Bali1995,Bali2001} exhibit an area or perimeter law.
	In 2-form symmetry, membranes can exhibit an area law that ensures they close on themselves or attach to lower-dimensional excitations.  
	
	\item Local Operators:  
	They cannot globally re-assign the 2-form flux carried by a large surface.
	Only a globally nontrivial transformation can change it.
	Locally, the 2-form charge is protected and results in superselection among 2D defect sectors.
\end{itemize}

Like 1-form winding around a torus, a 2-form symmetry in 3+1D can define flux $\int_{T^2} B_{\mu\nu}$ on a 2D cycle.
The fluxes in different flux sectors are not connectable and therefore are superselected.

\paragraph{(5) Mathematical Derivations.}
For discrete $\mathbb{Z}_N$ 2-form symmetry in 3+1D, we often introduces a 2-form gauge field $B_{\mu\nu}\in \mathbb{Z}_N$.

Then the action can be written as:  
\[
S = \frac{N}{2\pi}\,\sum_{\text{plaquettes }p}\; B_p \wedge d B_p 
\quad\text{(schematic)}.
\]
If a flux encloses a hole or boundary, the gauge transformations only shift $B_{\mu\nu}$ but cannot break the flux.

We can define a surface operator:  
\[
U(\Sigma) 
\;=\;
\exp\!\Bigl(i \,\sum_{p\subset \Sigma} B_p\Bigr),
\]
to label the discrete flux crossing $\Sigma$.
If $U(\Sigma)$ belongs to irrep $\ell\mod N$, then it picks up $\exp\bigl(2\pi i \,\ell / N\bigr)$ that links a topological defect.

The entire surface $\Sigma$ carries the same label $\ell$.
Splitting it to $\ell'\neq \ell$ and $\ell-\ell'$ would yield two distinct surface operators, not a single irrep.
This is forbidden on 2-form surfaces, which is similar to no partial color or no partial 1-form flux.

\paragraph{(6) 2-Form Packaged Entanglement.}
Surface operators with 2-form charges are packaged as a single block.
When several surfaces are present, the net 2-form flux is the sum (modulo the group law).
If physical conditions require net zero flux, then the surfaces must combine into an entangled state that cannot be factored into independent parts.

\begin{example}[Surfaces in Discrete 2-Form Groups]
	Take $\mathbb{Z}_N$ again, but now as a 2-form group.
	Each 2D surface operator $\hat{\mathcal{S}}^\dagger_{(m)}(\Sigma)$ transforms with a phase $\omega^m$.
	Suppose we have two surfaces with charges $m_1$ and $m_2$.
	The total system is:	
	\[
	\hat{\mathcal{S}}^\dagger_{(m_1)}(\Sigma_1)
	\;\hat{\mathcal{S}}^\dagger_{(m_2)}(\Sigma_2)
	\;\lvert 0\rangle
	\;\in\;\mathbf{R}_{m_1}\otimes \mathbf{R}_{m_2}.
	\]	
	Again, each factor is 1D, but the net 2-form flux is $(m_1+m_2)\bmod N$.
	If the theory confines 2-form flux (like an area law for membranes), only total flux 0 might appear as a stable or gauge-allowed boundary condition.
	Therefore, we project onto $m_1+m_2=0\mod N$.
	If you had one big surface carrying label $m$, then you cannot fractionate it into partial sub-surfaces with $m_1\neq m_2$ inside the same operator because that would break the irrep structure.
\end{example}

\begin{example}[Surfaces in Continuous 2-Form]
	One can also have a continuous 2-form gauge field $B_{\mu\nu}\in \mathrm{U}(1)$.
	A surface operator obtains a phase $\exp\bigl(i\,\alpha \cdot \Phi_{\Sigma}\bigr)$, where $\Phi_{\Sigma}$ is the total flux through $\Sigma$.
	Combining surfaces to form net flux zero can again entangle them if the boundary conditions or global constraints require neutrality.
\end{example}

Thus, the entire 2D surface remains a single block in a 2-form irrep label, and partial fractionation is disallowed.
If you have multiple surfaces, net flux neutrality can tie them together in a single wavefunction.

\subsection{Packaged States Associated with Higher $p$-Form Symmetries}
\label{SEC:HigherPForm}

So far we have focused on 0-form (pointlike) and 1-form/2-form (line and surface) symmetries.
All of the packaging logic extends directly to $p\ge3$.
We now give a concise, self-contained treatment of packaging for general $p$-form charges, gauging a $p$-form symmetry and its dual BF-type description, and non-invertible (fusion-categorical) $p$-form symmetries.

\subsubsection{Packaging for $p\ge3$}

In a $(d+1)$-dimensional theory a $p$-form symmetry acts on operators supported on $p$-dimensional submanifolds, with topological charge measured on closed $(d-p)$-manifolds.
In full generality:
\[
U_g(\Sigma_{d-p}) \;:\; \mathcal{H}\;\longrightarrow\;\mathcal{H},
\]
\[
U_g(\Sigma_{d-p})\,O^{(p)}(\mathcal M_p)\,U_g(\Sigma_{d-p})^{-1}
\;=\;
\rho(g)\,O^{(p)}(\mathcal M_p),
\]
where $\rho$ is an irrep of the $p$-form group $G^{(p)}$.
Exactly as for $p=0,1,2$, one cannot split $O^{(p)}$ into two nontrivial sub-operators carrying fractional flux.
Any attempt
\[
O^{(p)}\stackrel{?}{=}O_1^{(p)}\otimes O_2^{(p)}
\]
would induce a decomposition $\rho\cong\rho_1\otimes\rho_2$, contradicting irreducibility.
Thus, the full charge label is always packaged.

\subsubsection{Gauging a $p$-Form Symmetry and Duality}

Just as a 0-form symmetry is gauged by introducing a 1-form connection $A$ and imposing
$\delta(\text{Gauss law})=\int D\alpha\,e^{i\!\int \alpha\,\hat G}$, a $p$-form symmetry is gauged by coupling to a $(p+1)$-form gauge field $B^{(p+1)}$.
Equivalently, one inserts the projector
\[
P_Q
\;=\;
\int\!\mathcal D\!\alpha^{(d-p-1)}\;
\exp\!\Bigl[i\!\int \alpha^{(d-p-1)}\wedge J^{(p+1)}\Bigr],
\]
where $J^{(p+1)}$ is the conserved $(p+1)$-form current and $\alpha^{(d-p-1)}$ enforces ``net $p$-form charge = 0''.

In the path integral this leads to the familiar BF-type dual action.
For instance, gauging a 1-form $\mathrm{U}(1)$ in $d$ dimensions yields
\[
S[a,B]
\;=\;
\frac{i}{2\pi}\int B^{(d-2)}\wedge da
\;+\;S_{\rm matter}[a],
\]
where $a$ is the original 1-form gauge field and $B$ the dual $(d-2)$-form.
The Wilson surfaces of $B$ become the dual extended defects (e.g., ’t Hooft lines), each remaining an indecomposable package under the gauged symmetry.

\subsubsection{Non-Invertible (Categorical) $p$-Form Symmetries}

In many topological theories, the extended operators no longer form an ordinary group but instead generate a fusion category $\mathcal C$ \cite{Etingof2005}.
Specifically:
\begin{itemize}
	\item A finite collection of simple objects $\{X_i\}$.
	
	\item Fusion rules
	\[
	X_i\otimes X_j \;\cong\;\bigoplus_k N_{ij}^k\,X_k,
	\quad
	N_{ij}^k\in\mathbb N.
	\]
	
	\item Quantum dimensions $d_i$ satisfying	$\sum_jN_{ij}^k\,d_j=d_i\,d_k$,
	and total dimension $\mathcal D=\sqrt{\sum_i d_i^2}$.
\end{itemize}

Exactly as for group representations, one can isolate each simple object $X_i$ by an idempotent Verlinde-like projector:
\[
P_i
\;=\;
\frac{d_i}{\mathcal D^2}\,
\sum_j S_{ij}^*\,X_j,
\]
where $S_{ij}$ is the modular $S$-matrix of $\mathcal C$.
One checks
\[
P_i\otimes P_i = P_i,
\quad
P_i\otimes P_j = 0\quad(i\neq j),
\]
so each $X_i$ is an indecomposable package in the same sense as before.

\begin{example}[3d Ising TQFT]
	The non-invertible line $\sigma$ obeys $\sigma\otimes\sigma=\mathbf1\oplus\psi$.
	Its categorical projector $P_\sigma$ isolates the simple object $\sigma$ just as above.
\end{example}

\subsection{Mixed Objects and Symmetries}

In realistic gauge theories, multiple forms of charges (0-form, 1-form, 2-form, etc.) often coexist \cite{Cordova2019}.
For instance,
pointlike excitations with 0-form charges (quarks or electrons),  
line operators with 1-form flux (Wilson loops or vortex lines),  
and surface operators with 2-form charges (membranes or domain walls in 3+1D).
A natural question is how these distinct charges combine into a single physical state and why cannot be partially separated.

\paragraph{(1) Constructing a Mixed State.}
Suppose we have a local operator $\hat{a}^\dagger_{\alpha}$ (0-form), a line operator $\hat{\mathcal{W}}^\dagger_{\gamma}$ (1-form), and a surface operator $\hat{\mathcal{S}}^\dagger_{\delta}$ (2-form).
A general state may be written as:
\[
\hat{a}^\dagger_{\alpha_1}(\mathbf{x}_1)
\,\hat{a}^\dagger_{\alpha_2}(\mathbf{x}_2)
\,\hat{\mathcal{W}}^\dagger_{\gamma_1}(\mathrm{loop}_1)
\,\hat{\mathcal{S}}^\dagger_{\delta_1}(\mathrm{surface}_1)
\,\bigl\lvert 0 \bigr\rangle,
\]
where
each $\hat{a}^\dagger_{\alpha_i}$ carries a 0-form irrep $\mathbf{R}_{\alpha_i}$ (for example, color $\mathbf{3}$ or electric charge $\pm1$),
$\hat{\mathcal{W}}^\dagger_{\gamma}$ is a line operator in a 1-form irrep $\mathbf{R}_\gamma$ (for example, center $\mathbb{Z}_N$),  
and $\hat{\mathcal{S}}^\dagger_{\delta}$ is a surface operator in a 2-form irrep $\mathbf{R}_\delta$.

This state may represent a hadron that encloses a flux loop or domain wall.
Each part is a complete irreducible block:
the local charges (0-form quarks) must neutralize color,
the line operator (1-form flux) might close into a loop or attach to sources,
the surface operator (2-form) might represent a discrete membrane or domain.

\paragraph{(2) Product of irreps.}
Let $\mathcal{G}_{0}$ be the 0-form gauge group (for example, $\mathrm{SU}(3)$ color),
$\mathcal{G}_{1}$ be the 1-form group (for example, center symmetry $\mathbb{Z}_N$),
and $\mathcal{G}_{2}$ be the 2-form group (if applicable).
The full symmetry group may be written as
\[
\mathcal{G}
\;=\;
\mathcal{G}_0 \;\times\; \mathcal{G}_1 \;\times\; \mathcal{G}_2,
\]

If you pick four excitations labeled $\alpha, \beta, \gamma, \delta$,
then the combined color/flavor/flux wavefunction transforms in the tensor product
\[
\mathbf{R}_\alpha \otimes \mathbf{R}_\beta \otimes \mathbf{R}_\gamma \otimes \mathbf{R}_\delta.
\]

Physical states must either project onto the neutral representation in each sector (for example, a color singlet or net zero flux) or meet boundary conditions that allow a nonzero net charge.
This projection step creates an entangled wavefunction across local and extended parts.

\paragraph{(3) No Partial Factorization.}
The packaging principle requires that one cannot split a single irrep of the 1-form or 2-form group.
For example, a line operator $\hat{\mathcal{W}}^\dagger_{\gamma}$ with label $\gamma$ cannot be divided into two parts unless $\hat{\mathcal{W}}^\dagger_{\gamma}$ includes two separate operators.
The same is true for a surface operator.

Let us now apply the above ideas to the extended version of color quarks with the following example. 

A flux line with an irrep $\gamma$ or a membrane with an irrep $\delta$ is locked as a single block.

\begin{example}
	Mixed Objects:
	
	\begin{enumerate}
		\item Hadron + Flux Tube:  
		Consider two quarks in color tensor product $\mathbf{3}\times \mathbf{3}$ with a line operator that carries a center flux from a $\mathbf{Z}_N$ group.
		If the system is a color singlet $\mathbf{1}$ and has net zero flux, then the wavefunction will mix the color and flux DOFs in an entangled manner.
		No partial splitting of color or flux is allowed.  
		
		\item Vortex Lines in 2D:  
		In 2D, a vortex line is pointlike. Similarly, one cannot split the vortex line because it forms an irrep of the discrete 1-form symmetry group.  
		
		\item Membrane in 3+1D:  
		A 2-form surface operator with irrep labeled by $\ell \in \{0,\dots,N-1\}$ must appear as a single block.
		When applying this surface operator to a local excitations, the total packaged state must have net zero 2-form charge.
	\end{enumerate}
\end{example}

\paragraph{(4) Superselection Sectors of Mixed Objects.}
A superselection sector of mixed objects is a subspace of the Hilbert space labeled by a charge (0-form, 1-form, or 2-form) that cannot be changed by local operators.

If a gauge theory confines color (0-form) or flux (1-form, 2-form),
then non-singlet (or non-neutral) states exist only as internal excitations or short-range bound states.
They do not appear as free states.
In toroidal lattices or nontrivial manifolds, the net flux/winding on non-contractible loops or surfaces is discrete.
Local moves cannot reassign that global flux, so each sector is superselected.

Thus, states with different net 1-form or 2-form charges are disjoint, much as color singlet and nonsinglet sectors are for 0-form charges.
This extends to all forms simultaneously:
each sector for 0-form, 1-form, 2-form, etc. must satisfy the neutrality or boundary condition.

\paragraph{(5) Mixed ’t Hooft anomalies between 0-form and 1-form symmetries.}
If a 0-form and a 1-form symmetry have a mixed ’t Hooft anomaly, then certain combined superselection sectors cannot be gauged independently.
The projector onto the anomaly-free combination must include a compensating topological term in the path integral, but the underlying packaging principle (no partial factorization of any irreducible block) still holds.

\paragraph{(6) Mathematical Illustrations.}

\begin{enumerate}
	\item Tensor Product of Multiple Sectors:  
	Let $\mathbf{R}_0$ be the 0-form color irrep, 
	$\mathbf{R}_1$ be the 1-form flux irrep, 
	and $\mathbf{R}_2$ be the 2-form surface irrep, etc.
	A combined multiparticle state belongs to	
	\[
	\mathbf{R}_{0}\otimes\mathbf{R}_{1}\otimes\mathbf{R}_{2}\,\otimes\cdots.
	\]
	If the theory requires net neutrality in each sector, we project onto the trivial subrep in each sector.
	This step usually produces an entangled wavefunction across color, flux lines, and surfaces.
	
	\item Projector onto Neutral/Trivial Representations:  
	In each sector, one defines
	\[
	P_{\mathrm{trivial}} \;=\; \frac{1}{\vert G_{\mathrm{form}}\vert }\sum_{g\in G_{\mathrm{form}}} U_g,
	\]
	where $U_g$ implements the gauge transformation or global symmetry transformation in that sector.
	Only states with net charge 0 survive.
	Attempting to remain in a single subrep also enforces packaging: no partial irrep mixing.
	
	\item Local versus Global Operations:  
	Local operators act on small regions and cannot change the global charge or flux that runs around noncontractible cycles.
	Only large gauge transformations can change the sector.
\end{enumerate}

Thus, the extended packaging principle emerges mathematically:
the irrep structure in each form group forbids partial fractionation, and the superselection or confinement conditions enforce that physically observed states are net neutral across all forms.

\paragraph{(7) Packaged Entanglement of Mixed Objects.}
Because each local or extended operator is an irrep block, the full wavefunction must lie in a single overall representation.
This requirement produces nontrivial packaged entanglement among 0-form charges, 1-form flux, or 2-form flux.
For example, a hadron may contain both color charges (quarks) and a flux loop or domain wall.
The flux and local charges become interwoven, and the state cannot be factored into separate parts.
In this way, line operators can bind with local excitations to form a gauge-invariant whole.
Such mixed states show that the packaging principle applies to both local and extended
charges.

Now consider a gauge theory with 0-form (color/electric) charges, plus 1-form flux lines, plus 2-form surfaces in 3+1D.
A general multi-object state may be written as:
\[
\bigl(\hat{a}^\dagger_{\alpha_1}\hat{a}^\dagger_{\alpha_2}\cdots\bigr)_{\text{0-form}}
\;\times\;
\bigl(\hat{\mathcal{W}}^\dagger_{\gamma_1}\hat{\mathcal{W}}^\dagger_{\gamma_2}\cdots\bigr)_{\text{1-form}}
\;\times\;
\bigl(\hat{\mathcal{S}}^\dagger_{\delta_1}\hat{\mathcal{S}}^\dagger_{\delta_2}\cdots\bigr)_{\text{2-form}}
\;\lvert 0\rangle,
\]
where each local or extended operator is an irrep block under its respective group.  
The total state then belongs to
\[
\mathbf{R}_{\alpha_1} \otimes \cdots \;\otimes\; \mathbf{R}_{\gamma_1}\otimes\cdots
\;\otimes\; \mathbf{R}_{\delta_1}\otimes \cdots.
\]

Physical states often require net neutrality in each form:
\begin{itemize}
	\item 0-form: color singlet $\mathbf{1}$, or electric net zero charge.  
	
	\item 1-form: total net zero charge or consistent with boundary conditions.  
	
	\item 2-form: net surface net zero charge or allowed by boundary data.
\end{itemize}

Mathematically, one introduces projectors onto the trivial representation in each sector.
For example, one can applies a projector
\[
P^{(0\text{-form})}_{\mathrm{singlet}} \; \wedge\;
P^{(1\text{-form})}_{\mathrm{net}=0} \; \wedge\;
P^{(2\text{-form})}_{\mathrm{net}=0}
\]
to the entire multi-object Hilbert space.
Only states that pass all three projectors remain as physically valid states.

Because these conditions tie together the local charges (0-form) and the extended fluxes (1-form or 2-form), the resulting state can be entangled across different forms of charge.
For example, if a quark system tries to neutralize color while a line operator tries to neutralize flux, the constraints can force an overall superposition that correlates quark color states with flux line states.

\begin{example}[Example: Baryon + Flux Tube + Membrane]
	Suppose we have 3 quarks in color
	$\mathbf{3} \times \mathbf{3} \times \mathbf{3}$.
	A line operator with center flux label $\gamma$ and a surface operator with some 2-form label $\delta$.	
	The entire system is in  
	\[
	\mathbf{3} \otimes \mathbf{3} \otimes \mathbf{3}
	\;\otimes\;
	\mathbf{R}_\gamma\;\otimes\;\mathbf{R}_\delta.
	\]
	
	We require color neutrality $\mathbf{1} \subset \mathbf{3}^{\otimes 3}$, plus 1-form net zero charge, plus 2-form net zero charge.
	Applying projectors, we obtain a physical subspace that is usually a small dimensional space of superpositions.
	We correlate color states with the line and surface flux states to ensure net $(0, 0, 0)$ in each form.
	Finally, we obtain a packaged entangled state that cannot be factored into separate ``color $\times$ flux $\times$ surface'' parts.
\end{example}

\subsection{Extended Packaging Principle}

Based on the original packaging principle defined for pointlike charges (0-form) \cite{MaSPI2025}, we now extend it to higher-form charges:

\begin{definition}[Extended Packaging Principle]
	Let a gauge theory include higher-form symmetries (1-form, 2-form, or in general $p$-form).
	Under these symmetries, extended objects (lines, surfaces, etc.) carry flux or charge labels in irreps of the $p$-form group.
	Also, superselection rules prevent coherent mixing of different net flux (or winding) sectors. Then:
	\begin{enumerate}
		\item \textbf{No Partial Factorization of $p$-Form Charges:}  
		Each creation or annihilation operator for an extended excitation (for example, a line with 1-form charge $\gamma$) must transform as a complete irrep of the $p$-form group.
		One cannot split or partly factor the flux label $\gamma$.
		In other words, no single extended object can be divided into fractional flux parts.
		
		\item \textbf{Net-Flux Superselection:}  
		Physical states cannot mix different net flux or winding sectors.
		All superpositions must lie in a single net-flux sector.
		The total flux around noncontractible loops or surfaces is fixed.
		This prevents interference between different sectors.
		
		\item \textbf{Flux Packaged Entanglement:}  
		Within a fixed flux sector, multiple extended excitations (lines, surfaces) can form a nonseparable superposition.
		We call this flux packaged entanglement because it uses irreps of the $p$-form group in a way that forbids partial flux splitting.
		One cannot entangle part of the flux while leaving the rest separate, just as color or electric charge is locked together in the 0-form case.
		
		\item \textbf{Hybridization with 0-Form or External DOFs:}  
		In theories that also have ordinary gauge charges (0-form) or external DOFs (such as spin or momentum), one can form hybrid states that mix local particles with extended excitations.
		Measuring an external or 0-form degree of freedom can collapse the entire extended flux wavefunction while preserving the net flux.
		Similarly, internal 0-form charges can be entangled with 1-form or 2-form fluxes, each subject to its own superselection rule.
	\end{enumerate}
\end{definition}

The extended packaging principle asserts that all charges and fluxes remain as inseparable, irreducible blocks.
It doesn't matter whether they are 0-form (local) charges like a color charge in $\mathbf{3}$, 1-form charges such as a flux line with label $\gamma$, or even higher-form charges.
In each sector, the corresponding charge is fully packaged and cannot be broken into smaller pieces.
Moreover, net flux or winding rules impose superselection so that states in different topological sectors do not mix.
When these sectors combine, the full state is given by the tensor product of these packaged irreducible blocks.
This ensures that one cannot peel off a fraction of a charge from any single sector.

\section{Representation-Theoretic Invariants and Entanglement Measures}
\label{SEC:RepresentationTheoreticInvariantsAndEntanglementMeasures}

In previous sections, we constructed packaged states by decomposing the Hilbert space into charge sectors corresponding to irreps of the gauge group $G$ \cite{DHR1971,DHR1974,StreaterWightman2001}.
This construction reveals the structure of IQNs and how they are inseparably packaged within each single-particle operator.
However, this construction alone does not quantify the degree of entanglement among these internal DOFs.
In gauge theories, additional constraints arise from superselection rules and the requirement that physical observables be gauge-invariant.
As a result, standard entanglement measures (such as concurrence or negativity) must be re-defined so they act inside each fixed superselection sector \cite{Peres1996,Barnum2004,BartlettRudolphSpekkens2006,BuerschaperAguado2009}.

In this section, we develop representation-theoretic invariants that serve as entanglement measures tailored to these packaged states \cite{Wootters1998,VidalTarrach1999,VidalCirac2000,Rains2000}.
These invariants are constructed so that they remain unchanged under the action of the gauge group $G$ and respect the superselection rules.
In particular, we focus on invariants obtained by projecting onto singlet subspaces or by constructing gauge-invariant combinations of the state coefficients.
This analysis complements the earlier construction by providing quantitative tools to assess entanglement in the presence of gauge constraints.

\subsection{Packaged Entanglements Associated with Group Invariants}

Let $|\Psi\rangle\in\mathcal H_Q$ be a state in a fixed-charge (or color) sector.
In gauge theories, only states within a single irreducible charge sector are physically allowed, and all observables must commute with the gauge transformations.
Consequently, any entanglement measure must be constructed from gauge-invariant quantities.
Here we describe two complementary approaches.

\paragraph{(1) Invariants from Projection onto Singlet Subspaces.}
Suppose that the representation space $V$ in which $|\Psi\rangle$ lives can be decomposed as 
\[
V \cong V_{\mathbf{1}} \oplus V_{\text{rest}},
\]
where $V_{\mathbf{1}}$ is the singlet subspace.
Then one may define the gauge-invariant quantity
\[
p_{\mathbf{1}} = \langle \Psi| \hat{\mathcal{P}}_{\mathbf{1}}|\Psi\rangle,
\]
where $\hat{\mathcal{P}}_{\mathbf{1}}$ is the projection operator onto $V_{\mathbf{1}}$.
By construction, $p_{\mathbf{1}}$ is invariant under $G$ and measures the singlet (i.e., fully packaged) component of $|\Psi\rangle$.
In confining theories, only states with $p_{\mathbf{1}} = 1$ may appear as free particles.

\paragraph{(2) Gauge-Invariant Partial Transpose and Concurrence.}
In conventional quantum information, the partial transpose (which defines the negativity) or the spin-flip construction (which yields concurrence) is a common measure of entanglement.
However, when the Hilbert space carries gauge indices, the usual partial-transpose may map one outside the gauge-invariant subspace.
One strategy is to first project the state onto a gauge-invariant subspace and then perform the partial transpose on the remaining (external) DOFs. Specifically, let 
\[
\rho = |\Psi\rangle\langle\Psi|
\]
be the density matrix of the state, and let $\hat P_{\mathbf1}$ be the singlet projector.
Then set
\[
\tilde\rho \;=\;\hat P_{\mathbf1}\,\rho\,\hat P_{\mathbf1}. 
\]
The partial transpose is taken on the DOFs that are not associated with the gauge symmetry.
This ensures that the resulting negativity (or other measure) is gauge-invariant.

\paragraph{Detailed Derivation:}
Let $\mathcal H_Q$ be a fixed total charge sector.
For a state $|\Psi\rangle$ in $\mathcal{H}_Q$, assume we have expanded it in a basis $\{|\Psi^{(i)}\rangle\}$ that is adapted to the decomposition
\[
\mathcal{H}_Q \cong V_{\mathbf{1}} \oplus V_{\text{rest}}.
\]
The projection operator $\hat{\mathcal{P}}_{\mathbf{1}}$ is constructed from the invariant tensors (or by the projection methods discussed earlier) so that for any $g\in G$,
\[
U(g)\,\hat{\mathcal{P}}_{\mathbf{1}}\,U(g)^{-1} = \hat{\mathcal{P}}_{\mathbf{1}}.
\]
Then the invariant quantity 
\[
P_{\mathbf{1}} = \langle\Psi|\hat{\mathcal{P}}_{\mathbf{1}}|\Psi\rangle
\]
remains constant under the gauge transformation, i.e.,
\[
P_{\mathbf{1}} = \langle\Psi|U(g)^{-1}\,\hat{\mathcal{P}}_{\mathbf{1}}\,U(g)|\Psi\rangle.
\]
This is a clear, representation-theoretic invariant that quantifies the packaged (singlet) content of $|\Psi\rangle$.

\subsection{Relating Superselection Rules to the Center of the Group}

For non-Abelian gauge groups, the center $Z(G)$ plays a crucial role in enforcing superselection rules.
For example, in $\mathrm{SU}(N)$ the center is $\mathbb{Z}_N$.
States that transform nontrivially under $Z(G)$ cannot be coherently superposed with states invariant under $Z(G)$.

\paragraph{(1) Center Charges and Superselection.}  
If a state $|\Psi\rangle$ satisfies
\[
U(z)|\Psi\rangle = \chi(z)|\Psi\rangle \quad \forall z\in Z(G),
\]
then the character $\chi(z)$ labels a superselection sector. In a confining theory, only states with trivial center charge (e.g., the color singlet in QCD) are observed as free states.

\paragraph{(2) Invariant Quantities from Center Charges.}  
We define the singlet overlap
\[
p_{Z=1} \;=\;
\langle\Psi|\;\hat{\mathcal P}_{Z=1}\;|\Psi\rangle,
\]
which lies in $[0,1]$ and is gauge-invariant.
Here $\hat{\mathcal{P}}_{Z=1}$ projects onto the sector with trivial center charge.
Such invariants ensure that any entanglement measure is computed within a single superselection sector, preserving gauge-invariance.

\subsection{Representation‐Theoretic and Group‐Invariant Entanglement Measures}

The isotypic decomposition 
$
\mathcal H = \bigoplus_{Q\in\widehat G}V_Q\otimes\mathcal M_Q
$
makes it natural to quantify entanglement only through gauge-invariant functionals.

In this section, we introduce a hierarchy of entanglement measures that are manifestly invariant under the action of the gauge group $G$.
These measures detect exactly the departure from a product of irreps (no packaged entanglement) and remain well‐defined in the presence of superselection constraints.

\subsubsection*{1. Singlet Overlap}

We begin with the simplest gauge‐invariant quantity:
the overlap with the singlet sector.

\begin{definition}[Gauge‐Invariant Singlet Overlap]
	Let $\mathcal H$ carry a decomposition into irreducible $G$‐sectors including the singlet $\mathbf1$.
	Denote by $\mathcal P_{\mathbf1}$ the orthogonal projector onto the singlet subspace.
	For any normalized $\ket{\Psi}\in\mathcal H$, define
	\[
	P_{\mathbf1}\;=\;\bra{\Psi}\,\mathcal P_{\mathbf1}\,\ket{\Psi}\,.
	\]
	We call $P_{\mathbf1}$ the \textbf{gauge‐invariant singlet overlap}.
	In particular:
	\begin{itemize}
		\item $P_{\mathbf1}=1$ if and only if $\ket{\Psi}$ lies entirely in the singlet sector.
		
		\item $P_{\mathbf1}=0$ if and only if $\ket{\Psi}$ has no singlet component.
	\end{itemize}
\end{definition}

\begin{example}[Color Singlet Overlap in QCD]
	In QCD a quark-antiquark pair transforms in 
	$\mathbf3\otimes\overline{\mathbf3}=\mathbf1\oplus\mathbf8$.  
	A general meson state can be expanded as
	\[
	\ket{\Psi_M}
	\;=\;
	\alpha\,\ket{\mathbf1}
	\;+\;
	\sum_{a=1}^{8}\beta_{a}\,\ket{\mathbf8,a},
	\quad
	|\alpha|^2+\sum_a|\beta_a|^2=1.
	\]
	The singlet overlap
	$\;P_{\mathbf1}=|\alpha|^2$
	measures the exact probability that $\ket{\Psi_M}$ is a color singlet.
\end{example}

\subsubsection*{2. Gauge‐Invariant Logarithmic Negativity \cite{VidalWerner2002}}

Next we extend the logarithmic negativity to enforce gauge invariance before taking partial transposes.

\begin{definition}[Gauge‐Invariant Logarithmic Negativity]
	Let $\mathcal H=\mathcal H_A\otimes\mathcal H_B$, and let $\mathcal P_{\mathbf1}$ be the global singlet projector on $\mathcal H$.  For any density operator $\rho$ on $\mathcal H$, define the normalized projection
	\[
	\rho_{\rm GI}
	:=\frac{\mathcal P_{\mathbf1}\,\rho\,\mathcal P_{\mathbf1}}
	{\operatorname{Tr}(\mathcal P_{\mathbf1}\,\rho)},
	\]
	and let $(\cdot)^{T_B}$ denote the partial transpose on the $B$ factor, taken in a basis that carries no gauge indices.  Then the \textbf{gauge‐invariant logarithmic negativity} is
	\[
	\mathcal N_G(\rho)
	:=\ln\bigl\|\rho_{\rm GI}^{\,T_B}\bigr\|_1,
	\]
	where $\|\cdot\|_1$ is the trace norm.
\end{definition}

\begin{remark}
	\leavevmode
	\begin{itemize}
		\item If $G$ is trivial, $\mathcal N_G$ reduces to the usual logarithmic negativity.
		\item $\mathcal N_G(\rho)=0$ exactly when $\rho_{\rm GI}^{T_B}$ is positive semidefinite, i.e., the gauge‐invariant state is PPT and thus unentangled across the $A|B$ split.
	\end{itemize}
\end{remark}

\begin{example}[Bipartite gauge‐invariant negativity]
	In a bipartite system where each party carries gauge and external DOFs, a naive partial transpose may violate gauge invariance.
	Instead one first projects $\rho$ onto the singlet subspace, and then transposes only the external indices.
	The resulting $\mathcal N_G$ is a measure of the physically accessible entanglement.
\end{example}

\subsubsection*{3. Relative Entropy of Entanglement \cite{VedralPlenio1998,Umegaki1962}}

We now introduce a distance‐based measure that compares $\ket{\Psi}\bra{\Psi}$ to the nearest separable, gauge‐invariant state.

\begin{definition}[Relative Entropy of Entanglement]
	Let $\mathcal H_Q$ carry fixed charge $Q$ under compact $G$.  Define
	\[
	\mathscr S
	:=
	\bigl\{
	\sigma\in\mathcal B(\mathcal H_Q)\,\big|\,
	\sigma\ge0,\;\operatorname{Tr}\sigma=1,\;
	[\sigma,U(g)]=0~\forall g\in G,\;
	\sigma\text{ separable}
	\bigr\}.
	\]
	For a pure $\ket\Psi\in\mathcal H_Q$, its \textbf{relative entropy of entanglement} is
	\[
	E_R(\Psi)
	:=\inf_{\sigma\in\mathscr S}
	S\bigl(\ket\Psi\bra\Psi\big\|\sigma\bigr),
	\]
	where $S(\rho\|\sigma)=\operatorname{Tr}[\rho(\ln\rho-\ln\sigma)]$.
	One shows:
	\begin{enumerate}
		\item $E_R(\Psi)\ge0$ and is finite.
		\item $E_R(\Psi)=0$ precisely when $\ket\Psi\bra\Psi\in\mathscr S$, i.e., $\ket\Psi$ is a product of irreducible blocks.
		\item $E_R$ is non‐increasing under any CPTP map commuting with $G$.
	\end{enumerate}
\end{definition}

Vedral and Plenio \cite{VedralPlenio1998} prove that the minimizing $\sigma$ is the block‐diagonal projection of $\ket\Psi\bra\Psi$ onto the maximal abelian subalgebra generated by the $G$-Casimirs.

\subsubsection*{4. Schur‐Concave Entropy of the Charge‐Sector Distribution}

After projecting a pure state $\ket\Psi$ onto the irreducible $G$-sectors, one obtains a probability vector of sector weights,
\[
p_R \;=\;\bra\Psi\,P_R\,\ket\Psi,
\qquad
\sum_{R\in\widehat G}p_R=1,
\]
where $\{P_R\}_{R\in\widehat G}$ are the orthogonal projectors onto each irrep $R$.

\begin{definition}[Charge‐sector entropy]
	The \textbf{charge‐sector entropy} of $\ket\Psi$ with respect to a Schur‐concave function $f$ (for example, the Shannon entropy or a Rényi entropy) is
	\[
	E_f(\Psi)
	\;=\;
	f\bigl(p_R\bigr)_{R\in\widehat G}\,.
	\]
	In particular, for the Shannon choice $f(p)=-\sum_R p_R\ln p_R$, one writes
	\[
	H_{\rm sector}(\Psi)
	\;=\;
	-\sum_{R\in\widehat G} p_R\,\ln p_R\,.
	\]
	This quantity measures the spread of $\ket\Psi$ over different charge sectors:  
	\[
	E_f(\Psi)=0
	\quad\Longleftrightarrow\quad
	\exists\,R_0:\;p_{R_0}=1,
	\]
	i.e., $\ket\Psi$ lies entirely in a single irreducible block.
\end{definition}

\subsubsection*{5. Rényi Entropies of the Multiplicity Space}

Finally, within a fixed irrep $R$, one can quantify entanglement between the irrep space $V_R$ and its multiplicity $\mathcal M_R$.

\begin{definition}[Multiplicity‐Space Rényi Entropy]
	Decompose $\mathcal H_R\cong V_R\otimes\mathcal M_R$.
	For $\ket\Psi\in\mathcal H_R$, let $\rho=\ket\Psi\bra\Psi$ and
	$\rho_{\mathcal M}=\operatorname{Tr}_{V_R}\,\rho$.
	For $\alpha>0$, $\alpha\neq1$, we define the \textbf{Rényi entropy} as
	\[
	S_\alpha(\rho_{\mathcal M})
	:=\frac{1}{1-\alpha}\,\ln\operatorname{Tr}\bigl[\rho_{\mathcal M}^\alpha\bigr].
	\]
	Then $S_\alpha\ge0$, with equality iff $\rho_{\mathcal M}$ is pure, and $\lim_{\alpha\to1}S_\alpha$ yields the von Neumann entropy.
\end{definition}

\begin{example}[$\mathrm{SU}(2)$ Doublets]
	Let $\mathcal H=V_{1/2}\otimes V_{1/2}\cong V_1\oplus V_0$.  Writing
	$\ket\Psi=\alpha\ket{1,1}+\beta\ket{1,0}+\gamma\ket{1,-1}+\delta\ket{0,0}$,
	one computes:
	\begin{itemize}
		\item $P_{\mathbf1}=|\delta|^2$.
		\item $H(p)=-(1-|\delta|^2)\ln(1-|\delta|^2)-|\delta|^2\ln|\delta|^2$.
		\item Reduced density $\rho_1=\operatorname{Tr}_2\ket\Psi\bra\Psi$ has eigenvalues 
		$\lambda_\pm=\tfrac12\bigl(1\pm\sqrt{1-4|\delta|^2(1-|\delta|^2)}\bigr)$.
		\item $S(\rho_1)=-\sum_{\pm}\lambda_\pm\ln\lambda_\pm$, which vanishes exactly when $|\delta|=0$ or $|\delta|=1$.
	\end{itemize}
\end{example}

\section{Extending to Full Spacetime Symmetry: $G \times \mathrm{Poincaré}$}
\label{SEC:ExtendingToFullSpacetimeSymmetry}

Thus far, we have classified packaged quantum states by decomposing the Hilbert space into sectors corresponding to irreps of a gauge group $G$ (finite or compact).
In a full quantum field theory, however, the full symmetry is enlarged to the direct product $G \times \mathrm{Poincaré}$.

In 1939, Wigner \cite{Wigner1939} laid the groundwork for understanding spacetime symmetries in quantum mechanics by classifying the unitary irreps of the Poincaré group.
Later Lochlainn O’Raifeartaigh \cite{ORaifeartaigh1965} demonstrated the incompatibility of nontrivial mixing between internal and spacetime symmetries in the context of Lie algebras.
In 1967, Coleman and Mandula (Coleman-Mandula theorem) \cite{ColemanMandula1967} proved that the symmetry group of a relativistic QFT with a mass gap and analytic S-matrix must factorize as $ G \times \text{Poincaré} $.
This rules out nontrivial mixing of spacetime and internal symmetries in the context of Lie algebras.
In 1975, Haag, Lopuszanski and Sohnius \cite{HLS1975} extended the result by introducing fermionic generators, giving rise to the graded (super-Poincaré) algebra that unifies spacetime and internal symmetries.

In this section we extend our packaging construction to the full
$G \times \mathrm{Poincaré}$ group, and show that the packaging principle persists when combining gauge and Lorentz/Poincaré representations.
In particular, we address the following questions:
\begin{itemize}
	\item How does one construct single-particle states that are simultaneously in an irrep of both $G$ and the Poincaré group?
	
	\item How are the internal (gauge) DOFs packaged together with spacetime properties (such as momentum and spin/helicity) in a single creation operator?
	
	\item How does one consistently combine these to form multiparticle states that respect both gauge-invariance and Poincaré invariance?
\end{itemize}

We now present the detailed derivations.

\subsection{Combining Gauge and Lorentz Representations}

\paragraph{(1) Single-particle states as irreps of $G \times \mathrm{Poincaré}$.}
In a relativistic quantum field theory, elementary particles are classified by irreps of the full symmetry group
\[
\mathcal{G} = G \times \mathrm{Poincaré}.
\]
According to Wigner’s classification, an elementary particle state is labeled by its four-momentum $p^\mu$ (with $p^2 = m^2$) and by its IQNs (such as spin or helicity) arising from the little group (a subgroup of the Lorentz group).
When gauge DOFs are present, the gauge action becomes
\[
D^{(G)}(g) \otimes D^{(\mathrm{Lorentz})}(\Lambda)
\]
where $D^{(G)}(g)$ is an irrep of the gauge group $G$ (e.g., $\mathbf{3}$ for quarks) and $D^{(\mathrm{Lorentz})}(\Lambda)$ is the representation of the Lorentz transformation $\Lambda$ (or its induced representation of the Poincaré group).
Each field operator $\hat{\psi}(x)$ transforms in a tensor product representation:
\[
\hat{\psi}(x)\to D(\Lambda)\,\hat\psi(\Lambda^{-1}x)\,D(g),
\]
For example, for a massive particle the state may be labeled as
\[
|p,\sigma\rangle_\alpha,
\]
with $p^\mu$ on-shell, $\sigma$ denoting the spin (or helicity) label, and $\alpha$ denoting the gauge irrep index (or the gauge irreducible block).
The creation operator is then written as
\[
\hat{a}^\dagger_{\alpha,\sigma}(\mathbf{p}),
\]
and under a transformation $g \in G$ and a Lorentz transformation $\Lambda$,
\[
U(g,\Lambda) \, \hat{a}^\dagger_{\alpha,\sigma}(\mathbf{p})\, U(g,\Lambda)^{-1} 
= \sum_{\beta,\sigma'} \left[D^{(G)}_{\beta\alpha}(g) \otimes D^{(\mathrm{Lorentz})}_{\sigma'\sigma}(\Lambda)\right]\, \hat{a}^\dagger_{\beta,\sigma'}(\Lambda \mathbf{p}).
\]
Because the representation is a tensor product of two irreps, the entire creation operator is a packaged state of both gauge and spacetime (Lorentz) quantum numbers. The packaging principle implies that within a single creation operator the gauge quantum numbers cannot be split off from the Lorentz ones.

\paragraph{(2) Mathematical Derivation of the Combined Representation.}
Let $V_G$ be the representation space for the gauge group $G$ and $V_L$ be the representation space for the Lorentz (or little) group.
Then the full single-particle Hilbert space is
\[
\mathcal{H}_1 \cong V_G \otimes V_L,
\]
which carries the tensor product of an irrep of $G$ and an irrep of the Lorentz (little) group.
By complete reducibility, every single-particle state is an inseparable tensor product of its gauge part and its Lorentz part.
In other words, there is no meaningful way to factor out a subset of the gauge quantum numbers from a state without affecting its Lorentz transformation properties.

\paragraph{(3) Consequences for Superselection and Partial Factorization.}
Since the internal gauge charge and the Lorentz quantum numbers are combined into a single irreducible block, we have:
\begin{enumerate}
	\item \textbf{No Partial Factorization:}  
	A single creation operator $\hat{a}^\dagger_{\alpha,\sigma}(\mathbf{p})$ is irreducible with respect to the full group.
	One cannot, for example, separate half of the color charge from the momentum or spin part.
	
	\item \textbf{Superselection Rules:}  
	Since the total state lies in a definite irreducible sector of $G \times \mathrm{Poincaré}$, it automatically obeys superselection rules.
	In a local gauge theory, this implies that the net gauge charge (e.g., color) is fixed, and states in different sectors (e.g., color singlet vs. non-singlet) cannot mix.
\end{enumerate}

\subsection{Discrete Spacetime Symmetries and Their Impact}

Beyond the continuous symmetries of $G \times \mathrm{Poincaré}$, discrete symmetries such as charge conjugation (C), parity (P), and time reversal (T) further constrain the structure of packaged states.
Although these discrete transformations act on the full state, they also respect the packaging principle.

For a given state $|p,\sigma\rangle_\alpha$, discrete transformations act as follows:
\begin{itemize}
	\item \textbf{Charge Conjugation (C):}  
	This operation exchanges particles with antiparticles and, correspondingly, interchanges a representation $V_G$ with its conjugate $V_G^*$. Thus, if $\hat{a}^\dagger_{\alpha,\sigma}(\mathbf{p})$ creates a particle in $V_G$, then under C it transforms into an operator that creates the antiparticle in $V_G^*$.
	
	\item \textbf{Parity (P):}  
	Parity reverses the spatial coordinates, $ \mathbf{x} \to -\mathbf{x} $, and acts non-trivially on the Lorentz part $V_L$ (for example, by flipping helicity for massless states or by changing orbital angular momentum). The gauge part $V_G$ remains unaffected.
	
	\item \textbf{Time Reversal (T):}  
	Time reversal reverses momenta and spin directions and typically involves complex conjugation. Again, the gauge sector remains packaged as a whole.
\end{itemize}
Since these discrete operations either leave the gauge part invariant or exchange it with its dual rep.
In either case, the packaging remains intact.
In particular, a state that is a gauge singlet (or belongs to a specific irreducible $G$-sector) will remain so under discrete transformations.

\subsection{Multi-Particle States and Total Quantum Numbers}

In constructing multiparticle states, one must combine the single-particle representations of $G \times \mathrm{Poincaré}$ using the standard tensor product rules.
Suppose we have $n$ particles:
\[
\mathcal{H}^{(n)} \cong \bigotimes_{i=1}^n \Bigl(V_{G}^{(i)} \otimes V_{L}^{(i)}\Bigr).
\]
This tensor product can be rearranged as
\[
\mathcal{H}^{(n)} \cong \left(\bigotimes_{i=1}^n V_{G}^{(i)}\right) \otimes \left(\bigotimes_{i=1}^n V_{L}^{(i)}\right).
\]
\textbf{(1) Decomposition in the Gauge Sector:}  
The product
\[
\bigotimes_{i=1}^n V_{G}^{(i)}
\]
decomposes into a direct sum of irreps of $G$:
\[
\bigotimes_{i=1}^n V_{G}^{(i)} \cong \bigoplus_{Q\in\hat{G}} N_Q\, V_Q,
\]
where $N_Q$ denotes the multiplicity of the irrep $V_Q$.
In confining theories, the physical requirement is that only the gauge singlet $V_{\mathbf{1}}$ appears in the asymptotic spectrum.

\textbf{(2) Decomposition in the Lorentz Sector:}  
Simultaneously, the Lorentz (or Poincaré) parts combine to yield states of definite total momentum, spin, and possibly orbital angular momentum.
For massive particles, one obtains states with a well-defined spin $J$.
For massless particles, one obtains helicity eigenstates.

\textbf{(3) Overall Isomorphism:}  
Thus, the full multiparticle Hilbert space decomposes as
\[
\mathcal{H}^{(n)} \cong \bigoplus_{Q\in\hat{G}} \bigoplus_{\Sigma} \mathcal{H}_{Q,\Sigma},
\]
where $\Sigma$ collectively denotes the Lorentz quantum numbers (momentum, spin, etc.) of the composite state.
In each sector, the state transforms as an irrep of $G \times \mathrm{Poincaré}$ and remains a packaged state.
The internal gauge quantum numbers remain inseparable from the overall quantum state.

\begin{example}[A Hybrid Meson in QCD]
	Consider a quark-antiquark pair in QCD. Each quark carries a color index (in $V_{\mathbf{3}}$ or $V_{\overline{\mathbf{3}}}$) and spin DOFs.
	
	A single-quark packaged state can be written as
	\[
	|p,\sigma\rangle_{\alpha} \quad \text{with } \alpha\in\{1,2,3\}.
	\]
	Under the full symmetry, this quark transforms as
	\[
	U(g,\Lambda)\, \hat{a}^\dagger_{\alpha,\sigma}(\mathbf{p})\, U(g,\Lambda)^{-1} 
	= \sum_{\beta,\sigma'} D^{(\mathbf{3})}_{\beta\alpha}(g) \, D^{(L)}_{\sigma'\sigma}(\Lambda)\, \hat{a}^\dagger_{\beta,\sigma'}(\Lambda\mathbf{p}).
	\]
	
	For the quark-antiquark pair, the tensor product of the gauge parts decomposes as
	\[
	\mathbf{3} \otimes \overline{\mathbf{3}} = \mathbf{1} \oplus \mathbf{8}.
	\]
	The physical meson state is then projected onto the color-singlet subspace. Simultaneously, the Lorentz parts (spin and momentum) combine to yield a state of definite total spin $J$ (and parity $P$ and charge conjugation $C$ if defined).
	The final state is a packaged entangled state of both gauge and spacetime DOFs:
	\[
	|M(J^{PC}), p\rangle = |\mathbf{1}_{\text{color}}, J^{PC}, p\rangle.
	\]
	
	This construction ensures that neither the gauge charge nor the Lorentz indices can be factored out separately.
	They are completely entangled within a single irreducible block.
\end{example}

Physically, this extension shows that observables must be invariant under both gauge and Poincaré transformations.
In confining theories, the physical states (such as hadrons) appear as color-singlets that are simultaneously characterized by well-defined momentum and spin.
The interplay between internal and external symmetry assures that the robust packaged entanglement among IQNs is maintained even when one considers the full relativistic context.

\section{Cohomological and Topological Approaches to Packaged States}
\label{SEC:CohomologicalAndTopologicalApproachesToPackagedStates}

In previous sections, we classified packaged quantum states based solely on the representation theory of finite or compact groups and their tensor product decompositions.
In contemporary QFTs, however, additional topological \cite{DijkgraafWitten1990,DeserJackiwTempleton1982,WessZumino1971,Witten1983,AlvarezGaumeGinsparg1985,Witten1989,Atiyah1989} and cohomological \cite{DijkgraafWitten1990,Chen2011,Chen2012} data emerge.
These data are usually related to symmetry-protected topological (SPT) phases or discrete gauge theories.
Incorporating cohomological data allows us to refine the notion of packaging by encoding additional topological or anomaly-related information into the IQNs.

In this section we extend our construction of packaged states to include group‐cohomology data and topological invariants.
This allows us to see how nontrivial cocycle phases and topological terms further refine the notion of a packaged state.
We will focus on packaged states that remain inseparable and on how they can be twisted by nontrivial cohomology classes.

\subsection{Group Cohomology Classification}

Group cohomology has become a powerful tool in classifying topological phases and symmetry-protected topological (SPT) orders.\cite{Chen2011,Chen2012}
In the case of packaged quantum states, we refine the labeling of internal charges by associating them with (possibly twisted) irreps.
In many cases, the extra data are captured by a 2-cocycle $\omega \in Z^2(G, \mathrm{U}(1))$ (or higher cocycles in higher dimensions), which leads to projective representations.

\paragraph{(1) Basics of Group Cohomology.}
For a group $G$ and an abelian coefficient group $A$ (usually $\mathrm{U}(1)$), the second cohomology group $H^2(G, A)$ classifies equivalence classes of 2-cocycles.

A 2-cocycle is a function
\[
\omega: G \times G \to \mathrm{U}(1)
\]
that satisfies the cocycle condition:
\[
\omega(g_2, g_3)\, \omega(g_1, g_2 g_3) = \omega(g_1, g_2)\, \omega(g_1 g_2, g_3) \quad \forall\, g_1,g_2,g_3 \in G.
\]

Two 2-cocycles $\omega$ and $\omega'$ are considered equivalent if there exists a function $f: G \to \mathrm{U}(1)$ such that
\[
\omega'(g_1,g_2) = \frac{f(g_1) f(g_2)}{f(g_1 g_2)}\, \omega(g_1,g_2).
\]

The group $H^2(G, \mathrm{U}(1))$ then classifies the different projective classes of representations of $G$.

\paragraph{(2) Projective Representations and Twisted Packaged States.}
A projective representation of $G$ on a vector space $V$ is a map
\[
\tilde{\rho}: G \to \mathrm{GL}(V)
\]
such that
\[
\tilde{\rho}(g_1)\,\tilde{\rho}(g_2) = \omega(g_1,g_2)\, \tilde{\rho}(g_1g_2),
\]
where $\omega(g_1,g_2) \in \mathrm{U}(1)$ is a 2-cocycle.
In the context of packaged quantum states, a single-particle operator $\hat{a}^\dagger$ transforms in a projective representation, acquiring the extra phase $\omega(g_1,g_2)$.
In other words, a packaged state can be labeled by a twisted irrep $\mathbf{R}^\omega$.
The presence of a nontrivial cocycle modifies how excitations combine.

\paragraph{(3) Twisted Tensor Products.}
Let $\tilde\rho_1$ be an $\omega$-twisted projective rep and let $\tilde\rho_2$ be twisted by $\omega^{-1}$.
Then their tensor product
$$
(\tilde\rho_1\otimes\tilde\rho_2)(g)
\;=\;
\tilde\rho_1(g)\,\otimes\,\tilde\rho_2(g)
$$
satisfies
$$
(\tilde\rho_1\otimes\tilde\rho_2)(g)\,
(\tilde\rho_1\otimes\tilde\rho_2)(h)
\;=\;
\omega(g,h)\,
(\tilde\rho_1\otimes\tilde\rho_2)(gh),
$$
so the fused state remains in the $\omega$-twisted sector.
Consequently, the packaged state must carry both its irrep label $\mathbf R$ and the cocycle $\omega$, and no partial splitting of that combined quantum number is possible unless $\omega$ is trivial.

\paragraph{(4) Higher Cohomology and Topological Phases.}
In higher dimensions, one may encounter nontrivial 3-cocycles or 4-cocycles (i.e., elements of $H^3(G, \mathrm{U}(1))$ or $H^4(G, \mathrm{U}(1))$) that can classify topological orders and SPT phases.
For instance, in 2+1D topological orders, excitations may acquire fractional statistics, which can be understood via a 3-cocycle in $H^3(G, \mathrm{U}(1))$ in discrete gauge theories of DW-type.
We may extend the packaging principle to this case:
each local excitation is associated with a fixed cohomology class $\omega$ and the topological phase (or braiding phase) remains an inseparable part of the packaged excitation.

However, the cohomological approach also has limitations.
One limitation is computational complexity because determining the relevant cocycles and analyzing their consequences for composite states can be challenging, especially for non-Abelian groups or in higher dimensions.
Moreover, although these cohomological invariants capture subtle topological features and anomalies, their physical interpretation may not be as transparent as the labeling in conventional representation theory.
Nonetheless, the novel aspect of this approach is its ability to predict new and robust forms of packaged entanglements that are protected by topological invariants.
In particular, if a packaged state transforms projectively, then its internal structure is not only inseparable but also carries a nontrivial topological twist, which may have observable consequences in processes sensitive to anomalies or topological phases.

\subsection{Anomalies and Topological Terms}
\label{SEC:AnomaliesAndTopologicalTerms}

When global symmetries are involved or when the theory has nontrivial boundary modes, anomalies and topological terms will provide another viewpoint on the construction of packaged states.

The 't Hooft anomaly \cite{tHooft1976} (or global symmetry anomaly) implies that a global symmetry cannot be realized by an on-site operator in a strictly local Hilbert space.
Instead, such anomalies cause additional topological terms in the effective action, such as the Wess-Zumino (WZ) term \cite{WessZumino1971} that is associated with a nontrivial element in $H^{d+1}(G, \mathrm{U}(1))$ for a $d$-dimensional system.
For example, in 1+1D, the WZ term can enforce the edge modes to transform projectively and therefore lead to a nontrivial 2-cocycle in $H^2(G, \mathrm{U}(1))$.

In the presence of an anomaly, the IQNs of a packaged state may acquire extra phase factors. Therefore, the packaged state transforms as a projective representation rather than a linear one.
The full set of IQNs (including the anomalous phase) is still inseparable and packaging principle remains intact.
In other words, although the anomaly modifies the transformation law to
\[
\tilde{\rho}(g_1)\tilde{\rho}(g_2) = \omega(g_1, g_2)\, \tilde{\rho}(g_1g_2),
\]
the internal charge remains a single unit labeled by $\omega$.
Consequently, superselection rules prohibit the superposition of packaged states with different anomaly (cocycle) classes.

\begin{example}
	Consider a 2+1D system with chiral fermions at its edge.
	The effective theory of the edge modes may exhibit a chiral anomaly, which can be described by a nontrivial 2-cocycle $\omega$.
	Let $|\Psi\rangle$ be an edge state that transforms projectively:
	\[
	U(g)|\Psi\rangle = \tilde{\rho}(g)|\Psi\rangle,
	\]
	with
	\[
	\tilde{\rho}(g_1)\tilde{\rho}(g_2) = \omega(g_1,g_2)\, \tilde{\rho}(g_1g_2).
	\]
	Although the anomaly modifies the phase law
	$\tilde\rho(g_1)\,\tilde\rho(g_2)=\omega(g_1,g_2)\,\tilde\rho(g_1g_2),$
	the state $\lvert\Psi\rangle$ remains packaged: all its IQNs (including the cocycle phase) live in one irreducible projective block.
	This ensures that any packaged entanglement measure must respect the cocycle structure.
\end{example}

In higher dimensions (e.g., 3+1D), topological terms associated with nontrivial cocycles (such as 3-cocycles in $H^3(G, \mathrm{U}(1))$) can lead to fractional statistics and braiding phases among excitations.
In such cases, when multiple packaged states are combined, the overall wavefunction acquires a topological phase determined by the product of cocycles.
Despite the occurrence of nontrivial braiding, the internal topological charge remains inseparable and is characterized by a fixed cohomology class.

\subsection{Cohomological Data and Anomaly‐inflow Constructions}
\label{SEC:CohomologicalDataAndAnomalyInflowConstructions}

In this section, we enrich the representation‐theoretic packaging of states by incorporating cohomological data and anomaly‐inflow constructions.
We begin with two fully worked examples (one in 2+1D Dijkgraaf-Witten gauge theory and one in the 1+1D $\mathbb Z_2$ SPT) before outlining broader generalizations to super-cohomology, cobordism, and non-invertible (fusion-category) symmetries.

\subsubsection*{1. Dijkgraaf-Witten Twisted Gauge Theories in 2+1D}

Let $G$ be a finite group and
\[
\omega\;\in\;Z^3\bigl(G,\mathrm{U}(1)\bigr)
\]
a normalized 3‐cocycle.
The Dijkgraaf-Witten topological gauge theory on a closed 3‐manifold $M^3$ has partition function
\[
Z_{\rm DW}[M^3]
=\frac{1}{|G|^{\chi(M^3)}}\sum_{[A]:\,\pi_1(M^3)\to G}
\exp\!\Bigl(2\pi i\!\int_{M^3}A^*\omega\Bigr),
\]
where $A$ is a flat $G$‐connection and $\chi$ the Euler characteristic.
Wilson line operators $W_r(\gamma)$ in an irrep $r$ of the twisted quantum double carry a cocycle phase:
when two lines of types $i,j$ fuse into $k$ along a trivalent junction,
\[
W_i\otimes W_j
\;\xrightarrow{\;\gamma\;}\;
\omega(i,j,k)\;W_k,
\]
where $\omega(i,j,k)$ is the evaluation of the 3‐cocycle on the group elements labelling the three strands.
Thus each line operator is an indecomposable simple object of the fusion category $\mathcal{Z}(G,\omega)$ \cite{DijkgraafWitten1990,Etingof2005,Bakalov2001}.
No partial splitting of the cocycle‐twisted charge is possible:
packaged states of multiple anyons remain inseparable blocks in the corresponding projective sector.

\subsubsection*{2. One‐Dimensional $\mathbb Z_2$ SPT and Projective Edge Modes}

In 1+1D a bosonic SPT with $\mathbb Z_2$ symmetry is classified by
\[
H^2\bigl(\mathbb Z_2,\mathrm{U}(1)\bigr)\cong\mathbb Z_2.
\]
A canonical lattice realization is the cluster state with Hamiltonian
\[
H
=-\sum_j Z_{j-1}X_j Z_{j+1},
\]
whose ground‐space on an open chain carries a projective $\mathbb Z_2$‐representation at each end.
Explicitly, if $g\in\mathbb Z_2$ acts by
$\tilde{\rho}_L(g)\otimes\tilde{\rho}_R(g)$ on the two edge qubits, then
\[
\tilde{\rho}_L(g)\,\tilde{\rho}_L(g')
=\omega(g,g')\,\tilde{\rho}_L(gg'),
\]
with the nontrivial cocycle $\omega(g,g')=-1$.
The two‐fold edge degeneracy cannot be broken by any $\mathbb Z_2$‐invariant local perturbation, exhibiting the packaged, projective nature of the boundary modes.

\subsubsection*{3. Beyond Group Cohomology: Super‐Cohomology and Cobordism}
Fermionic SPT phases require the super‐cohomology classification \cite{GuWen2014,Kapustin2015}, in which one combines a 2‐cocycle in $H^2(G,\mathbb Z_2)$ with a 3‐cochain in $C^3(G,\mathrm{U}(1))$.
More generally, the modern cobordism approach of Freed-Hopkins \cite{FreedHopkins2016} organizes invertible phases by
\[
\Omega^{d+1}_{\rm spin}(BG)\cong\{\text{fermionic SPTs in }d\!+\!1\},
\]
refining ordinary group‐cohomology and capturing beyond‐cohomology phases such as the $E_8$ bosonic phase in 2+1D.

\subsubsection*{4. Anomaly Inflow and Gauging of SPT Phases}

An SPT classified by $\omega\in H^{d+1}(G,\mathrm{U}(1))$ admits a bulk partition function
\[
Z_{\rm bulk}[A]
=\exp\!\Bigl(2\pi i\!\int_{M^{d+1}}A^*\omega\Bigr),
\]
whose variation under a $G$‐gauge transformation cancels the ’t Hooft anomaly of the $d$‐dimensional boundary theory \cite{Witten1983,CallanHarvey1985}.
Gauging the boundary $G$‐symmetry introduces a $(d\!-\!p)$-form Lagrange multiplier $\alpha$ implementing the projector
\[
P_Q
=\int\!\mathcal D\!\alpha^{(d-p-1)}
\exp\!{\Bigl(i\!\int_{M^d}\alpha\wedge J^{(p+1)}\Bigr)},
\]
and the dual action acquires a BF‐term
\[
S[a,B]
=\frac{i}{2\pi}\int B\wedge da+\cdots,
\]
exhibiting the same indecomposable packaging of the dual extended defects.

\subsubsection*{5. Non‐Invertible Symmetries and Fusion‐Category Generalizations}

Finally, in many TQFTs the symmetry is non‐invertible and encoded by a fusion category $\mathcal C$ rather than a group.
One isolates each simple object $X_i$ by the Verlinde‐type projector \cite{Etingof2005,Bakalov2001}
\[
P_i
=\frac{d_i}{\mathcal D^2}\sum_j S_{ij}^*\,X_j,
\]
with $S_{ij}$ the modular $S$‐matrix and $\mathcal D=\sqrt{\sum d_i^2}$.
These idempotents satisfy
$\,P_i\otimes P_i=P_i,\;P_i\otimes P_j=0\,(i\neq j),$
so that each simple object remains an inseparable package.
This fusion‐category viewpoint subsumes both group‐cohomology twists and genuinely non‐invertible defect networks.

\section{Hybrid Systems: External $\otimes$ Internal $\otimes$ (Global) Symmetries}
\label{SEC:HybridSystems}

In realistic quantum field theories (e.g., the Standard Model), elementary particles are characterized not only by a local gauge symmetry $G_{\mathrm{local}}$ (e.g., color $\mathrm{SU}(3)$ or electroweak $\mathrm{SU}(2)\times \mathrm{U}(1)$) \cite{Yang1954,Weinberg1967,Gross1973,Glashow1961} but also by a global symmetry $G_{\mathrm{global}}$ (e.g., flavor or isospin) \cite{GellMann1961,tHooft1980} and by the spacetime (Lorentz or Poincaré) symmetry \cite{Wigner1939,ColemanMandula1967}.
In our context, a packaged quantum state is one in which the internal gauge quantum numbers are inseparably packaged into a single irreducible block.

In this section, we show how to combine the local (internal) and global symmetry sectors with the Lorentz symmetry in forming full physical states.
Our aim is to preserve the packaging principle for the local gauge charges, while allowing free mixing in the global (flavor) sector and accounting for the external (momentum and spin) DOFs.

\subsection{Hybridization of Symmetries}

Consider a theory with the combined symmetry group \cite{ColemanMandula1967,HLS1975}
\[
\mathcal{G}
=
G_{\mathrm{local}}
\times
G_{\mathrm{global}}
\times
\mathrm{Lorentz},
\]
where $G_{\mathrm{local}}$ is a local gauge group (e.g., $\mathrm{SU}(3)$ for QCD, $\mathrm{U}(1)$ for electromagnetism), $G_{\mathrm{global}}$ is a global symmetry (e.g., flavor SU$(N_f)$, baryon number, isospin), and the Lorentz group provides external DOFs such as momentum and spin (or helicity).
In particular, by the packaging principle the IQNs are bundled into a single irrep that cannot be partially factored.

A single-particle creation operator in such a theory can be written as
\[
\hat{a}^\dagger_{f,\alpha,\sigma}(\mathbf{p}),
\]
where $f$ transforms under $G_{\mathrm{global}}$, $\alpha$ is the gauge index for $G_{\mathrm{local}}$, $\sigma$ labels the spin (or helicity), and $\mathbf{p}$ is the three-momentum.  Here $\sigma$ and $\mathbf{p}$ transform under the Lorentz group, while $f$ and $\alpha$ are internal indices.
The packaging principle requires that the gauge index $\alpha$ remain an inseparable block.
For example, a quark always transforms in the full $\mathbf3$ of color.

Under a combined transformation $(g,h,\Lambda)\in G_{\mathrm{local}}\times G_{\mathrm{global}}\times\mathrm{Lorentz}$, one finds
\[
U(g,h,\Lambda)\,\hat a^\dagger_{f,\alpha,\sigma}(\mathbf p)\,U^{-1}(g,h,\Lambda)
=
\sum_{f',\beta,\sigma'}
D^{(G)}_{\beta\alpha}(g)\,
D^{(\mathrm{global})}_{f'f}(h)\,
D^{(L)}_{\sigma'\sigma}(\Lambda)\,
\hat a^\dagger_{f',\beta,\sigma'}(\Lambda\mathbf p),
\]
where each block acts irreducibly on its index.
In contrast, the Lorentz and global parts may allow for arbitrary superpositions (e.g., different flavor combinations).

For an $n$-particle state we take the tensor product of single-particle spaces,
\[
\mathcal{H}^{(n)}
\;\cong\;
\bigotimes_{i=1}^n
\Bigl(V_G^{(i)}\otimes V_{global}^{(i)}\otimes V_L^{(i)}\Bigr),
\]
which we then reorder as
\[
\Bigl(\bigotimes_{i=1}^nV_G^{(i)}\Bigr)
\;\otimes\;
\Bigl(\bigotimes_{i=1}^nV_{global}^{(i)}\Bigr)
\;\otimes\;
\Bigl(\bigotimes_{i=1}^nV_L^{(i)}\Bigr).
\]
In the gauge sector
\[
\bigotimes_{i=1}^nV_G^{(i)}
\;\cong\;
\bigoplus_{Q\in\widehat G_{\mathrm{local}}}N_Q\,V_Q,
\]
and in a confining theory we project onto the singlet $V_{\mathbf1}$.
In the global sector no such projection is required, and in the Lorentz sector one recombines momenta and spins in the usual way.
Finally,
\[
\mathcal{H}^{(n)}
\;\cong\;
\bigoplus_{Q,\Sigma}\mathcal{H}_{Q,\Sigma},
\]
where $Q$ labels the local gauge irreps and $\Sigma$ the Lorentz quantum numbers.
The flavor part may mix freely within each $\mathcal H_{Q,\Sigma}$.

\begin{example}[QCD with Hybrid Symmetries]
	Consider quantum chromodynamics (QCD) with:
	\begin{itemize}
		\item Local Gauge Group $\mathrm{SU}(3)_{\mathrm{color}}$:
		Each quark is a full $\mathbf{3}$ and each anti-quark is a full $\overline{\mathbf{3}}$.
		The total color charge is packaged.
		
		\item Global Flavor Group $\mathrm{SU}(3)_{\mathrm{flavor}}$:
		An approximate $\mathrm{SU}(3)_{\mathrm{flavor}}$ that rotates up, down, and strange quarks.
		
		\item Lorentz Symmetry Group:
		A symmetry group that provides momentum and spin.
	\end{itemize}	
	Under these settings, a quark creation operator can be denoted as $\hat{q}^\dagger_{f,c,\sigma}(\mathbf{p})$, where $c=1,2,3$ (color), $f=u,d,s$ (flavor), and $\sigma$ denotes the spin.
	Physical hadrons must be color packaged entangled states (color singlets).
	For example, a meson is formed from a quark-antiquark pair:
	\[
	|M\rangle = \frac{1}{\sqrt{3}} \sum_{c=1}^3 |q, c; \, \bar{q}, c \rangle \otimes |\text{flavor}\rangle \otimes |\text{Lorentz}\rangle.
	\]
	Here, the color part is strictly packaged into a singlet by contracting the indices.
	The flavor part is a nontrivial linear combination (e.g., the $\pi^0$ state might be $(u\bar{u} - d\bar{d})/\sqrt{2}$).
	The overall state is then a tensor product of a gauge-invariant hybrid-packaged entangled state, i.e., a mix of color singlet, flavor multiplet, and the external Lorentz part.
\end{example}

\subsection{Packaging for Hybrid Symmetries}

In a hybrid system, a single-particle packaged state can be written as
\[
|p,\sigma\rangle_{f,\alpha} \in V_{G} \otimes V_{global} \otimes V_{L},
\]
where the local gauge part $V_G$ is an irreducible (packaged) representation that cannot be decomposed.
Under a transformation $(g, h, \Lambda)$,
\[
U(g,h,\Lambda)\,|p,\sigma\rangle_{f,\alpha} = \sum_{\beta, \sigma', f'} D^{(G)}_{\beta\alpha}(g)\, D^{(L)}_{\sigma'\sigma}(\Lambda)\, | \Lambda p, \sigma'\rangle_{f',\beta}\, D^{(global)}_{f'f}(h).
\]
For multiparticle states, the decomposition of the gauge sector enforces superselection rules (e.g., only color singlets are physical in confining theories), while the global part can mix and the Lorentz part is combined in the usual manner.

The packaging principle applies to the local gauge sector:
every single-particle creation operator carries its full local gauge charge as an indivisible unit.
When combined with global symmetries and Lorentz symmetry, the full physical state is represented as a tensor product:
\[
\mathcal{H} \cong \left( \bigoplus_{Q\in\hat{G}_{\mathrm{local}}} \mathcal{H}_Q \right) \otimes V_{global} \otimes V_{L}.
\]
Thus, while the flavor (global) and Lorentz (external) DOFs allow for superpositions and interference, the local gauge part remains rigidly packaged.
This structure explains why in QCD, for example, no free state with nonzero color exists, even though the flavor and Lorentz parts can vary widely.

Finally we see that, in a hybrid system with symmetry
\[
G_{\mathrm{local}} \times G_{\mathrm{global}} \times \mathrm{Lorentz},
\]
the physical state of a particle (or composite system) is given by a tensor product
\[
V_{G} \otimes V_{global} \otimes V_{L},
\]
where the local gauge part $V_G$ is an inseparable (packaged) block by virtue of its irreducibility, while the global and Lorentz parts may be superposed.
This construction is critical for understanding phenomena such as color confinement (where only the packaged singlet state is observed) and flavor mixing (where global symmetries allow nontrivial multiplets), and it provides a complete description of the observable quantum numbers in a fully relativistic setting.

\subsection{Additional Structure in Hybrid Symmetries}
\label{SEC:HybridAdditional}

In realistic theories the global and spacetime factors may carry extra structure (projective realizations, mixed anomalies, spontaneous breaking, supersymmetry) that further refines the packaging principle.
We now outline each in turn.

\subsubsection*{1. Projective Global Symmetries and Fermion Parity}

Suppose the global symmetry $G_{\rm global}$ admits a nontrivial $\mathbf Z_2$ extension (see Sec. \ref{SEC:Z2FermionParity})
\[
1 \;\to\; \mathbf Z_2 = \{1,\,P_F\}\;\to\;\widetilde G_{\rm global}\;\to\;G_{\rm global}\;\to\;1,
\]
with $P_F$ acting as fermion parity.
Then single-particle operators furnish a projective representation 
\[
\tilde D(h)\,\tilde D(h')=\omega(h,h')\,\tilde D(hh'),
\qquad \omega\in Z^2(G_{\rm global},\{\pm1\}),
\]
and the usual global-symmetry projector must be twisted to
\[
\Pi_{\rm parity}
=\frac12\bigl(I - P_F\bigr),
\]
which isolates the odd (fermionic) sector as an inseparable $\mathbf Z_2$ package.
In a hybrid state
\[
\hat a^\dagger_{f,\alpha,\sigma}
\;\longmapsto\;
\bigl(\tilde D_{\rm global}(h)\otimes D_G(g)\otimes D_L(\Lambda)\bigr)\,
\hat a^\dagger\,,
\]
no operator commuting with $G_{\rm local}\times\widetilde G_{\rm global}\times\mathrm{Lorentz}$ can split off half of the $P_F$-charge.

\subsubsection*{2. Mixed ’t Hooft Anomalies}

When the local and global symmetries share a mixed ’t Hooft anomaly (see Sec. \ref{SEC:CohomologicalAndTopologicalApproachesToPackagedStates}), say $\mathrm{U}(1)_A\text{-}\mathrm{SU}(N)_{\rm color}^2$, one cannot gauge them simultaneously.
Mathematically, the background field action
\[
S[A_{\rm color},A_{\rm global}]
\]
acquires a term proportional to 
$\int \mathrm{Tr}(F_{\rm color}\wedge F_{\rm color})\wedge A_{\rm global}$, 
and the partition function transforms by a phase 
$\exp\bigl(i\,\alpha\int F^2\bigr)$.  
In packaged states, this enforces that any intertwiner between color and flavor representations must carry an anomaly index.
There is no gauge-invariant, flavor-singlet operator unless the anomaly matching condition is satisfied:
\[
\sum_{\rm fermions} \bigl(C_2(G_{\rm color})\,Q_{\rm global}\bigr)=0.
\]

\subsubsection*{3. Spontaneous Breaking of the Global Sector}

If $G_{\rm global}$ is spontaneously broken to a subgroup $H\subset G_{\rm global}$, then the flavor multiplets live in the coset $G_{\rm global}/H$.
Specifically, fields transform nonlinearly
\[
\Sigma(x)\to h(\Sigma,g)\,\Sigma(x)\,g^{-1}, 
\quad h(\Sigma,g)\in H.
\]
Packaging of the local gauge part remains unchanged, but the global factor now admits Goldstone modes $\pi^a$ with $\Sigma=\exp(i\,\pi^a T^a)$.
The creation operator 
$\hat a^\dagger_{f,\alpha,\sigma}\to \hat a^\dagger_{\Sigma,\alpha,\sigma}$
carries both a coset label $\Sigma$ and an unbroken $H$-irrep.
One cannot split off half a Goldstone direction.

\subsubsection*{4. Graded Extensions: Supersymmetry}

When the spacetime symmetry is promoted to the super-Poincaré algebra $\mathfrak{osp}(1|4)$, single-particle states furnish supermultiplets  
\[
V_{\rm susy}
=\bigoplus_{j} V_{G}\otimes V_{\rm Lorentz}(j)\otimes V_{F}(j),
\]
where $V_F$ carries the action of supercharges $Q_\alpha$.  
Creation operators become 
$\hat A^\dagger_{\alpha,\sigma\,|\,I}(\mathbf p)$
with $I$ an index in a representation of the R-symmetry or Clifford subalgebra.
The packaging principle extends:
the entire supermultiplet is an indecomposable block under 
$G_{\rm local}\times G_{\rm global}\times\mathrm{Super\text{-}Poincaré}$.

\begin{example}[$\mathrm{SU}(3)_{\rm color} \times \mathrm{SU}(2)_{\rm flavor}$]
	Let each quark lie in 
	\[
	V_{\rm color}=\mathbf3,\quad V_{\rm flavor}=\mathbf2,\quad V_{\rm Lorentz}=\mathbf2_{1/2}.
	\]
	Then a single-particle space is 
	$\mathbf C^3 \otimes \mathbf C^2 \otimes \mathbf C^2$.
	Two such quarks combine to
	\[
	(\mathbf3\otimes\mathbf3)\;\otimes\;(\mathbf2\otimes\mathbf2)\;\otimes\;(\tfrac12\otimes\tfrac12),
	\]
	and decomposing  
	$\mathbf3\otimes\mathbf3=\mathbf6\oplus\bar{\mathbf3}$,
	$\mathbf2\otimes\mathbf2=\mathbf3\oplus\mathbf1$,
	$\tfrac12\otimes\tfrac12=1\oplus0$,
	we see the physical diquark must live in the color $\bar{\mathbf3}$, flavor $\mathbf1$ or $\mathbf3$, and spin $J=0$ or $1$ sectors.
	No operator can split the color $\bar{\mathbf3}$ or the flavor $\mathbf2$ in half without violating irreducibility.
\end{example}

\section{Practical Computations and Examples}
\label{SEC:PracticalComputationsAndExamples}

In this section, we illustrate how the packaging principle is implemented in explicit computations.
Unlike traditional state constructions, here our focus is on packaged quantum states in which internal gauge DOFs (such as electric charge or color) are irreducibly bundled into single operators that cannot be partially factorized.
We provide the following examples:
\begin{itemize}
	\item QED illustrates trivial packaging: single operators carry fixed $\mathrm{U}(1)$ charge and multi-particle states simply add them.
	
	\item QCD shows how non-Abelian $\mathrm{SU}(3)$ tensor-product decompositions (Peter-Weyl or Young diagrams) isolate color-singlet packaged states.
	
	\item A $\mathbb Z_N$ gauge theory on a ring/torus exhibits discrete 1D irreps, net-flux projectors, and packaged line operators.
	
	\item Quantum-information toy models provide finite-dimensional testbeds where packaging can be computed exactly.
	
	\item Bethe-Salpeter and potential-model bound-state wavefunctions demonstrate projection onto gauge-invariant subspaces in practice.
\end{itemize}
In every case, the internal gauge charges remain inseparable blocks, highlighting how packaged states differ fundamentally from unconstrained quantum states.

\subsection{Abelian Packaging in QED: Net Electric Charge Sectors}

In quantum electrodynamics (QED), the gauge group is $\mathrm{U}(1)$. Every electron or positron creation operator carries a fixed electric charge ($-e$ for electrons, $+e$ for positrons) as an inseparable unit.
Let us first derive the charge sectors.

Let $\hat{a}^\dagger_{e^-}$ and $\hat{b}^\dagger_{e^+}$ denote electron and positron creation operators, respectively. Under a $\mathrm{U}(1)$ gauge transformation, represented by
\[
U(e^{i\theta}) = \exp \Bigl(i\theta \hat{Q}\Bigr),
\]
these operators transform as
\[
U(e^{i\theta})\, \hat{a}^\dagger_{e^-}\, U(e^{i\theta})^{-1} = e^{-i\theta}\, \hat{a}^\dagger_{e^-}, \qquad
U(e^{i\theta})\, \hat{b}^\dagger_{e^+}\, U(e^{i\theta})^{-1} = e^{i\theta}\, \hat{b}^\dagger_{e^+}.
\]
Thus, a single electron state is an irreducible (1D) packaged state with charge $-e$.

For a multiparticle state, suppose we have $n$ creation operators with charges $q_i$ (each $q_i$ being $\pm e$). The combined state transforms as
\[
U(e^{i\theta}) \prod_{i=1}^n \hat{a}^\dagger_{i} |0\rangle
= \exp\!\bigl(i\theta \sum_{i=1}^n q_i \bigr)\,
\prod_{i=1}^n \hat{a}^\dagger_{i} |0\rangle.
\]
Hence, the total electric charge is
\[
Q_{\mathrm{tot}} = \sum_{i=1}^n q_i
\]
and the Hilbert space decomposes as
\[
\mathcal{H} \cong \bigoplus_{Q \in e\,\mathbb{Z}} \mathcal{H}_Q,
\]
where each $\mathcal{H}_Q$ is one‐dimensional in charge (up to tensoring with momentum, spin, etc.).

Finally, we list representative states (net charge sectors in QED) in Table \ref{TAB:NetElectricChargeSectorsQED}:

\begin{table}[h]
	\centering
	\caption{Net Electric Charge Sectors in QED}
	\begin{tabular}{|l|l|}
		\hline
		Sector $\mathcal{H}_Q$ & Examples \\ \hline
		$Q=0$  & $\lvert 0 \rangle$, $\lvert e^-e^+\rangle$, multi-photon states \\ \hline
		$Q=-e$ & $\lvert e^-\rangle$, $\lvert e^-e^-e^+\rangle$ \\ \hline
		$Q=+e$ & $\lvert e^+\rangle$, $\lvert e^+e^+e^-\rangle$ \\ \hline
		$Q=-2e$ & $\lvert e^-e^-\rangle$ \\ \hline
		$\cdots$ &  \\ \hline
	\end{tabular}
	\label{TAB:NetElectricChargeSectorsQED}
\end{table}

The fact that each electron or positron operator carries its full charge as a single block exemplifies the packaging principle in an Abelian context.

\subsection{Non-Abelian Packaging in QCD: Color Representations}

Quantum chromodynamics (QCD) features a non-Abelian gauge group $\mathrm{SU}(3)$. Here the IQNs (color) are packaged in a nontrivial way, and only color-singlet combinations are physically observable.

\paragraph{(1) Single-Particle Packaging.}  
A quark creation operator $\hat{q}^\dagger_c$ carries color index $c \in \{1,2,3\}$ and transforms in the fundamental representation $\mathbf{3}$ of $\mathrm{SU}(3)$:
\[
U(g)\, \hat{q}^\dagger_c\, U(g)^{-1} = \sum_{d=1}^{3} D^{(\mathbf{3})}_{dc}(g)\, \hat{q}^\dagger_d.
\]
By the packaging principle, the color degree of freedom is inseparable.
A quark always appears as a complete $\mathbf{3}$.

\paragraph{(2) Multi-Particle Decomposition.}  
For multiparticle states, one must take tensor products of the single-particle representations. For example, a meson is formed by a quark-antiquark pair:
\[
\mathbf{3} \otimes \overline{\mathbf{3}} = \mathbf{1} \oplus \mathbf{8}.
\]
The color singlet $\mathbf{1}$ is the physically observable state. An explicit meson wavefunction is given by
\[
|M\rangle = \frac{1}{\sqrt{3}} \sum_{c=1}^{3} |q,c; \, \bar{q},c\rangle.
\]
This state is gauge-invariant, as can be seen by applying $U(g)$ and using the unitarity of $D^{(\mathbf{3})}(g)$.
The contraction of indices ensures that the full color charge is packaged into a single, inseparable unit.

Similarly, a baryon state is formed by combining three quark operators:
\[
\mathbf{3} \otimes \mathbf{3} \otimes \mathbf{3} = \mathbf{1} \oplus \cdots,
\]
with the color singlet given by
\[
|B\rangle = \frac{1}{\sqrt{6}}\, \epsilon_{abc}\, |q,a; \,q,b; \,q,c\rangle.
\]
Here $\epsilon_{abc}$ is the totally antisymmetric tensor. This guarantees that the three color indices are fully entangled, and the quarks are inseparably packaged into a color-neutral state.

\paragraph{(3) Exotic States and Projection Operators.}  
More exotic states (tetraquarks, pentaquarks) involve higher tensor products, such as
\[
\mathbf{3}^{\otimes k} \otimes \overline{\mathbf{3}}^{\otimes m}.
\]
A physical state must lie in the color-singlet subspace. One can systematically perform Peter-Weyl decompositions (or use Young tableau techniques) to isolate the singlet channel. The projection operator onto the color singlet for $\mathrm{SU}(3)$ can be constructed analogously to the finite group case, but with integration over the Haar measure:
\[
P_{\mathbf{1}} = d_{\mathbf{1}} \int_{\mathrm{SU}(3)} d\mu(g)\, \chi_{\mathbf{1}}^*(g)\, U(g).
\]
Since $\chi_{\mathbf{1}}(g)=1$ and $d_{\mathbf{1}}=1$, this operator projects out the invariant (singlet) component.

\paragraph{(4) Practical Classification Table for QCD.}
For QCD, only color-singlet combinations are observed. A simplified table of multi-quark packaged states is:
\begin{table}[hbt!]
	\centering
	\caption{Example: Color-Singlet States in QCD}
	\begin{tabular}{|l|l|}
		\hline
		State & Color Composition \\ \hline
		Meson & $q\bar{q}$ with $\mathbf{3}\otimes\overline{\mathbf{3}} \rightarrow \mathbf{1}$ \\ \hline
		Baryon & $qqq$ with $\mathbf{3}\otimes\mathbf{3}\otimes\mathbf{3} \rightarrow \mathbf{1}$ \\ \hline
		Tetraquark & $qq\bar{q}\bar{q}$ with proper contraction to yield $\mathbf{1}$ \\ \hline
		Pentaquark & $qqqq\bar{q}$ with proper contraction to yield $\mathbf{1}$ \\ \hline
	\end{tabular}
\end{table}
These classification tables arise naturally by enforcing the packaging principle: each quark is an inseparable block in the $\mathbf{3}$ (or $\overline{\mathbf{3}}$), and only the full contraction yielding the singlet $\mathbf{1}$ is physical.

\subsection{Packaging in Discrete $\mathbb{Z}_N$ Gauge on a Ring}

Consider a ring with $L$ links, each carrying $\mathbf Z_N$ variables $g_\ell$.
The global constraint
\[
U_j=\prod_{\ell=1}^L g_\ell^j,\quad j=0,\dots,N-1
\]
projects onto net flux $j$.
One finds
\[
P_{j}
\;=\;\frac1N\sum_{k=0}^{N-1} \omega^{-jk} U_k,
\quad \omega=e^{2\pi i/N},
\]
and any state with fractional flux on subsegments is killed by these projectors.
This exemplifies the packaging principle in a purely discrete gauge setting.

\subsection{Quantum Information Toy Models}

To illustrate packaged states in a controlled setting, one may construct quantum information toy models on small lattices. For example, consider a lattice model where each site has a local Hilbert space
\[
\mathcal{H}_{\text{site}} \cong \mathbb{C}^d \otimes \mathbb{C}^{d'},
\]
where $\mathbb{C}^d$ represents the local gauge (color) degree of freedom and $\mathbb{C}^{d'}$ represents an external degree (e.g., spin or a global symmetry label). Gauge invariance forces one to project onto subspaces where the local gauge indices combine to yield a singlet.
For instance, in an $\mathrm{SU}(2)$ gauge model, if each site carries a two-dimensional gauge degree of freedom, physical states on links or plaquettes must be in the gauge singlet (or invariant) sector.
Explicit computations using projection operators (as described in earlier sections) then yield the packaged entangled states in the model.

\subsection{Bound-State Wavefunctions}

\paragraph{(1) Bethe-Salpeter Equation.}  
Relativistic bound states are commonly computed using the Bethe-Salpeter equation:
\[
\Gamma(p, P)
\;=\;
\int \!\! \frac{d^4 k}{(2\pi)^4}\;
K(p,k,P)\;
S\Bigl(k+\frac{P}{2}\Bigr)\,\Gamma(k,P)\,S\Bigl(k-\frac{P}{2}\Bigr),
\]
where $\Gamma(p,P)$ is the vertex function, $K(p,k,P)$ is the interaction kernel (e.g., one-gluon exchange in QCD), and $S$ are the propagators.
In a confining theory like QCD, the solution $\Gamma(p,P)$ must be projected onto the color-singlet channel.
One typically expands:
\[
\Gamma_{\alpha\beta}(p,P)
\;=\;
\frac{\delta_{\alpha\beta}}{\sqrt{3}}\,\Gamma_{\mathrm{singlet}}(p,P)
\;+\; \Gamma_{\mathrm{octet}}(p,P),
\]
and only $\Gamma_{\mathrm{singlet}}$ corresponds to a physical meson state.
This explicit projection embodies the packaging principle at the level of bound state wavefunctions.

\paragraph{(2) Potential Models.}  
In potential models (such as the Cornell potential, $V(r) \sim -\alpha/r + \sigma r$), hadron wavefunctions are often factorized as
\[
\Psi_{\alpha\beta}(\mathbf{r}) = \varphi_{\text{color}}(\alpha,\beta)\, \Phi_{\text{spin,space}}(\mathbf{r}).
\]
For mesons, the color part is fixed by the packaged state:
\[
\varphi_{\text{color}}(\alpha,\beta) = \frac{\delta_{\alpha\beta}}{\sqrt{3}},
\]
which ensures that the quark and antiquark form an inseparable color singlet.
The external (spin and spatial) part $\Phi_{\text{spin,space}}(\mathbf{r})$ is then solved for using the potential model.
This separation demonstrates how the packaging principle constrains the internal gauge part, while the external part retains the usual quantum mechanical DOFs.

\section{Discussion}

Our analysis reveals that the inseparability of IQNs is not an incidental feature but a direct consequence of the group-theoretic structure underlying gauge symmetries.
Each single-particle creation operator furnishes an irrep of the symmetry group $G$, and no nontrivial invariant subspace exists within these irreps.
Consequently, any multi-particle state built by tensoring such operators necessarily carries packaged internal charges.

In fact, our results align with recent experimental studies at high energy colliders, which reveal quantum entanglement phenomena in systems such as top-quark pairs, tau leptons, and bottom quarks \cite{TheATLASCollaboration2024,TheCMSCollaboration2024,FabbrichesiFloreaniniGabrielli2023,Afik2024}.
In particular, the observation of spin entanglement in top-quark production \cite{TheATLASCollaboration2024,TheCMSCollaboration2024,HayrapetyanCMSCollaboration2024} and emerging studies of flavor entanglement in mesons and neutrinos \cite{Blasone2009,Go2007} demonstrate that quantum correlations extend well into the relativistic, high energy regime.

Our framework shows how color in QCD and electric charge in QED must each be packaged into inseparable blocks via the irreps of the underlying gauge group.
This provides a natural explanation for why only gauge-invariant (fully packaged) states are physically observed.
Moreover, constructing Bell-like and GHZ-type entangled states under discrete symmetries (e.g., the central $\mathbb Z_2$ of charge conjugation, parity, or time reversal) informs recent spin-entanglement measurements in both fermionic and bosonic sectors \cite{AfikNova2021,AfikNova2022,Barr2022,AguilarSaavedra2023}.

We further demonstrate that IQN inseparability is universal across quantum field theories by extending the packaging principle to higher-form symmetries and embedding it in the full $G\times\mathrm{Poincaré}$ framework.
From this broader vantage, our results not only deepen theoretical insights but also furnish new tools for analyzing experimental data.

This group-theoretic perspective complements our field-theoretic approach in \cite{MaSPI2025}.
Together, these studies offer a more complete description of electric charge confinement in QED and color confinement in QCD.
Finally, our prediction of packaged entanglement under discrete group symmetries can be tested experimentally—either by searching for forbidden partial-fractionation signatures or by quantifying correlation observables that isolate a single irrep block.
We anticipate that future investigations will further clarify how these packaged states may be identified, harnessed, or even exploited in collider settings and beyond.

\paragraph{Concrete Experimental Proposals.}
To make packaged states observable, we suggest:
\begin{itemize}
	\item \emph{Spin-spin correlations} in $t\bar t$ production near threshold, where color exchange is minimal and the irreducible packaging of color is manifest in spin observables.
	
	\item \emph{Flavor-oscillation entanglement} measurements in neutral meson systems (e.g., $B^0$-$\bar B^0$) by correlating decay products in paired detectors.
	
	\item \emph{Photon-jet entanglement} in $e^+e^-\to q\bar q\gamma$, exploiting the fact that the photon carries no color but the quark pair’s color must remain packaged.
\end{itemize}
Each proposal targets an observable sensitive only to the packaged irrep structure, filtering out partial-charge admixtures.

\section{Conclusion}

In this work, we present a detailed group-theoretic study of the symmetry packaging principle for specific groups.
This work both constructs packaged quantum states associated with each group and classifies them by their underlying symmetry.

We first proved an existence theorem for packaged subspaces.
It shows that any finite or compact group $G$ with a nontrivial representation on a quantum system necessarily induces a packaged subspace associated with $G$.
These packaged irrep blocks appear at both the single-particle and multi-particle levels.  
At the multi-particle level, a novel class of packaged entangled states arises in each fixed charge sector.  
Applying the gauge projector filters out non-physical states, leaving only the physical singlets.

We then applied the group-theoretic framework to three broad families of symmetries:
\begin{enumerate}
	\item \textit{Finite groups:}
	We illustrated the symmetry packaging with general finite groups, special $\mathbb{Z}_2$ symmetries, and dihedral groups.
	We showed how charge conjugation $\hat{C}$, fermion parity, parity $\hat{P}$, and time-reversal $\hat{T}$ generate Bell-like packaged entangled states, with superselection rules strictly forbidding cross-sector coherence.
	
	\item \textit{Compact groups:}
	We illustrated the symmetry packaging on compact groups such as $\mathrm{U}(1)$, $\mathrm{SU}(N)$, $\mathrm{SU}(2)$, and $\mathrm{SU}(3)$.
	Our discussion of Peter-Weyl decompositions highlights how color confinement in QCD or isospin structure in nuclear physics naturally emerges from the requirement that physical states reside in a net singlet subspace.
	
	\item\textit{Higher-form and hybrid symmetries.}
	Furthermore, we extended the packaging from point-like particles to extended objects (flux tubes, membranes, or domain walls) in theories with $p$-form symmetries, and to hybrid configurations that combine local gauge, global, and spacetime (Lorentz/Poincaré) symmetries.
	Practical implementations include Bethe-Salpeter bound‐state equations, lattice simulations, and quantum‐information toy models.
\end{enumerate}

These results deepen our understanding of gauge theories by showing that the inseparability of IQNs is not accidental but a direct consequence of the underlying group structure.
The symmetry packaging principle thus provides a unified explanation for phenomena ranging from discrete symmetry packaged entangled states to the inescapable color singlet confinement of quarks.

The insights from symmetry packaging also suggest new applications:
\begin{itemize}
	\item \emph{Protected quantum resources}:
	Packaged entangled states are naturally immune to local noise that respects the symmetry, making them promising for teleportation and error correction.
	
	\item \emph{Fault‐tolerant simulation}:
	Embedding symmetry‐packaging constraints into quantum‐simulation protocols may enhance robustness against gauge‐violating errors.
	
	\item \emph{Collider phenomenology}:
	Correlation measurements, e.g., top‐quark spin entanglement or meson flavor oscillations, can be reinterpreted through the lens of packaged states.
\end{itemize}

We anticipate that these insights will be useful for both theoretical and experimental efforts (from refined lattice studies to high‐energy tests of entanglement).
They may uncover new protected forms of quantum coherence rooted in fundamental symmetries.

\section*{Acknowledgments}

We are grateful to an anonymous referee for her/his valuable feedback that strengthened this work.
We thank Professor Juan Ramón Muñoz de Nova for identifying relevant prior works.

\paragraph{Declaration of generative AI and AI-assisted technologies in the writing process.}
During the preparation of this work the author used ChatGPT o3 in order to polish English language and refine technical syntax.
After using this tool/service, the author reviewed and edited the content as needed and takes full responsibility for the content of the publication.

\end{document}